\pdfoutput=1

\documentclass[reqno,titlepage]{amsart}

\DeclareRobustCommand{\SkipTocEntry}[5]{}

\usepackage{amsaddr,amssymb,array,bbm,bigstrut,enumerate,float,lmodern,tikz}

\usepackage[T1]{fontenc}

\usepackage[margin=2.75cm]{geometry}

\usepackage[linktocpage=true, colorlinks=true, linkcolor=blue, citecolor=red, urlcolor=green]{hyperref}

\newcommand{\mc}[1]{\mathcal{#1}}
\newcommand{\mf}[1]{\mathfrak{#1}}
\newcommand{\mb}[1]{\mathbb{#1}}
\newcommand{\bs}[1]{\boldsymbol{#1}}
\newcommand{\ol}[1]{\,\overline {\!#1\!}\,}

\newcommand{\id}{\mathbbm1}

\newcommand{\git}{/\!\!/}

\newcommand\rA{{\mathrm A}}
\newcommand\rB{{\mathrm B}}
\newcommand\rC{{\mathrm C}}
\newcommand\rD{{\mathrm D}}
\newcommand\rE{{\mathrm E}}
\newcommand\rF{{\mathrm F}}
\newcommand\rG{{\mathrm G}}

\renewcommand{\ker}{\Ker}

\DeclareMathOperator{\Mat}{Mat}
\DeclareMathOperator{\Hom}{Hom}
\DeclareMathOperator{\End}{End}
\DeclareMathOperator{\diag}{diag}

\DeclareMathOperator{\ad}{ad}
\DeclareMathOperator{\im}{Im}

\DeclareMathOperator{\Ker}{Ker}
\DeclareMathOperator{\Span}{Span}

\DeclareMathOperator{\st}{st}
\DeclareMathOperator{\spin}{spin}
\DeclareMathOperator{\Spin}{Spin}

\DeclareMathOperator{\Spec}{Spec}
\DeclareMathOperator{\ext}{\bigwedge}
\DeclareMathOperator{\sym}{\mathsf S}

\theoremstyle{plain}
\newtheorem{theorem}{Theorem}[section]
\newtheorem{lemma}[theorem]{Lemma}
\newtheorem{proposition}[theorem]{Proposition}
\newtheorem{corollary}[theorem]{Corollary}

\theoremstyle{definition}
\newtheorem{definition}[theorem]{Definition}
\newtheorem{example}[theorem]{Example}

\theoremstyle{remark}
\newtheorem{remark}[theorem]{Remark}

\numberwithin{equation}{section}

\definecolor{light}{gray}{.9}

\usetikzlibrary{arrows.meta,calc,decorations,matrix,positioning}

\tikzset{%
    add/.style args={#1 and #2}{
        to path={%
 ($(\tikztostart)!-#1!(\tikztotarget)$)--($(\tikztotarget)!-#2!(\tikztostart)$)%
  \tikztonodes},add/.default={.2 and .2}}
}

\pgfdeclaredecoration{arrows}{draw}{
\state{draw}[width=\pgfdecoratedinputsegmentlength]{%
  \path [every arrow subpath/.try] \pgfextra{%
    \pgfpathmoveto{\pgfpointdecoratedinputsegmentfirst}%
    \pgfpathlineto{\pgfpointdecoratedinputsegmentlast}%
   };
}}

\tikzset{every arrow subpath/.style={->, draw, thick}}

\newcommand{\foresix}[8]{
\begin{tikzpicture}[baseline=(a.base),scale=#1,every node/.style={scale=#2},inner sep=0pt,outer sep=0pt]
\node (a) at (1,0) {$#8$};
\node at (2,0) {$#7$};
\node at (3,0) {$#6$};
\node at (3,-1.5) {$#4$};
\node at (4,0) {$#5$};
\node at (5,0) {$#3$};
\end{tikzpicture}
}

\newcommand{\foreseven}[9]{
\begin{tikzpicture}[baseline=(a.base),scale=#1,every node/.style={scale=#2},inner sep=0pt,outer sep=0pt]
\node (a) at (1,0) {$#9$};
\node at (2,0) {$#8$};
\node at (3,0) {$#7$};
\node at (4,0) {$#6$};
\node at (4,-1.5) {$#4$};
\node at (5,0) {$#5$};
\node at (6,0) {$#3$};
\end{tikzpicture}
}

\newcommand{\foreight}[9]{
\begin{tikzpicture}[baseline=(a.base),scale=.1,every node/.style={scale=#1},inner sep=0pt,outer sep=0pt]
\node (a) at (1,0) {$#9$};
\node at (2,0) {$#8$};
\node at (3,0) {$#7$};
\node at (4,0) {$#6$};
\node at (5,0) {$#5$};
\node at (5,-1.5) {$#3$};
\node at (6,0) {$#4$};
\node at (7,0) {$#2$};
\end{tikzpicture}
}

\newlength{\halogap}

\newlength\triplesep
\newlength\triplelinewidth
\tikzset{triple/.style={
line width=\triplelinewidth,
black,
 preaction={
  preaction={
    draw,
    line width=2\triplesep+3\triplelinewidth,
    black
   },
   draw,
   line width=2\triplesep+\triplelinewidth,white
  }
 }
}

\tikzset{DynkinNode/.style={circle,minimum size=.7em,inner sep=0pt,font=\scriptsize}}

\newcommand{\GDynkin}[1]{
\begin{tikzpicture}[scale=.8,anchor=base,baseline]
\foreach\kthweight[count=\k] in {#1}{
\ifnum\k=1\node[DynkinNode] (1) at (-1,0) {$\scriptscriptstyle\kthweight$};\fi
\ifnum\k=2\node[DynkinNode] (2) at (-2,0) {$\scriptscriptstyle\kthweight$};\fi
}
\draw[-Implies,double distance=1pt] (2) -- (1);
\draw (2) -- (1);
\end{tikzpicture}
}

\newcommand{\FDynkin}[1]{
\begin{tikzpicture}[scale=.6,anchor=base,baseline]
\foreach\kthweight[count=\k] in {#1}{
\ifnum\k=1\node[DynkinNode] (1) at (-1,0) {$\scriptscriptstyle\kthweight$};\fi
\ifnum\k=2\node[DynkinNode] (4) at (-4,0) {$\scriptscriptstyle\kthweight$};\fi
\ifnum\k=3\node[DynkinNode] (2) at (-2,0) {$\scriptscriptstyle\kthweight$};\fi
\ifnum\k=4\node[DynkinNode] (3) at (-3,0) {$\scriptscriptstyle\kthweight$};\fi
}
\draw[thick] (1) -- (2);
\draw[-Implies,double,thick] (3) -- (2);
\draw[thick] (3) -- (4);
\end{tikzpicture}
}

\newcommand{\EDynkin}[2]{
\begin{tikzpicture}[scale=.4,anchor=base,baseline]
\foreach\kthweight[count=\k] in {#1}{
\pgfmathtruncatemacro{\prevnode}{\k-1}
\ifnum\k=1\node[DynkinNode] (\k) at (0,-1) {$\scriptscriptstyle\kthweight$};\fi % exchanged 1 and 2
\ifnum\k>1\node at (4-\k,0) {$\scriptscriptstyle\kthweight$};
\node[DynkinNode] (\k) at (4-\k,0) {$\scriptscriptstyle\kthweight$};
\ifnum\k>2\draw[thick] (\k) -- (\prevnode);\fi
\ifnum\k=4\draw[thick] (\k) -- (1);\fi
\fi
}
\foreach\Node[count=\i] in {#2}
{
\xdef\imax{\i}
\coordinate (AuxNode-\i) at ($(\Node)$);
}
\ifnum\imax=3%
\draw[red,rounded corners] ($(AuxNode-1)+(-\halogap,\halogap)$) -|  ($(AuxNode-2)+(-\halogap,\halogap)$)
    -- ($(AuxNode-2)+(\halogap,\halogap)$) |-
    ($(AuxNode-3)+(\halogap,\halogap)$)
    --($(AuxNode-3)+(\halogap,-\halogap)$) -- ($(AuxNode-1)+(-\halogap,-\halogap)$) -- cycle;
\else\ifnum\imax=2%
\draw[red,rounded corners] ($(AuxNode-1)+(-\halogap,\halogap)$) --  ($(AuxNode-2)+(\halogap,\halogap)$)
    -- ($(AuxNode-2)+(\halogap,-\halogap)$) -- ($(AuxNode-1)+(-\halogap,-\halogap)$)--
    cycle;
\fi
\fi
\end{tikzpicture}
}

\tikzset{every arrow subpath/.style={<-, draw, thick}}

\newcommand\bibit[1]{\bibitem[#1]{#1}}

\let\half1

\newcommand\qcss{{\mathbf S}}

\newcommand\kform[2]{\left(#1|#2\right)}

\newcommand\eps\varepsilon
\renewcommand\le\leqslant
\renewcommand\ge\geqslant

\newcommand\fs{f^{\mathrm s}}
\newcommand\fn{f^{\mathrm n}}

\newcommand\partition\bs
\newcommand\ph\varphi

\newcommand\ot\leftarrow
\newcommand\transpose{^\intercal}

\begin{document}

\setcounter{tocdepth}2

\title{Integrable triples in semisimple Lie algebras}

\author{Alberto De Sole}
\address{Dipartimento di Matematica, Sapienza Universit\`a di Roma,
P.le Aldo Moro 2, 00185 Rome, Italy}
\email{desole@mat.uniroma1.it}
\thanks{{\it Web addresses:} \texttt{http://www1.mat.uniroma1.it/\textasciitilde desole}}

\author{Mamuka Jibladze}
\address{Razmadze Mathematical Institute,
TSU, Tbilisi 0186, Georgia
}
\email{jib@rmi.ge}
\author{Victor G. Kac}
\address{Department of Mathematics, MIT,
77 Massachusetts Avenue, Cambridge, MA 02139, USA}
\email{kac@math.mit.edu}
\author{Daniele Valeri}
\address{School of Mathematics and Statistics, University of Glasgow, G12 8QQ Glasgow, UK}
\email{daniele.valeri@glasgow.ac.uk}
\begin{abstract}
We classify all integrable triples in simple Lie algebras, up to equivalence. The importance of this problem stems from the fact that for each such equivalence class one can construct the corresponding integrable hierarchy of bi-Hamiltonian PDE. The simplest integrable triple $(f,0,e)$ in $\mf{sl}_2$ corresponds to the KdV hierarchy, and the triple $(f,0,e_\theta)$, where $f$ is the sum of negative simple root vectors and $e_\theta$ is the highest root vector of a simple Lie algebra, corresponds to the Drinfeld-Sokolov hierarchy.

\

\begin{center}
\ \hfill\it\Large Dedicated to the memory of Boris Anatolievich Dubrovin
\end{center}
\end{abstract}

\maketitle

\tableofcontents

\section{Introduction}

The present paper is a continuation of our paper \cite{DSJKV20}, where we studied integrability of $W$-algebras.
Namely, we showed that, for the classical affine $W$-algebra $\mc W(\mf g,f)$ attached to a simple Lie algebra $\mf g$
and its non-zero nilpotent element $f$, the Lie algebra $\mc W(\mf g,f)/\partial\mc W(\mf g,f)$ contains an infinite-dimensional abelian subalgebra,
except, possibly, for the following $f$ (in the notation of \cite{CMG93}):
$4\rA_1$, $2\rA_2+2\rA_1$, $2\rA_3$, $\rA_4+\rA_3$ and $\rA_7$ in $\rE_8$; $\widetilde\rA_2+\rA_1$ in $\rF_4$; $\widetilde\rA_1$ in $\rG_2$.
Consequently, for all these $W$-algebras (with the seven exceptions above) one constructs an integrable hierarchy of bi-Hamiltonian PDE,
the simplest being the KdV hierarchy, constructed for $\mc W(\mf{sl}_2,f)$.

The proof of this theorem consists of the following ingredients.
First, it is the Drinfeld-Sokolov \cite{DS85} (abbreviated DS) method of constructing integrals of motion in involution
in the case when $f$ is a principal nilpotent element of $\mf g$,
which has been extended to the case when $f$ is a nilpotent element of semisimple type for $\mf g$ in \cite{DSKV13}.

Second, we showed in \cite{DSJKV20} that a simple modification of the DS method works also
for $f$ of non-nilpotent type, which covers all nilpotent elements of even depth.

Third, in the same paper we showed that for $f$ of nilpotent type this modification works as well,
except for the seven nilpotent conjugacy classes mentioned above.

Let us now introduce the relevant definitions.
By the Jacobson-Morozov theorem, any non-zero nilpotent $f$ of a simple Lie algebra $\mf g$
can be included in an $\mf{sl}_2$-triple $\mf s=\{e,f,h\}$.
This produces a $\mb Z$-grading of $\mf g$ by eigenspaces of $\ad h$:
\begin{equation}\label{eq:grading}
\mf g=\bigoplus_{j\in\mb Z}\mf g_j,
\end{equation}
which, up to conjugation, is independent of the choice of the $\mf{sl}_2$-triple.
The maximal $j>0$ for which $\mf g_j\ne0$ is denoted by $d$, and is called the \emph{depth} of $f$.

An element of the form $f+E$, where $E$ is a non-zero element of $\mf g_d$, is called a \emph{cyclic element} associated to $f$.
The key tool in the DS method is Kostant's theorem that for the principal nilpotent $f$ all the associated cyclic elements are semisimple.
The main result of \cite{DSKV13} extends the DS method in the framework of Poisson vertex algebras,
to arbitrary $f$, which admits an associated semisimple cyclic element.
Such $f$ is called a nilpotent element of \emph{semisimple type} \cite{EKV13}.

If all cyclic elements, associated to $f$, are nilpotent, then $f$ is called a nilpotent element of \emph{nilpotent type}.
It is proved in \cite[Theorem 1.1]{EKV13} that $f$ is of nilpotent type if and only if its depth is an odd integer.
As has been mentioned above, the DS method works for $f$ of non-nilpotent type.

In the case of $f$ of nilpotent type, i.~e. of odd depth $d$,
we introduced in \cite{DSJKV20} the notion of a \emph{quasi-cyclic} element $f+E$, associated to $f$.
It requires $E$ to be a non-zero element of $\mf g_{d-1}$ with the additional requirement that the centralizer of $E$ in $\mf g_1$
is coisotropic with respect to the symplectic form $\omega$ on $\mf g_1$, defined by
\begin{equation}\label{eq:skewform}
\omega(a,b)=(f\mid[a,b]),\quad a,b\in\mf g_1.
\end{equation}
Hereafter $(\cdot\mid\cdot)$ is a fixed non-degenerate symmetric invariant bilinear form on $\mf g$.

We show in \cite{DSJKV20} that the DS method works for $f$ of nilpotent type,
provided that there exists a non-nilpotent quasi-cyclic element, associated to $f$.
Moreover, such an element exists if and only if $f$ is not one of the seven nilpotent elements of odd depth mentioned above.
In the present paper we study quasi-cyclic elements, associated to arbitrary non-zero nilpotent elements in simple Lie algebras.

The main tool of the proof of the Integrability Theorem of \cite{DSJKV20}
is the notion of an \emph{integrable triple} associated to the nilpotent element $f$.
It is a triple of elements $(f_1,f_2,E)$, where $f_1,f_2\in\mf g_{-2}$, $E\in\mf g_j$ with $j\ge1$, such that
\begin{enumerate}[({I}1)]
\item\label{def:intrip1} $f=f_1+f_2$ and $[f_1,f_2]=0$;
\item\label{def:intrip2} $[E,\mf g_{\ge2}]=0$ and the centralizer $\mf l^\perp$ of $E$ in $\mf g_1$ is coisotropic with respect to the symplectic form $\omega$, defined by \eqref{eq:skewform};
\item\label{def:intrip3} $f_1+E$ is semisimple and $[f_2,E]=0$.
\end{enumerate}

The coisotropy condition (I\ref{def:intrip2}) is important for the construction of the corresponding to $f$ classical affine $W$-algebra, and condition (I\ref{def:intrip3}) is used for the construction of integrals of motion in involution.

Note that for an integrable triple $(f_1,f_2,E)$ the decomposition $f+E=(f_1+E)+f_2$ is the Jordan decomposition,
and that $E$ commutes with the subalgebra $\mf n:=\mf l^\perp+\mf g_{\ge2}$.

It turns out that for $j<d-1$ there are no integrable triples (Proposition \ref{prop:nonebelow}).
Obviously, for an integrable triple $(f_1,f_2,E)$, associated to $f$,
the element $f+E=f_1+f_2+E$ is a cyclic (resp. quasi-cyclic) element, associated to $f$, if $j=d$ (resp. $j=d-1$).
Such a cyclic (resp. quasi-cyclic) element is called \emph{integrable}.
In this case the element $E\in\mf g_j$ is called \emph{integrable} for $f$.
Conversely, given a cyclic or quasi-cyclic element $f+E$,
we can obtain an integrable triple associated to $f$, by taking the Jordan decomposition of $f+E$, provided that $E$ is integrable for $f$.
One of the main problems is when such an integrable triple exists.

In \cite{DSJKV20} we have established existence of integrable triples, associated to any nilpotent element $f$,
except for the seven cases mentioned above. In the present paper we give a complete solution to this problem.
Namely, we find all integrable cyclic and quasi-cyclic elements $f+E$, up to equivalence, for each non-zero nilpotent element $f$ in a simple Lie algebra $\mf g$. Two integrable cyclic or quasi-cyclic elements $f+E$ and $f+E'$ are called \emph{equivalent} if $E'$ is proportional to an element from the orbit $Z(\mf s)(E)$,
where $Z(\mf s)$ stands for the centralizer of the $\mf{sl}_2$-triple $\mf s$ in the adjoint group $G$ of $\mf g$.
The importance of this problem stems from the fact that, as established in \cite{DSJKV20},
for each equivalence class of integrable triples one can construct the associated integrable hierarchy of bi-Hamiltonian PDE.

A partial solution to this problem was given in \cite{EJK20}, where all semisimple cyclic elements have been classified
(the corresponding integrable triple has the form $(f,0,E)$).
The key observation there (checked by case-wise verification) was that the linear reductive group $Z(\mf s)|\mf g_d$ is \emph{polar}.
Polar linear groups were introduced in \cite{DK85} as reductive linear groups, having properties, similar to the adjoint group
(see Section \ref{sec:review} for the precise definition).
This observation allows one to reduce considerations to the case $E\in C$, a Cartan subspace of $\mf g_d$,
since it was proved in \cite{EKV13}, Proposition 2.2(a), that for semisimple cyclic element $f+E$, the orbit $Z(\mf s)(E)$ must be closed.

In the present paper we find that, remarkably, the linear reductive group $Z(\mf s)|\mf g_{d-1}$ is polar as well!
This is Theorem \ref{thm:allpolar}. Unfortunately we still need a case-wise analysis in its proof. However
Remark \ref{rem:ssproof} explains why $Z(\mf s)|\mf g_j$ should be polar for $j=d$ and $d-1$.
Note that $Z(\mf s)|\mf g_j$ is not polar in general for $1\le j<d-1$.

Our first main result is Theorem \ref{thm:nonniliffnonnil}, which states that a cyclic element $f+E$ is not nilpotent
if and only if the Zariski closure of the orbit $Z(\mf s)(E)$ does not contain $0$. Our second main result is Theorem \ref{thm:intiffss}, which states that a cyclic element $f+E$ is integrable if and only if the orbit $Z(\mf s)(E)$ is closed.
Such $E$ are classified by a Cartan subspace of $Z(\mf s)|\mf g_d$.
A similar result holds for a quasi-cyclic element $f+E$, where $E\in\mf g_{d-1}$ and $d$ is odd;
in the case of even $d$ we were able to prove only the ``only if'' part of this result (see Theorem \ref{thm:quasintiffss}).

These results allow us to classify completely all integrable cyclic and quasi-cyclic elements in all simple Lie algebras.
As has been mentioned above, this is equivalent to the classification of integrable triples.
For exceptional $\mf g$ this is done using the GAP package SLA by W. de Graaf \cite{SLA},
and the results are listed in Table \ref{tab:odd} and Tables \ref{tab:withinv}, \ref{tab:noinv},
which represent the cases of $d$ odd and even respectively.
For classical $\mf g$ this is done by explicit calculations in the standard representation of $\mf{sl}_N$, $\mf{sp}_N$ and $\mf{so}_N$. In particular, we have the following results on existence of integrable cyclic and quasi-cyclic elements in simple Lie algebras $\mf g$, associated to non-zero nilpotent elements $f\in\mf g$.

\begin{theorem}
\begin{enumerate}[(a)]
\item An integrable cyclic element, associated to $f$, exists if and only if its depth $d$ is even.
\item An integrable quasi-cyclic element, associated to $f$, exists in precisely the following cases:
\begin{enumerate}[(i)]
\item $f$ has odd depth and is of the following type in exceptional $\mf g$:
\begin{quote}
$3\rA_1$ in $\rE_6$ and $\rE_8$; $3\rA_1'$ in $\rE_7$; $2\rA_2+\rA_1$ in $\rE_6$, $\rE_7$, $\rE_8$; $4\rA_1$ in $\rE_7$; $\rA_1+\tilde\rA_1$ in $\rF_4$,
\end{quote}
\item $f$ has even depth and is of the following type in exceptional $\mf g$:
\begin{quote}
$\rA_1$ in all exceptional $\mf g$; $2\rA_1$ in $\rE_6$; $
\rA_2+\rA_1$ in $\rE_6$, $\rE_7$, $\rE_8$; $\rA_4+\rA_1$ in $\rE_6$, $\rE_7$, $\rE_8$; $\rA_4+2\rA_1$ in $\rE_8$; $\rD_7(a_2)$ in $\rE_8$,
\end{quote}
\item all $f$ of odd depth in classical $\mf g$, which happens only for $\mf g=\mf{so}_N$, and the partition, corresponding to $f$, has odd largest part $p_1$ of multiplicity $1$ and the next part $p_2=p_1-1$,
\item the following $f$ of even depth, corresponding to the partition $(p_1^{(r_1)},p_2^{(r_2)},...)$ in classical $\mf g$:
\begin{quote}
$\mf g=\mf{sl}_N$: $r_1\le r_2$;\qquad
 $\mf g=\mf{so}_N$: $p_1$ is even, $r_1=2$, $p_2=p_1-1$, $r_2\ge2$.
\end{quote}
\end{enumerate}
\end{enumerate}
\end{theorem}
Many statements in the paper are established by a case-wise verification. It would be very interesting to find more conceptual proofs. Here are some of these statements.
\begin{enumerate}
    \item The linear groups $Z(\mf s)|\mf g_j$ for $j=d$ and $d-1$ are polar.
    \item There exists at most one, up to equivalence, quasi-cyclic element for each nilpotent element $f$ of even depth.
    \item There are no integrable triples $(f_1,f_2,E)$ with $E\in\mf g_j$, where $j<d-1$.
\end{enumerate}

Throughout the paper the base field $\mb F$ is an algebraically closed field of characteristic 0.
Though the theory of polar linear groups was developed in \cite{DK85} over $\mb C$, all results hold over $\mb F$ by the Lefschetz principle.

\subsubsection*{Acknowledgments}
The first author was partially supported by the national PRIN 2017 grant ``Moduli and Lie Theory'' number 2017YRA3LK\_001, and by the University 2019 grant  ``Vertex algebras and integrable Hamiltonian systems''.
The second author was partially supported by the grant FR-18-10849 of Shota Rustaveli
National Science Foundation of Georgia.
The third author was partially supported by the Bert and Ann Kostant fund, and by the Simons collaboration grant.

We would like to thank J.~Dadok, A.~Elashvili, P.~Etingof, and D.~Panyushev for very useful discussions. Extensive use of
the computer algebra system GAP, and the package SLA by Willem de Graaf in particular, is gratefully acknowledged.

\section{Polar linear groups and gradings for nilpotent elements}
\subsection{Review of polar linear groups}\label{sec:review}
Let $G$ be a reductive algebraic group, acting linearly and faithfully on a finite-dimensional vector space $V$, which will be denoted by $G|V$.
It is well known that the subalgebra $\mb F[V]^G$ of $G$-invariant polynomials is finitely generated,
hence the inclusion $\mb F[V]^G\to\mb F[V]$ induces the map of the corresponding affine varieties
\begin{equation}
\pi:V\to V\git G,
\end{equation}
where $V\git G:=\Spec\mb F[V]^G$. It is well known that the map $\pi$ is surjective
and that each of the fibers of $\pi$ contains a unique closed $G$-orbit (the orbit of minimal dimension in the fiber).
The fiber over $\pi(0)$, called the zero fiber, consists of elements $v\in V$, such that the Zariski closure of the orbit $G(v)$ of $v$ contains 0.
Such elements are called \emph{nilpotent} elements of the linear group $G|V$.
Elements $v\in V$, such that $G(v)$ is closed, are called \emph{semisimple} elements of $G|V$.
This terminology is motivated by the well-known fact that for the adjoint linear group $G|\mf g$ an element is semisimple (resp. nilpotent)
if and only if it is a semisimple (resp. nilpotent) element of the Lie algebra $\mf g$.

An efficient way of constructing semisimple elements for a reductive linear group $G$ is given by the following
\begin{proposition}[\cite{DK85}]\label{prop:dkss}
Let $P$ be a set of weights of the $\mf g$-module $V$ from its irreducible components with the following properties:
\begin{enumerate}[(i)]
\item $\lambda_i-\lambda_j$ is not a root of $\mf g$ if $i\ne j$ and $\lambda_i,\lambda_j\in P$ are weights of the same irreducible component of $V$;
\item zero is an interior point of the convex hull of $P$.
\end{enumerate}
Let $v_{\lambda_i}$, $\lambda_i\in P$, be linearly independent weight vectors from irreducible components of $V$.
Then the vector $\sum_iv_{\lambda_i}$ is semisimple.
\end{proposition}
\begin{proof}
The proposition is slightly stronger than Proposition 1.2 from \cite{DK85}, but its proof is the same.
\end{proof}
\begin{corollary}[Kostant theorem]\label{cor:Kostant}
Any vector from the zero weight space of $V$ is semisimple.
\end{corollary}
Now we turn to the discussion of the especially nice class of reductive linear groups $G|V$,
called polar linear groups, which were introduced in \cite{DK85}.
Let $v\in V$ be a semisimple element, and let
\begin{equation}
C_v=\{x\in V\mid\mf g(x)\subseteq\mf g(v)\}.
\end{equation}
Then \cite{DK85}
\begin{equation}\label{eq:lepolar}
\dim C_v\le\dim V\git G.
\end{equation}
The reductive linear group $G|V$ is called \emph{polar} if there exists a semisimple $v\in V$, such that
\begin{equation}\label{eq:polar}
\dim C_v=\dim V\git G,
\end{equation}
and in this case $C_v$ is called a \emph{Cartan subspace} of $V$.

Note that the Cartan subalgebra $\mf h$ of the Lie algebra $\mf g$ of $G$ is a Cartan subspace for the adjoint linear group $G|\mf g$,
since for a regular $v\in\mf h$ its $G$-orbit is closed and \eqref{eq:polar} holds because
$\dim C_v=\dim V\git G=\operatorname{rank}\mf g$.

More generally, we have the following
\begin{proposition}[\cite{EJK20}]
Let $G|V$ be a reductive linear group and let $C$ be its zero weight space. Then
\begin{equation}\label{eq:ledim}
\dim C\le\dim V\git G,
\end{equation}
and in the case of equality the linear group $G|V$ is polar.
\end{proposition}
\begin{proof}
By Corollary \ref{cor:Kostant}, any element $v\in C$ is semisimple. Let $\mf g^\circ=\{g\in\mf g\mid g(C)=0\}$.
Then there exists $v\in C$, such that $\{g\in\mf g\mid g(v)=0\}=\mf g^\circ$, and hence we have $\mf g(v)=\mf g(C)$.
It follows that $\{x\in C\mid\mf g(x)\subseteq\mf g(v)\}=C$, hence, if equality holds in \eqref{eq:ledim}, $C$ is a Cartan subspace.
By the same argument, we have inequality \eqref{eq:ledim}, due to the inequality \eqref{eq:lepolar}.
\end{proof}
\begin{remark}
By the definition, $G|V$ is polar if $\dim V\git G=1$, or $\dim V\git G=0$.
Note also that the direct sum $(G_1\times G_2)|(V_1\oplus V_2)$ of polar linear groups $G_i|V_i$, $i=1,2$, is polar.
\end{remark}
The following theorem shows that a Cartan subspace of a polar linear group $G|V$ has the same basic properties as a Cartan subalgebra of $\mf g$.
\begin{theorem}\label{thm:Cartanstuff}
Let $G|V$ be a polar linear group and let $C\subset V$ be a Cartan subspace.
\begin{enumerate}[(a)]
\item\label{smallWeyl}$($\cite[Theorem 2.8]{DK85}$)$.
Let
\[
N=\{g\in G\mid g(C)=C\},\quad Z=\{g\in N\mid\text{$g(c)=c$ for all $c\in C$}\},\quad W=N/Z.
\]
Then $G(C)$ consists of all semisimple elements of $V$, and for any semisimple $v\in V$, the orbit $G(v)$ intersects $C$ by a (non-empty) orbit of $W$.
\item$($\cite[Theorem 2.3]{DK85}$)$.
If $C'\subset V$ is another Cartan subspace, then $g(C')=C$ for some $g\in G$.
\item$($\cite[Theorem 2.10]{DK85}$)$. If $G$ is connected, then $V\git G$ is an affine space (of dimension $\dim C$).
\end{enumerate}
\end{theorem}

Recall that a linear reductive group $G|V$ is called \emph{stable} if $V$ has a non-empty Zariski open subset, consisting of closed $G$-orbits.
The following proposition is very useful in verifying that a stable linear reductive group $G|V$ is polar.
\begin{proposition}[\cite{DKII}]
Let $G|V$ be a stable reductive linear group. Let $C\subset V$ be a subspace, such that
\begin{equation}\label{eq:C+gC}
V=\mf g(C)\oplus C
\end{equation}
and
\begin{equation}\label{eq:fulldimC}
\dim C=\dim V\git G.
\end{equation}
Then $G|V$ is polar and $C$ is a Cartan subspace.
\end{proposition}
\begin{proof}
Due to \eqref{eq:C+gC}, $G(C)$ contains a Zariski open subset, and, due to stability,
it contains a Zariski open subset $\Omega$ consisting of closed $G$-orbits of maximal dimension.
Let $C^\circ=C\cap\Omega$, and let $v\in C^\circ$.
Then the tangent space $T_v$ to $G(v)$ at $v$ lies in $\mf g(C)$, and, due to \eqref{eq:C+gC} and \eqref{eq:fulldimC}, actually $T_v=\mf g(C)$.
Since this holds for all $v\in C^\circ$, we obtain that $C_v=C$. Hence, by \eqref{eq:fulldimC}, $G|V$ is polar and $C$ is a Cartan subspace.
\end{proof}
\begin{remark}
(a) It follows from Theorem \ref{thm:Cartanstuff} \eqref{smallWeyl} that conditions \eqref{eq:C+gC} and \eqref{eq:fulldimC} are also necessary
for a reductive linear group to be stable polar.

(b) By Popov's stability theorem \cite{Po70}, for a semisimple algebraic group $G$ stability of a linear group $G|V$ is equivalent to the condition
that there exists $v\in V$ whose stabilizer $G_v$ is reductive and
\begin{equation}\label{eq:dimdiff}
\dim V-\dim G(v)=\dim V\git G.
\end{equation}
(In general, LHS\eqref{eq:dimdiff} $\ge$ RHS\eqref{eq:dimdiff}.)

(c) The following example, provided by J.~Dadok, shows that condition \eqref{eq:C+gC} alone is not sufficient.
Let $\mf g$ be a simple Lie algebra and $G|(\mf g\oplus\mf g)$ the action of the adjoint group on the direct sum of 2 copies of $\mf g$.
Let $C$ be the sum of Cartan subalgebras in each copy.
Then \eqref{eq:C+gC} holds, but this linear group is not polar.
\end{remark}

\subsection{Nilpotent elements and polar linear groups}
Let $f$ be a non-zero element of a simple Lie algebra, included in an $\mf{sl}_2$-triple $\mf s$,
and let \eqref{eq:grading} be the corresponding $\mb Z$-grading of $\mf g$. Let $d$ be the depth of this grading.
Then each $\mf g_j$, $|j|\le d$, is a $\mf g_0$-module.
Since $\mf z(\mf s)$, the centralizer of $\mf s$ in $\mf g$, is a subalgebra of $\mf g_0$, each $\mf g_j$ is a $\mf z(\mf s)$-module.
Since $\mf z(\mf s)$ is a reductive subalgebra of $\mf g$, we obtain a reductive linear Lie algebra $\mf z(\mf s)|\mf g_j$
by taking the image of $\mf z(\mf s)$ in $\End\mf g_j$.
Recall that $Z(\mf s)$ is the centralizer of $\mf s$ in the adjoint group $G$ of $\mf g$.
\begin{theorem}\label{thm:allpolar}
All linear groups $Z(\mf s)|\mf g_j$ for $j=d$ or $d-1$ are polar.
\end{theorem}
\begin{proof}
We do not know a proof of this remarkable fact without the case-wise verification.
For $j=d$ it was pointed out in \cite{EJK20} that this fact follows from Tables there for $d$ even and Remark 2.1 for $d$ odd.
For $j=d-1$, when $d$ is odd, this follows from Table 1 of \cite{DSJKV20}, where all these groups are listed,
and from Table 1 of \cite{DK85}, where all polar linear groups of simple Lie algebras are listed.
Finally, for $j=d-1$, where $d$ is even, this fact follows from Tables \ref{tab:zsdm1}, \ref{tab:withinv} and \ref{tab:noinv} of this paper.
Indeed, looking at Table 1 of \cite{DK85},  one can see that, apart from theta groups, all examples that occur in these tables are polar,
except, possibly, for the linear reductive groups $\Spin_7\otimes\st(SO_2)$ and $\Spin_9\otimes\st(SL_2)$.  (Hereafter Spin (resp.  st) denotes the spinor (resp.  standard) representation.)
For both of them $\dim V\git G = 2$.  It is shown in the examples below that both are polar.
\end{proof}
\begin{remark}\label{rem:ssproof}
We tried to prove polarity of $Z(\mf s)|\mf g_j$ for $j=d$ or $d-1$ along the following lines.

Let $\eps=e^{\frac{2\pi i}{j+2}}$, where $j=d$ or $d-1$, and consider the automorphism $\sigma$ of $\mf g$,
defined by $\sigma|_{\mf g_s}=\eps^sI$, $s\in\mb Z$.
Then the fixed point set of $\sigma$ is $\mf g_0$ and the $\eps^{-2}$-eigenspace of $\sigma$ is $V:=\mf g_{-2}+\mf g_j$.
Let $G^0$ be the algebraic subgroup of the adjoint group $G$, corresponding to the subalgebra $\mf g_0$ of $\mf g$.
Then we get a theta group $G^0|V$. By \cite{DK85}, \cite[Theorem 6(c)]{EJK20} this representation is polar.

Let $SG^0=G^0\cap SL(\mf g_{-2})$, and consider the linear group $SG^0|\mf g_{-2}$.
By Proposition 1.1 from \cite{K80} there exists an $SG^0$-invariant polynomial $P$ on $\mf g_{-2}$,
such that $P(a)\ne0$ for $a\ne0$ if and only if the orbit $SG^0(a)$ is closed;
moreover such an orbit contains a non-zero multiple of $f$.
Since $\mf z(\mf s)=\mf g_0\cap\mf g^f$ and the group $Z(\mf s)|\mf g_{-2}$ is orthogonal,
we obtain that $Z(\mf s)\subset SG^0$ and $Z(\mf s)=(SG_0)^f$.
Let $\mf{sg}_0=\mf g_0\cap\mf{sl}(\mf g_{-2})$ and consider the slice representation $Z(\mf s)|N_f$, where $N_f=V/[\mf{sg}_0,f]$.
Then by the Luna slice theorem \cite{L72}
\[
\dim(V\git SG^0)=\dim(\mf g_j\git Z(\mf s))+1.
\]
Therefore, since $\dim(V\git G^0)\ge\dim(V\git SG_0)+1$, we obtain that
\begin{equation}\label{eq:dimGgeZ}
\dim(V\git G^0)\ge\dim(\mf g_j\git Z(\mf s)).
\end{equation}
Let $v\in V$ be such that $G^0(v)$ is a closed orbit of maximal dimension.
Since $G^0|V$ is polar,
\[
C_v:=\{v'\in V\mid\mf g_0(v')\subseteq\mf g_0(v)\}
\]
is a Cartan subspace. Let $\pi:V\to\mf g_j$ be the projection.
Since no nonzero $G^0$-orbit in $\mf g_{-2}$ is closed, we conclude that $\pi$ is injective on $C_v$.
Moreover we think that $\pi(C_v)\subseteq C_{\pi(v)}$ but we were unable to prove this.

But if this is true, then
\begin{equation}\label{eq:dimvlepi}
\dim C_v\le\dim C_{\pi(v)}.
\end{equation}
Comparing \eqref{eq:dimGgeZ} and \eqref{eq:dimvlepi}, we conclude that $\dim C_{\pi(v)}\ge\dim(\mf g_j\git Z(\mf s))$.
But by \eqref{eq:lepolar}, the LHS cannot be greater than the RHS, hence $\dim C_{\pi(v)}=\dim(\mf g_j\git Z(\mf s))$
and the linear group $Z(\mf s)|\mf g_j$ is polar.
\end{remark}
\begin{example}\label{ex:spin7stso2}
$G|V=\Spin_7\otimes\st(SO_2)$. This is a direct sum of two irreducible modules $V=V^+\oplus V^-$, which are 8-dimensional spin modules for $SO_7$;
$SO_2$ acts on $V^\pm$ by multiplication by $t^{\pm1}$, $t\in\mb F^\times$.
The set of weights of the $\mf g$-module $V$ is
\[
\frac12(\pm\eps_1\pm\eps_2\pm\eps_3)\pm\delta,
\]
the highest weights of $V^\pm$ being $\Lambda^\pm=\frac12(\eps_1+\eps_2+\eps_3)\pm\delta$, and lowest weights being $-\Lambda^\pm$.
Let $v^\pm$ be the corresponding highest weight vectors and $v_\pm$ lowest weight vectors. Let
\[
v_1=v^++v_-,\quad v_2=v^-+v_+.
\]
Then, by Proposition \ref{prop:dkss}, all vectors from $C:=\operatorname{span}\{v_1,v_2\}$ are semisimple.
It is easy to see that $\mf g(C)$ is the tangent space to the orbit $G(v)$ at a generic point $v\in C$.
It follows that $C$ is a Cartan subspace.
\end{example}
\begin{example}\label{ex:spin9stsl2}
$G|V=\Spin_9\otimes\st(SL_2)$. It is an irreducible 32-dimensional $\mf g$-module with the set of weights
\[
\frac12(\pm\eps_1\pm\eps_2\pm\eps_3\pm\eps_4)\pm\delta,
\]
the highest weight being $\Lambda=\frac12(\eps_1+\eps_2+\eps_3+\eps_4)+\delta$.
We let
\[
v_1=v_\Lambda+v_{-\Lambda},\quad v_2=e_{-\eps_1}e_{-\eps_2-\eps_3}v_\Lambda+e_{\eps_1}e_{\eps_2+\eps_3}v_{-\Lambda}.
\]
As in Example \ref{ex:spin7stso2}, it is easy to check that $C=\operatorname{span}\{v_1,v_2\}$ is a Cartan subspace.
\end{example}
\begin{theorem}\label{thm:Zgj}\
\begin{enumerate}[(a)]
\item\label{thm:Zgjforms} All linear groups $Z(\mf s)|\mf g_j$ for $j$ even (resp. odd) preserve a symmetric (resp. skew-symmetric) non-degenerate invariant bilinear form.
\item\label{thm:Zgjstable} All linear groups $Z(\mf s)|\mf g_j$ for $j=d$ or $d-1$ are stable,
except for all cases when $j=d$ is odd, and also when $j=d-1$ is odd and $Z(\mf s)|\mf g_{d-1}$ either is isomorphic to $\st(\mf{so}_n)\otimes\st(\mf{sp}_m)$
with $n$ odd, $n<m$, or has finitely many orbits.
\end{enumerate}
\end{theorem}
\begin{proof}
\eqref{thm:Zgjforms} can be found in \cite{Pa99} or \cite{EJK20}. By \eqref{thm:Zgjforms}, the stability follows when $j$ is even from \cite{L72}.
When $d$ is even and $j=d-1$, the stability is established by a case-wise verification,
by checking that $[\mf z(\mf s),C]=\mf g_j$ where $C$ is a Cartan subspace (which exists by Theorem \ref{thm:allpolar}).
\end{proof}
\begin{remark}
It follows from Tables \ref{tab:zsdm1}, \ref{tab:withinv}, \ref{tab:noinv} that for an odd nilpotent the Dynkin characteristic contains 2 only if $Z(\mf s)|\mf g_{d-1}$ has finitely many orbits.
Moreover, the Dynkin characteristic contains 2 if $\mf g_j=0$ for some $1\le j\le d-1$; this follows from the observation that when the Dynkin characteristic contains only $0$ or $1$,
the $\mf g_0$-module $\mf g_j$ is generated by $\mf g_1^j$.
\end{remark}
\begin{remark}
For all nilpotent elements of odd depth in Table \ref{tab:odd} the linear group $Z(\mf s)|\mf g_1$ is not polar, except for the last one.
\end{remark}

\section{Generalized cyclic elements and integrable triples}

In this section $\mf g$ is a semisimple Lie algebra. The following notion generalizes both notions of cyclic and quasi-cyclic elements, introduced in the Introduction, for $j=d$ and $j=d-1$ respectively.

\begin{definition}\label{def:gencyc}
A \emph{generalized cyclic element}, attached to a nilpotent element $f$ of $\mf g$ is an element of the form $f+E$,
where $E\in\mf g_j$ with $j\ge1$, $E\ne0$, $[E,g_{\ge2}]=0$
and the centralizer of $E$ in $\mf g_1$ is coisotropic with respect to the symplectic form $\omega$,
defined by \eqref{eq:skewform}. Two generalized cyclic elements $f+E$ and $f+E'$ are called \emph{equivalent} if $E'$ is proportional to an element from the orbit $Z(\mf s)(E)$.
\end{definition}

Recall the notion of an integrable triple, associated to $f$,
which is defined by properties (I\ref{def:intrip1}) -- (I\ref{def:intrip3}) in the Introduction.

\begin{proposition}\
\begin{enumerate}[(a)]
\item\label{prop:intripgcyc} If $(f_1,f_2,E)$ is an integrable triple, associated to $f$, then the element
\begin{equation}\label{eq:genc}
a=f_1+f_2+E
\end{equation}
is a non-nilpotent generalized cyclic element, associated to $f$.
\item\label{prop:gcycintrip} The element $a$, defined by \eqref{eq:genc}, determines the integrable triple uniquely.
\end{enumerate}
\end{proposition}
\begin{proof}
By the definition, $a=(f_1+E)+f_2$ is the Jordan decomposition of $a$, where
$f_1+E$ and $f_2$ are its semisimple and nilpotent parts.
Since $f_1+E$ is a non-zero semisimple element, we conclude that $a$ is not nilpotent.
By the definition, $a$ is a generalized cyclic element, proving \eqref{prop:intripgcyc}.

Due to uniqueness of the Jordan decomposition, $a$ determines $f_1+E$ and $f_2$,
hence it determines $f_1=f-f_2$ and $E=a-f$, proving \eqref{prop:gcycintrip}.
\end{proof}

\begin{definition}
A generalized cyclic element $f+E$ is called \emph{integrable} if it is obtained from an integrable triple as in \eqref{eq:genc}.
In this case the element $E\in\mf g_j$ is called \emph{integrable for} $f$.
\end{definition}

Note that a generalized cyclic element is integrable if and only if the nilpotent part of its Jordan decomposition lies in $\mf g_{-2}$, and also that a generalised cyclic element cannot be nilpotent.
Also, of course, any semisimple generalized cyclic element is integrable.

\begin{lemma}\label{lem:nilthennil}
If $E\in\mf g_j$, $j\ge1$, is a nilpotent element of the $Z(\mf s)$-module $\mf g_j$,
then the element $f+E$ is a nilpotent element of $\mf g$.
\end{lemma}

\begin{proof}
Clearly the Zariski closure of $Z(\mf s)(f+E)$ contains $f$.
Hence $\ol{G(f+E)}\supset\ol{Gf}\ni0$ since $f$ is a nilpotent element of $\mf g$.
\end{proof}

\begin{proposition}\label{prop:nonebelow}
An integrable generalized cyclic element $f+E$, associated to a nilpotent element $f$ of depth $d$, where $E\in\mf g_j$,
exists only when $j=d-1$ or $d$.
\end{proposition}
\begin{proof}
In most of the cases the center of the subalgebra $\mf m:=\bigoplus_{k\ge2}\mf g_k$ lies in $\mf g_{d-1}+\mf g_d$.
In a few cases this center has non-zero elements $E$ in $\mf g_{d-2}$, but it turns out that if the coisotropy condition in (I\ref{def:intrip2}) is satisfied, then $f+E$ is nilpotent, thus not integrable.
For exceptional $\mf g$ the only $f$ for which the center of $\mf m$ is larger than $\mf g_{d-1}+\mf g_d$ are of type $\rA_2+\rA_1$ in $\rE_6$, $\rE_7$, $\rE_8$ of depth $4$, but the elements from $\mf m\setminus(\mf g_{d-1}+\mf g_d)$ are not integrable.
This is checked by a case-wise verification, with the aid of computer. The proposition is proved by direct computations in the standard representation for classical $\mf g$.
It would be interesting to find a general proof.
\end{proof}

The following theorem classifies the non-nilpotent cyclic elements.

\begin{theorem}\label{thm:nonniliffnonnil}
Let $f$ be a non-zero nilpotent element of $\mf g$ and let $E\in\mf g_d$.
Then the cyclic element $f+E$ is not nilpotent if and only if
the element $E$ of the $Z(\mf s)$-module $\mf g_d$ is not nilpotent.
\end{theorem}
\begin{proof}
It follows from Lemma \ref{lem:nilthennil} that $f+E$ is nilpotent if $E$ is a nilpotent element of the $Z(\mf s)$-module $\mf g_d$.

Conversely, suppose $E$ is not a nilpotent element of the $Z(\mf s)$-module $\mf g_d$.
Then $\ol{Z(\mf s)(E)}$ contains a non-zero semisimple element $E_0$ of the $Z(\mf s)$-module $\mf g_d$,
and by Theorem \ref{thm:Cartanstuff} \eqref{smallWeyl} we may assume that $E_0\in C$, a Cartan subspace.
Hence, we have to prove that $f+E_0$ is not a nilpotent element of $\mf g$.
For that, by results of \cite{EKV13}, it suffices to prove that $\fs+E_0$ is not a nilpotent element of $\mf g$.
By the results of \cite{EJK20}, we may assume that $\fs$ is an irreducible nilpotent element of $\mf g$,
in which case $\mf g_d=C$, and the set $\qcss_{\mf g}(f)=\{E\in C\mid\text{$f+E$ is semisimple}\}$
is a complement in $C$ of a union of $\frac{m(m+1)}2$ hyperplanes, where $m=\dim\mf g_d$.
In the case when $m=1$ this means that $\qcss_{\mf g}(f)=C\setminus\{0\}$, and we are done.

The remaining cases are the following ones from \cite[Table 1]{EJK20}:
$4_{k-1}$, $m=2$; $7$, $m=2$; $11$, $m=3$; $16$, $m=2$; $17$, $m=2$; $18$, $m=4$.
For all but the first one, $\mf g$ is an exceptional Lie algebra, and the case can be checked on a computer.
In \cite{EJK20}, singular sets (complements to $\mathbf{S}_{\mf g}(f)$) for these cases have been described as unions of certain subspaces.
Using GAP package SLA, one checks that $f+E$ is not nilpotent for generic elements $E$ in all possible intersections of these subspaces.
In the first case $\mf g=\mf{so}_{4k}$ and $f$ corresponds to the partition $(2k+1,2k-1)$. This case is treated in Example \ref{ex:D2k} below.
\end{proof}

\begin{example}\label{ex:D2k}

Let $f$ be a nilpotent in a simple Lie algebra of type $\rD_{2k}$ corresponding to the partition $(2k+1,2k-1)$. By definition, this means that the standard representation of $\mf{so}_{4k}$ has a basis $x_{-k},x_{-(k-1)},...,x_{-1}$, $x_0$, $x_1,...,x_{k-1}$, $x_k$, $y_{-(k-1)},...,y_{-1}$, $y_0$, $y_1,...,y_{k-1}$, with $h$ acting by $hx_j=2jx_j$, $hy_j=2jy_j$, and $f$ acting by
\[
\begin{tikzpicture}
\node (x0) at (0,0) {$x_0$};
\node (x1) [right=1em of x0] {$x_1$};
\node (xp) [right=1em of x1] {$\cdots$};
\node (xkm1) [right=1em of xp] {$x_{k-1}$};
\node (xk) [right=1em of xkm1] {$x_k$};
\node (xm1) [left=1em of x0] {$x_{-1}$};
\node (xm) [left=1em of xm1] {$\cdots$};
\node (x1mk) [left=1em of xm] {$x_{-(k-1)}$};
\node (xmk) [left=1em of x1mk] {$x_{-k}$};
\node (y0) [below=1em of x0] {$y_0$};
\node (y1) [below=1em of x1] {$y_1$};
\node (yp) [right=1em of y1] {$\cdots$};
\node (ykm1) [below=1em of xkm1] {$y_{k-1}$};
\node (ym1) [below=1em of xm1] {$y_{-1}$};
\node (ym) [left=1em of ym1] {$\cdots$};
\node (y1mk) [below=1em of x1mk] {$y_{-(k-1)}$};
\draw[<-|] (xmk) -- (x1mk);
\draw[<-|] (x1mk) -- (xm);
\draw[<-|] (xm) -- (xm1);
\draw[<-|] (xm1) -- (x0);
\draw[<-|] (x0) -- (x1);
\draw[<-|] (x1) -- (xp);
\draw[<-|] (xp) -- (xkm1);
\draw[<-|] (xkm1) -- (xk);
\draw[<-|] (y1mk) -- (ym);
\draw[<-|] (ym) -- (ym1);
\draw[<-|] (ym1) -- (y0);
\draw[<-|] (y0) -- (y1);
\draw[<-|] (y1) -- (yp);
\draw[<-|] (yp) -- (ykm1);
\end{tikzpicture}
\]
We recall from \cite[Case $4_k$]{EJK20} (reversing everything there) that there is a basis $(E_1,E_2)$ of $\mf g_d$ with
$$
E_1x_{-(k-1)}=x_k,\ E_1x_{-k}=x_{k-1},\ E_2y_{-(k-1)}=-x_k,\ E_2x_{-k}=y_{k-1},
$$
and all other actions of $E_1$, $E_2$ are zero. (Explicitly, one can take $E_1=\frac{e_{\eps_1+\eps_2}+e_{\eps_1+\eps_3}}2$, $E_2=\frac{e_{\eps_1+\eps_2}-e_{\eps_1+\eps_3}}2$.)

Pictorially,
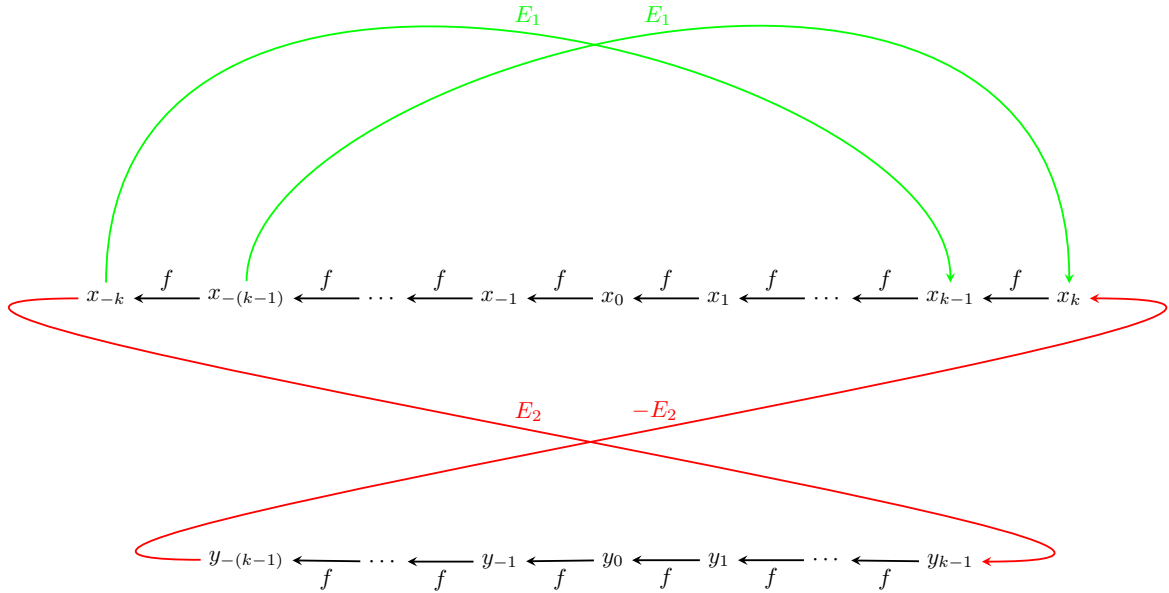
\begin{center}
\resizebox{.04\textwidth}{!}{
\begin{tikzpicture}[>=stealth,->,xscale=.05]
\useasboundingbox (-5,-5) rectangle (5,5);
\node (x0) at (0,0) {$x_0$};
\node (x1) [right=of x0] {$x_1$};
\node (xp) [right=of x1] {$\cdots$};
\node (xkm1) [right=of xp] {$x_{k-1}$};
\node (xk) [right=of xkm1] {$x_k$};
\node (xm1) [left=of x0] {$x_{-1}$};
\node (xm) [left=of xm1] {$\cdots$};
\node (x1mk) [left=of xm] {$x_{-(k-1)}$};
\node (xmk) [left=of x1mk] {$x_{-k}$};
\node (y0) [below=10em of x0] {$y_0$};
\node (y1) [below=10em of x1] {$y_1$};
\node (yp) [right=of y1] {$\cdots$};
\node (ykm1) [below=10em of xkm1] {$y_{k-1}$};
\node (ym1) [below=10em of xm1] {$y_{-1}$};
\node (ym) [left=of ym1] {$\cdots$};
\node (y1mk) [below=9.8em of x1mk] {$y_{-(k-1)}$};
\path[decoration=arrows, decorate] (xmk) -- node [above] {$f$} (x1mk) -- node [above] {$f$} (xm) -- node [above] {$f$} (xm1) -- node [above] {$f$} (x0) -- node [above] {$f$} (x1) -- node [above] {$f$} (xp) -- node [above] {$f$} (xkm1) -- node [above] {$f$} (xk);
\path[decoration=arrows, decorate] (y1mk) -- node [below] {$f$} (ym) -- node [below] {$f$} (ym1) -- node [below] {$f$} (y0) -- node [below] {$f$} (y1) -- node [below] {$f$} (yp) -- node [below] {$f$} (ykm1);
\draw[green,thick] (xmk) .. controls +(up:20em) and +(up:10em) .. node[above,green] {$E_1$} (xkm1);
\draw[green,thick] (x1mk) .. controls +(up:10em) and +(up:20em) .. node[above,green] {$E_1$} (xk);
\draw[red,thick] (y1mk) .. controls +(left:400em) and +(right:400em) .. node[above,red] {$-E_2$} (xk);
\draw[red,thick] (xmk) .. controls +(left:400em) and +(right:400em) .. node[above,red] {$E_2$} (ykm1);
\end{tikzpicture}
}
\end{center}

As calculated in \cite[Case $4_k$]{EJK20}, the cyclic element $f+\lambda_1E_1+\lambda_2E_2$ is not semisimple only when either $\lambda_2=0$ or $\lambda_2=\pm\lambda_1$.

Consider first the case $\lambda_2=0$, so that $E=\lambda E_1$. We then define $\fs$, $\fn$ by declaring that $\fs$ acts only on the basis elements $x_i$ and $\fn$ only on the basis elements $y_j$. Thus the corresponding picture is then
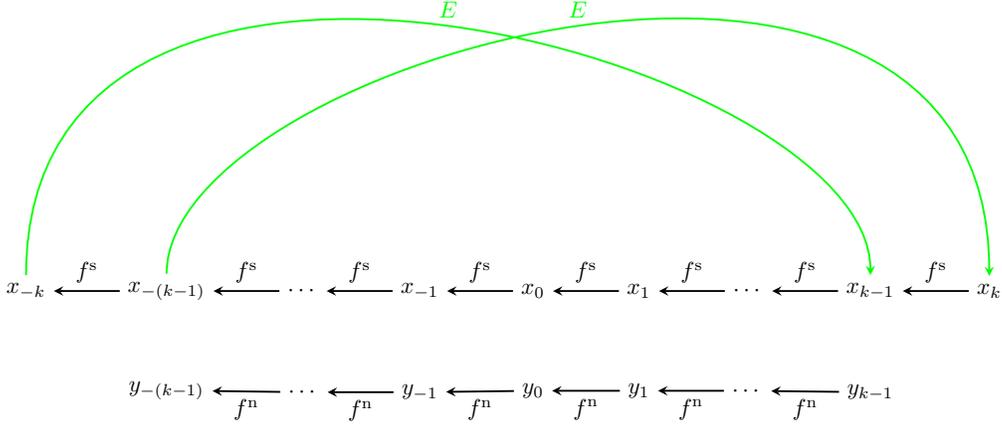
\begin{center}
\resizebox{.04\textwidth}{!}{
\begin{tikzpicture}[>=stealth,->,xscale=.05]
\useasboundingbox (-5,-3) rectangle (5,5);
\node (x0) at (0,0) {$x_0$};
\node (x1) [right=of x0] {$x_1$};
\node (xp) [right=of x1] {$\cdots$};
\node (xkm1) [right=of xp] {$x_{k-1}$};
\node (xk) [right=of xkm1] {$x_k$};
\node (xm1) [left=of x0] {$x_{-1}$};
\node (xm) [left=of xm1] {$\cdots$};
\node (x1mk) [left=of xm] {$x_{-(k-1)}$};
\node (xmk) [left=of x1mk] {$x_{-k}$};
\node (y0) [below=3em of x0] {$y_0$};
\node (y1) [below=3em of x1] {$y_1$};
\node (yp) [right=of y1] {$\cdots$};
\node (ykm1) [below=3em of xkm1] {$y_{k-1}$};
\node (ym1) [below=3em of xm1] {$y_{-1}$};
\node (ym) [left=of ym1] {$\cdots$};
\node (y1mk) [below=2.8em of x1mk] {$y_{-(k-1)}$};
\path[decoration=arrows, decorate] (xmk) -- node [above] {$\fs$} (x1mk) -- node [above] {$\fs$} (xm) -- node [above] {$\fs$} (xm1) -- node [above] {$\fs$} (x0) -- node [above] {$\fs$} (x1) -- node [above] {$\fs$} (xp) -- node [above] {$\fs$} (xkm1) -- node [above] {$\fs$} (xk);
\path[decoration=arrows, decorate] (y1mk) -- node [below] {$\fn$} (ym) -- node [below] {$\fn$} (ym1) -- node [below] {$\fn$} (y0) -- node [below] {$\fn$} (y1) -- node [below] {$\fn$} (yp) -- node [below] {$\fn$} (ykm1);
\draw[green,thick] (xmk) .. controls +(up:20em) and +(up:10em) .. node[above,green] {$E$} (xkm1);
\draw[green,thick] (x1mk) .. controls +(up:10em) and +(up:20em) .. node[above,green] {$E$} (xk);
\end{tikzpicture}
}
\end{center}
Then clearly $\fs+\fn=f$ and $\fn$ commutes both with $\fs$ and with $E$. Moreover both $\fs$ and $\fn$ are homogeneous of degree $-2$ since they both act on the standard representation lowering degree by $2$. Moreover $\fs+E$ is semisimple since the characteristic polynomial of its action in the standard representation is $t(t^{2k}-\lambda)$.

Next consider the case $\lambda_2=\pm\lambda_1$, so that $E=\lambda(E_1\pm E_2)$. Switch to the following basis:
\[
z_k=x_k, z_j=\frac{(k+j)x_j\mp(k-j)y_j}{2k}, t_j=\frac{x_j\pm y_j}{2k}\ (-k<j<k),\ t_{-k}=\frac{x_{-k}}{2k}.
\]
In this basis, the non-zero actions of $f$ and $E$ are as follows:
\[
\begin{aligned}
ft_j=t_{j-1}&\ (-(k-1)\le j<k),\\
fz_j=z_{j-1}+t_{j-1}&\ (-(k-1)<j\le k),\\
fz_{-(k-1)}=t_{-k}
\end{aligned}
\]
and
\[
Et_{-k}=\lambda t_{k-1},\ Ez_{-(k-1)}=\lambda z_k.
\]
In this basis, let us define $\fs$ and $\fn$ by the non-zero actions
\[
\begin{aligned}
\fs z_j=z_{j-1}&\ (-(k-1)<j\le k),\\
\fs t_j=t_{j-1}&\ (-(k-1)\le j<k),\\
\fn z_j=t_{j-1}&\ (-(k-1)\le j\le k).
\end{aligned}
\]
Then again both $\fs$ and $\fn$ are in $\mf g_{-2}$, $\fs+\fn=f$, $\fn$ commutes with both $\fs$ and $E$, and $\fs+E$ is semisimple. Note also that $\fs+E$ is not regular, having double eigenvalues. Its characteristic polynomial in the standard representation is $(t^{2k}-\lambda)^2$ (the minimal polynomial being $t^{2k}-\lambda$).

All this is completely obvious from the following diagram of actions:

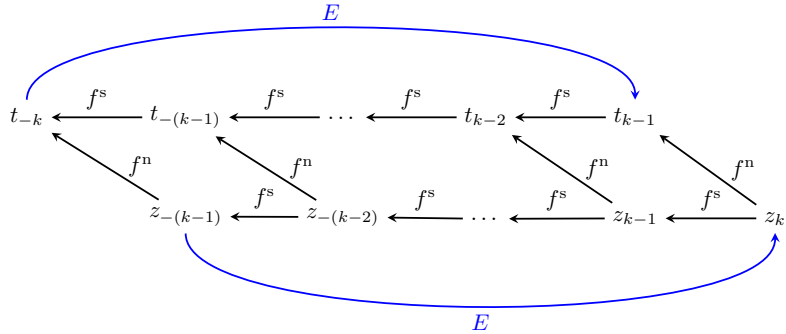
\begin{center}
\resizebox{!}{30ex}{
\begin{tikzpicture}[>=stealth,->,xscale=.05,node distance=4ex and 4em]
\node (zmk) at (0,0) {$t_{-k}$};
\node (zmkp1) [right=of zmk] {$t_{-(k-1)}$};
\node (zp) [right=of zmkp1] {$\dots$};
\node (zkm2) [right=of zp] {$t_{k-2}$};
\node (zkm1) [right=of zkm2] {$t_{k-1}$};
\node (tmkp1) [below=2.8em of zmkp1] {$z_{-(k-1)}$};
\node (tmkp2) [below=3.25em of zp] {$z_{-(k-2)}$};
\node (tp) [below=3.3em of zkm2] {$\dots$};
\node (tkm1) [below=3em of zkm1] {$z_{k-1}$};
\node (tk) [right=of tkm1] {$z_k$};
\path[decoration=arrows, decorate] (zmk) -- node [above] {$\fs$} (zmkp1) -- node [above] {$\fs$} (zp) -- node [above] {$\fs$} (zkm2) -- node [above] {$\fs$} (zkm1);
\path[decoration=arrows, decorate] (tmkp1) -- node [above] {$\fs$} (tmkp2) -- node [above] {$\fs$} (tp) -- node [above] {$\fs$} (tkm1) -- node [above] {$\fs$} (tk);
\path[decoration=arrows, decorate] (zmk) -- node [right] {$\ \ \fn$} (tmkp1);
\path[decoration=arrows, decorate] (zmkp1) -- node [right] {$\ \ \fn$} (tmkp2);
\path[decoration=arrows, decorate] (zkm2) -- node [right] {$\ \ \fn$} (tkm1);
\path[decoration=arrows, decorate] (zkm1) -- node [right] {$\ \ \fn$} (tk);
\draw[blue,thick] (zmk) .. controls +(up:5em) and +(up:5em) .. node[above,blue] {$E$} (zkm1);
\draw[blue,thick] (tmkp1) .. controls +(down:5em) and +(down:5em) .. node[below,blue] {$E$} (tk);
\end{tikzpicture}
}
\end{center}
Thus in both cases $(f_1,f_2,E)$ with $f_1=\fs$ and $f_2=\fn$ are integrable triples.
\end{example}

\begin{example}\label{ex:rank2}
According to \cite[Table 1]{EJK20}, there are three irreducible nilpotent orbits of rank 2 in exceptional Lie algebras:
the one with label $\rF_4(a_2)$ in $\rF_4$ and the orbits with labels $\rE_8(a_5)$ and $\rE_8(a_6)$ in $\rE_8$.
For all of them there are exactly three lines in $\mf g_d$ such that $f+E$ is semisimple if and only if $E\in\mf g_d$
does not belong to any of these lines. Direct calculations with the GAP package SLA show that $E$ remains integrable
along these lines too. We provide more detailed description using specific choices of representatives for these orbits as follows.

In $\rF_4$, take for the representative of $\rF_4(a_2)$ the element
\[
f=f_{1100}+f_{0120}+f_{0011}+f_{0001}.
\]

In $\rE_8$, take for the representative of $\rE_8(a_5)$ the element
\[
f=
f_{\foreight{.5}10000000}+
f_{\foreight{.5}10100000}+
f_{\foreight{.5}01110000}+
f_{\foreight{.5}01011000}+
f_{\foreight{.5}00111000}+
f_{\foreight{.5}00011100}+
f_{\foreight{.5}00001110}+
f_{\foreight{.5}00000011}
\]
and for the representative of $\rE_8(a_6)$ the element
\[
f=
f_{\foreight{.5}10110000}+
f_{\foreight{.5}11110000}+
f_{\foreight{.5}01111100}+
f_{\foreight{.5}00010000}+
f_{\foreight{.5}00011000}+
f_{\foreight{.5}00011100}+
f_{\foreight{.5}00001110}+
f_{\foreight{.5}00000011}.
\]

Then, for these choices of $f$, $\mf g_d$ has basis $E_1$, $E_2$ consisting of the highest root vector and the next to the highest root vector.
The cyclic element $f+x_1E_1+x_2E_2$ is semisimple except when $x_1=0$, $x_2=0$ or $x_1=x_2$.
In these three singular cases, Jordan decomposition of $f+E$ is $(\fs+E)+\fn$ with $\fn\in\mf g_{-2}$,
where

\noindent for $E=x_1E_1$ or $E=x_2E_2$, we have the following cases:

if $f$ is $\rF_4(a_2)$, then $\fs$ has label $\rB_3$ and $\fn$ has label $\widetilde\rA_2$,

if $f$ is $\rE_8(a_5)$, then $\fs$ has label $\rE_6$ and $\fn$ has label $\rD_4$,

\noindent while for $E=x(E_1+E_2)$,

if $f$ is $\rF_4(a_2)$, then $\fs$ has label $\rC_3$ and $\fn$ has label $\rA_1$,

if $f$ is $\rE_8(a_5)$, then $\fs$ has label $\rD_7$ and $\fn$ has label $2\rA_1$,

\noindent and if $f$ is $\rE_8(a_6)$, then for $E$ on any of the three singular lines $\fs$ has label $\rD_6$ and $\fn$ has label $\rA_3$.

Note also that the subalgebra $\mf a$ generated by $f$ and $E$, which is the whole $\mf g$ if $E$ does not lie in the union of singular lines,
is the direct sum of the semisimple $[\mf a,\mf a]$ and the center spanned by $\fn$,
where $[\mf a,\mf a]$ has the following type:

\noindent for $E=x_1E_1$ or $E=x_2E_2$, it is

$\rG_2$, when $f$ is $\rF_4(a_2)$,

$\rF_4$, when $f$ is $\rE_8(a_5)$,

\noindent and for $E=x(E_1+E_2)$, it is

$\rC_3$, when $f$ is $\rF_4(a_2)$,

$\rB_6$, when $f$ is $\rE_8(a_5)$,

\noindent while if $f$ is $\rE_8(a_6)$, then for $E$ from any of the three singular lines, this subalgebra has type $\rB_5$.

\end{example}

\begin{example}\label{ex:rank3}
For the nilpotent with label $\rE_7(a_5)$ in $\rE_7$, depth is $10$, with $\mf g_{10}$ 3-dimensional.
We take
\[
f=f_{\foreseven{.1}{.5}0000001}+f_{\foreseven{.1}{.5}0000011}+f_{\foreseven{.1}{.5}0000111}+f_{\foreseven{.1}{.5}0101110}+f_{\foreseven{.1}{.5}0011100}+f_{\foreseven{.1}{.5}0111100}+f_{\foreseven{.1}{.5}1011000};
\]
then $\mf g_{10}$ is spanned by $E_1=e_{\foreseven{.1}{.5}1224321}$, $E_2=e_{\foreseven{.1}{.5}1234321}$ and $E_3=e_{\foreseven{.1}{.5}2234321}$.
The singular set, i.~e. the subset of those $E\in\mf g_{10}$ with $f+E$ not semisimple, is the union of six 2-dimensional subspaces: $\langle E_i,E_j\rangle$ and $\langle E_i,E_j+E_k\rangle$, with $\{i,j,k\}=\{1,2,3\}$. Their pairwise intersections produce the following seven 1-dimensional subspaces: $\langle E_i\rangle$, $\langle E_i+E_j\rangle$ and $\langle E_1+E_2+E_3\rangle$.

It is also possible to describe this set without mentioning semisimplicity of $f+E$: it is the set of all those vectors which have nontrivial stabilizer with respect to the action $Z(\mf s)|\mf g_{10}$ which is the permutation representation of the symmetric group on three letters.

All these subspaces are exactly all those subspaces $V$ of $\mf g_{10}$ with the following property:
the Lie subalgebra $\mf a$ generated by $V$ and $f$ is the direct sum of the semisimple subalgebra $\mf q=[\mf a,\mf a]$ and the $1$-dimensional center $\mf z(\mf a)\subseteq\mf g_{-2}$; $\mf z(\mf a)$ is spanned by an element $\fn$ such that $\fs=f-\fn$ is of semisimple type in $\mf q$; $\fs$ has the same depth $10$ in $\mf q$ as $f$ in $\mf g$; and $\mf q_{10}=V$.

In all these cases, taking any $E\in\mf q_{10}=V$ which does not lie in a smaller subspace from the singular set, $\fs+E$ is semisimple, so that one obtains an integrable triple $(f_1,f_2,E)=(\fs,\fn,E)$.

The corresponding subspaces and nilpotent types are as follows:

\begin{center}
\resizebox{.6\textwidth}{!}{
\begin{tabular}{ccccc}
$V$&$\fs$&$\fn$&$\mf q$&$\fs$ in $\mf q$\\
\hline
$\langle E_i,E_j+E_k\rangle$&$\rD_6(a_2)$&$\rA_1$&$\rD_6$&$(7,5)$\\
$\langle E_i,E_j\rangle$&$\rE_6(a_3)$&$3\rA_1''$&$\rF_4$&$\rF_4(a_2)$\\
$\langle E_i\rangle$&$\rD_4$&$(\rA_3+\rA_1)''$&$\rG_2$&principal\\
$\langle E_i+E_j\rangle$&$\rA_5'$&$4\rA_1$&$\rC_3$&principal\\
$\langle E_1+E_2+E_3\rangle$&$\rA_5''$&$\rA_2$&$\rC_3$&principal
\end{tabular}
}
\end{center}
\end{example}

\begin{example}\label{ex:rank4}
For the nilpotent with label $\rE_8(a_7)$ in $\rE_8$, depth is $10$, with $\mf g_{10}$ 4-di\-men\-si\-o\-nal.
We take
\[
f=f_{\foreight{.5}00001110}+f_{\foreight{.5}00011100}+f_{\foreight{.5}01011100}+f_{\foreight{.5}00111111}+f_{\foreight{.5}01121000}+f_{\foreight{.5}10111100}
+f_{\foreight{.5}11121000}+f_{\foreight{.5}11111100}+f_{\foreight{.5}11111110}.
\]
Then $\mf g_{10}$ is spanned by $E_1=e_{\foreight{.5}23465321}$, $E_2=e_{\foreight{.5}23465421}$, $E_3=e_{\foreight{.5}23465431}$ and $E_4=e_{\foreight{.5}23465432}$.

The singular set in this case is the union of ten 3-dimensional subspaces $\langle E_i,E_j,E_k\rangle$ and $\langle E_i,E_j,E_k+E_\ell\rangle$, $\{i,j,k,\ell\}=\{1,2,3,4\}$, forming a single orbit under the action $Z(\mf s)|\mf g_{10}$, which is the action of the component group $S_5$ on its 4-dimensional irreducible representation.

As in the Example \ref{ex:rank3}, this set consists precisely of those vectors which have nontrivial stabilizer with respect to the action of the component group of $Z(\mf s)|\mf g_{10}$.

All of the possible intersections of these 3-dimensional subspaces produce twenty five 2-dimensional subspaces forming two orbits, one containing all $\langle E_i,E_j\rangle$ and all $\langle E_i,E_j+E_k+E_\ell\rangle$ and another containing all $\langle E_i, E_j+E_k\rangle$ and all $\langle E_i+E_j,E_k+E_\ell\rangle$, and fifteen 1-dimensional subspaces forming two orbits, one containing all $\langle E_i\rangle$ and $\langle E_1+E_2+E_3+E_4\rangle$, and another containing all $\langle E_i+E_j\rangle$ and all $\langle E_i+E_j+E_k\rangle$.

All these $10+25+15=50$ subspaces are also exactly all those subspaces $V$ of $\mf g_{10}$ with the following property:
the Lie subalgebra $\mf a$ generated by $V$ and $f$ is the direct sum of the semisimple (in fact, simple) subalgebra $\mf q=[\mf a,\mf a]$ and the $1$-dimensional center $\mf z(\mf a)\subseteq\mf g_{-2}$; $\mf z(\mf a)$ is spanned by an element $\fn$ such that $\fs=f-\fn$ is of semisimple type (in fact, irreducible) in $\mf q$;  $\fs$ has the same depth $10$ in $\mf q$ as $f$ in $\mf g$; and $\mf q_{10}=V$.

In all these cases, taking any $E\in\mf q_{10}=V$ which does not lie in a smaller subspace from the singular set, $\fs+E$ is semisimple, so that one obtains an integrable triple $(f_1,f_2,E)=(\fs,\fn,E)$.

The subspaces, their generic stabilizers in $S_5$, the corresponding subalgebras and nilpotent orbit labels are as follows:

\begin{center}
\resizebox{.9\textwidth}{!}{
\begin{tabular}{cccccc}
$V$&generic stabilizer&$\fs$&$\fn$&$\mf q$&$\fs$ in $\mf q$\\
\hline
$\langle E_i,E_j,E_k\rangle$, $\langle E_i,E_j,E_k+E_\ell\rangle$&$S_2$&$\rE_7(a_5)$&$\rA_1$&$\rE_7$&$\rE_7(a_5)$\\
$\langle E_i,E_j+E_k\rangle$, $\langle E_i+E_j,E_k+E_\ell\rangle$&$S_2\times S_2$&$\rD_6(a_2)$&$2\rA_1$&$\rD_6$&$(7,5)$\\
$\langle E_i,E_j\rangle$, $\langle E_i,E_j+E_k+E_\ell\rangle$&$S_3$&$\rE_6(a_3)$&$\rA_2$&$\rF_4$&$\rF_4(a_2)$\\
$\langle E_i+E_j\rangle$, $\langle E_i+E_j+E_k\rangle$&$S_2\times S_3$&$\rA_5$&$\rA_2+\rA_1$&$\rC_3$&principal\\
$\langle E_i\rangle$, $\langle E_1+E_2+E_3+E_4\rangle$&$S_4$&$\rD_4$&$\rD_4(a_1)$&$\rG_2$&principal
\end{tabular}
}
\end{center}
\end{example}

The following lemma helps to establish non-existence of non-nilpotent quasi-cyclic elements.

\begin{lemma}\label{lem:nononnil}
Let $C\subset\mf g_{d-1}$ be a Cartan subspace of the $Z(\mf s)$-module $\mf g_{d-1}$,
and suppose that $C$ contains no elements $E$, for which $f+E$ is quasi-cyclic.
Then $\mf g_{d-1}$ contains no elements $E$, for which $f+E$ is non-nilpotent quasi-cyclic.
\end{lemma}
\begin{proof}
If $f+E$ is not nilpotent quasi-cyclic, then, by Lemma \ref{lem:nilthennil},
$E$ is not a nilpotent element for the $Z(\mf s)$-module $\mf g_{d-1}$.
Then, as in the proof of Theorem \ref{thm:nonniliffnonnil},
$\ol{Z(\mf s)E}\cap C$ contains an element $E_0$, such that $f+E_0$ is quasi-cyclic.
\end{proof}

In order to prove the next proposition, we will need the following simple lemma.
\begin{lemma}\label{lem:linalg}
Let $a\in\mf g$. Then $(\ker\ad a)^{\perp}=\im \ad a$, where the orthogonal is with respect to the bilinear form $(\cdot\mid\cdot)$.
In particular, if $a\in\mf g_k$, then $(\ker\ad a|_{\mf g_\ell})^{\perp}=\im \ad a|_{\mf g_{-\ell-k}}$.
\end{lemma}
\begin{proof}
The inclusion $\im\ad a\subset(\ker\ad a)^\perp$ is immediate by the invariance of the bilinear form.
On the other hand
$$
\dim(\im\ad a)=\dim(\mf g)-\dim(\ker\ad a)=\dim\left((\ker\ad a)^\perp\right)
\,,
$$
since the bilinear form is non-degenerate.
Hence, $(\ker\ad a)^{\perp}=\im \ad a$. The second part of the lemma follows from the fact that
$(\mf g_k\,|\,\mf g_\ell)=0$ if $k\ne-\ell$.
\end{proof}
\begin{proposition}\label{prop:linalg}
Let $E\in\mf g_j$, $j\ge0$. The following statements are equivalent:
\begin{enumerate}[(a)]
\item the centralizer of $E$ in $\mf g_{1}$ is coisotropic with respect to the bilinear form \eqref{eq:skewform};
\item the map $\ad E\circ(\ad f)^{-1}\circ\ad E|_{\mf g_{-j-1}}$ is zero.
\end{enumerate}
\end{proposition}
\begin{proof}
The coisotropy condition (a) can be rephrased by saying that the orthogonal complement of $\mf g_1^E$ with
respect to $\omega$ is contained in $\mf g_1^E$. This, by definition \eqref{eq:skewform}
of $\omega$, is equivalent to the following condition:
\begin{itemize}
\item[(a$'$)]
if $a\in\mf g_1$ is such that $[f,a]\in(\ker\ad E|_{\mf g_1})^{\perp}$, then $a\in\ker\ad E|_{\mf g_1}$.
\end{itemize}
(Here the orthocomplement is with respect to the bilinear form $(\cdot\mid\cdot)$ of $\mf g$.)
On the other hand, the statement in b) can be equivalently rephrased as follows:
\begin{itemize}
    \item[(b$'$)] if $x\in\mf g_{-k-1}$ is such that
$[E,x]=[f,a]$, for some $a\in\mf g_1$, then $a\in\ker\ad E|_{\mf g_1}$,
\end{itemize}
or, equivalently,
\begin{itemize}
\item[(b$''$)] if $a\in\mf g_1$ is such that $[f,a]\in\im\ad E|_{\mf g_{-j-1}}$, then $a\in\ker\ad E|_{\mf g_1}$.
\end{itemize}
By Lemma \ref{lem:linalg},
$(\ker\ad E|_{\mf g_1})^{\perp}=\im \ad E|_{\mf g_{-j-1}}$, thus (a') and (b'') are equivalent.
\end{proof}

An important problem is when a (non-nilpotent) generalized cyclic element $f+E$ of a semisimple Lie algebra $\mf g$ is integrable.
For the solution of this problem the following lemma is important.
\begin{lemma}\label{lem:int=>ss}
Let $f$ be a non-zero nilpotent of $\mf g$ and let \eqref{eq:grading} be the corresponding $\mb Z$-grading of depth $d$. Let $j=d$ or $d-1$
and let $E\in\mf g_j$ be a non-zero element. If the nilpotent part of the Jordan decomposition of $f+E$ lies in $\mf g_{-2}$,
then $E$ is a semisimple element of the $Z(\mf s)$-module $\mf g_j$.
\end{lemma}
\begin{proof}
Consider the theta group $G_0|(\mf g_{-2}+\mf g_j)$, constructed in Remark \ref{rem:ssproof}. Then the element $f+E\in\mf g_{-2}+\mf g_j$ is a semisimple element of $\mf g$
if and only if the orbit $G_0(f+E)$ is closed. We have (cf. \cite{EKV13}, proof of Proposition 2.2(a)):
\begin{equation}\label{eq:fromG0toZs}
f+Z(\mf s)(E)=G_0(f+E)\cap(f+\mf g_j).
\end{equation}
If the nilpotent part of the Jordan decomposition of $f+E$ lies in $\mf g_{-2}$, we have: $f=f_1+f_2$, $f_1+E$ is semisimple, $f_2\in\mf g_{-2}$, and $[f_2,f_1+E]=0$.
Since $G_0(f_2)\subset\mf g_{-2}$, we deduce from \eqref{eq:fromG0toZs} that it still holds if $f$ is replaced by $f_1$.
Since $f_1+E$ is semisimple, the orbit $G_0(f_1+E)$ is closed, hence $f_1+Z(\mf s)E$ is a closed subset, and $E$ is a semisimple element of $Z(\mf s)|\mf g_j$.
\end{proof}

The following theorem describes all integrable cyclic elements, up to conjugation by $Z(\mf s)$.
\begin{theorem}\label{thm:intiffss}
Let $f$ be a non-zero nilpotent element of $\mf g$ of even depth $d$, and let $E\in\mf g_d$ be a non-zero element.
Then the cyclic element $f+E$ is integrable if and only if $E$ is a semisimple element of the $Z(\mf s)$-module $\mf g_d$.
\end{theorem}
\begin{proof}
If $f+E$ is integrable, then $E$ is a semisimple element of $Z(\mf s)|\mf g_d$ by Lemma \ref{lem:int=>ss}.

Conversely, let $E$ be a non-zero semisimple element of $Z(\mf s)|\mf g_d$.
Then the argument as in the proof of Theorem \ref{thm:nonniliffnonnil} reduces the proof to the case
when $f$ is an irreducible nilpotent element with $\dim\mf g_d>1$ and $E\notin\mathbf S_{\mf g}(f)$.
Again, the cases when $\mf g$ is an exceptional Lie algebra are checked on the computer using GAP \cite{SLA},
by computing Jordan decompositions of $f+E$ for $E$ generic in all possible nonzero intersections
of subspaces constituting the complement of $\mathbf S_{\mf g}(f)$ as described in \cite{EJK20}  (see Examples \ref{ex:rank2}, \ref{ex:rank3} and \ref{ex:rank4}),
while the case $\mf g=\mf{so}_{4k}$, with $f$ corresponding to the partition $(2k+1,2k-1)$ is treated as in Example \ref{ex:D2k}.
\end{proof}
\begin{theorem}\label{thm:quasintiffss}\
\begin{enumerate}[(a)]
\item\label{thm:quasintiffssodd} Theorem \ref{thm:intiffss} holds for $d$ odd and for $E\in\mf g_{d-1}$ such that $f+E$ is a quasi-cyclic element.
\item\label{thm:quasintiffsseven} The ``only if'' part of Theorem \ref{thm:intiffss} holds for $d$ even and for $E\in\mf g_{d-1}$ such that $f+E$ is a quasi-cyclic element.
\end{enumerate}
\end{theorem}
\begin{proof}
\eqref{thm:quasintiffssodd} Replacing $\mf g$ with $\mf g_{\text{even}}=\bigoplus_{j\in\mb Z}\mf g_{2j}$, the proof is the same as of Theorem \ref{thm:intiffss}.
\eqref{thm:quasintiffsseven} follows from Lemma \ref{lem:int=>ss}.
\end{proof}
\begin{example}\label{ex:noint}
A non-nilpotent quasi-cyclic element does not necessarily give rise to an integrable triple.
For $f$ with label $2\rA_2+\rA_1$ in $\rE_6$, take
\[
f=2f_{\foresix{.1}{.5}111100}+2f_{\foresix{.1}{.5}101110}+2f_{\foresix{.1}{.5}010111}+2f_{\foresix{.1}{.5}001111}+f_{\foresix{.1}{.5}011210}.
\]
The depth of $f$ is $5$ and $\mf g_4$ is $4$-dimensional,
spanned by $e_{\foresix{.1}{.5}111211}$, $e_{\foresix{.1}{.5}112211}$, $e_{\foresix{.1}{.5}111221}$ and $e_{\foresix{.1}{.5}112221}$.

The coisotropy condition on a generic element $x_1e_{\foresix{.1}{.5}111211}+x_2e_{\foresix{.1}{.5}112211}+x_3e_{\foresix{.1}{.5}111221}+x_4e_{\foresix{.1}{.5}112221}$ of $\mf g_4$ is
\[
4x_1x_4=(x_2+x_3)^2.
\]
It follows that for
\[
E=e_{\foresix{.1}{.5}111211}+e_{\foresix{.1}{.5}112211}-e_{\foresix{.1}{.5}111221}
\]
$f+E$ is a quasi-cyclic element.
Its Jordan decomposition has semisimple part
\begin{align*}
f_1+E=2f_{\foresix{.1}{.5}111100}+2f_{\foresix{.1}{.5}101110}+2f_{\foresix{.1}{.5}010111}+2f_{\foresix{.1}{.5}001111}&+\frac23(f_{\foresix{.1}{.5}111110}-f_{\foresix{.1}{.5}011111})+\frac23e_{\foresix{.1}{.5}111211}+e_{\foresix{.1}{.5}112211}-e_{\foresix{.1}{.5}111221}\\
\intertext{and nilpotent part}
f_2=f_{\foresix{.1}{.5}011210}&-\frac23(f_{\foresix{.1}{.5}111110}-f_{\foresix{.1}{.5}011111})+\frac13e_{\foresix{.1}{.5}111211},
\end{align*}
This is not an integrable triple because $f_2\notin\mf g_{-2}$.
\end{example}

%%%%%%%%%%%%%%%%%%%%%%%%%%%%%%%%%%%%%%%%%%%%%%%%%%%%%%%%%%%%%%%%%%%%%%%
\section{Integrable cyclic and quasi-cyclic elements associated to nilpotent elements of even depth}

\subsection{Integrable cyclic elements for nilpotent elements of even depth}

Let $f$ be a non-zero nilpotent element of even depth $d$ in a simple Lie algebra $\mf g$.
Recall that $f$ is included in an $\mf{sl}_2$-triple $\mf s$, and that, by Theorem \ref{thm:allpolar}, the linear group $Z(\mf s)|\mf g_d$ is polar.
All these linear groups are listed in \cite{EKV13}. Let $C\subseteq\mf g_d$ be a Cartan subspace.
By Theorem \ref{thm:intiffss}, any cyclic element $f+E$, where $E\in C$ is non-zero, is integrable.
Hence, up to conjugation by $Z(\mf s)$,
the integrable cyclic elements are classified by non-zero elements of $C$, up to conjugation by its Weyl group, and rescaling.

Recall also by \cite{EKV13} that the set of non-zero nilpotent elements in $\mf g$ (up to conjugation) is partitioned in bushes,
such that each bush contains a unique nilpotent element $\fs$ of semisimple type,
and all other nilpotent elements in the same bush have the same depth $d$ and the same Cartan subspace.

Below we give a more explicit description of $Z(\mf s)|\mf g_d$ (rather their unity components) for all classical simple Lie algebras $\mf g$.
As in \cite{DSJKV20}, throughout the paper, we use the following notation: $\st(\mf a)$ denotes the standard representation of the Lie algebra $\mf a$,
$\id$ stands for the trivial 1-dimensional representation,
$\oplus$ stands for the direct sum of linear reductive groups, $\operatorname{rank}=\dim\mf g_d\git Z(\mf s)$.

\subsubsection{$\mf g=\mf{sl}_N$, $N\ge2$}\label{sec:zsgl}
Non-zero nilpotent elements $f$, up to conjugation, are pa\-ra\-met\-ri\-zed by partitions
\begin{equation}\label{eq:partition}
\partition p=(p_1^{(r_1)},p_2^{(r_2)},\dots, p_s^{(r_s)}),\qquad N=\sum_ir_ip_i,
\end{equation}
where the $p_i$ are distinct and have multiplicities $r_i$: $p_1>\dots>p_s\ge1,\quad p_1>1$.
Then the associated to a partition $\partition p$ nilpotent element $f=f_{\partition p}$ is of semisimple type if and only if
\begin{equation}\label{eq:ass}
\partition p=(p_1^{(r_1)},1^{(r_2)}).
\end{equation}
The bush containing this partition consists of all partitions with the same $p_1$ and $r_1$.
All these partitions have the same depth $d=2p_1-2$,
and the same $\mf g_d=\Mat_{r_1\times r_1}$, and the action of $Z(\mf s)|\mf g_d$ is the action of $SL_{r_1}$
on $\Mat_{r_1\times r_1}$ by conjugation.
A Cartan subspace $C$ is the subspace of all diagonal matrices.

\subsubsection{$\mf g=\mf{sp}_N$, $N\ge2$ even}
Non-zero nilpotent elements $f_{\partition p}$, up to conjugation, are parametrized by partitions $\partition p$,
whose odd parts have even multiplicity. Then again $f=f_{\partition p}$ is of semisimple type if and only if \eqref{eq:ass} holds.
The bush containing $f_{\partition p}$ consists of all partitions (whose odd parts have even multiplicities) with the same $p_1$ and $r_1$ as ${\partition p}$.
All have the same depth $d=2p_1-2$, and the same linear group $Z(\mf s)|\mf g_d$, which depends on whether $p_1$ is even or odd:

$Z(\mf s)|\mf g_d=\sym^2\st(SO_{r_1})$ if $p_1$ is even,
$Z(\mf s)|\mf g_d=\sym^2\st(Sp_{r_1})=\ad(Sp_{r_1})$ if $p_1$ is odd (then $r_1$ is even).

For the bilinear form with matrix $I$, defining $SO_{r_1}$,
$\sym^2\st(SO_{r_1})$
is identified with the space of all symmetric matrices,
and we can choose for $C$ the subspace, consisting of diagonal matrices,
while $\sym^2\st(Sp_{r_1})$ is the adjoint representation of $Sp_{r_1}$, so that $C$ is any Cartan subalgebra.

\subsubsection{$\mf g=\mf{so}_N$, $N\ge3$, $N\ne4$}\label{subs:socases}
Non-zero nilpotent elements $f_{\partition p}$, up to conjugation, are parametrized by partitions $\partition p$, whose even parts have even multiplicity.
There are five types of elements $f_{\partition p}$ of semisimple type:
\begin{enumerate}[(a)]
\item\label{en:so31} $\partition p=(3,1^{(r_2)})$; $d=2$;
\item\label{en:soo1} $\partition p=(p_1,1^{(r_2)})$, $p_1\ge5$ odd; $d=2p_1-4$;
\item\label{en:soo2} $\partition p=(p_1,p_1-2,1^{(r_3)})$, $p_1\ge5$ odd; $d=2p_1-4$;
\item\label{en:soee} $\partition p=(p_1^{(r_1)},1^{(r_2)})$, $p_1\ge2$ even, $r_1\ge2$ even; $d=2p_1-2$;
\item\label{en:soeo} $\partition p=(p_1^{(r_1)},1^{(r_2)})$, $p_1\ge3$ odd, $r_1\ge2$ even; $d=2p_1-2$.
\end{enumerate}
The linear groups $Z(\mf s)|\mf g_d$ for the types \eqref{en:so31} -- \eqref{en:soeo} are as follows:
\begin{enumerate}[(a)]
\item $\st(SO_{r_2})\oplus\id$;
\item $\id$;
\item $\id\oplus\id$;
\item $\ext^2\st(Sp_{r_1})$;
\item $\ext^2\st(SO_{r_1})=\ad(SO_{r_1})$.
\end{enumerate}

Cartan subspaces are as follows:

\begin{tabular}{ll}
\eqref{en:so31} &$C=\mb Fv\oplus\mb F$, where $\kform vv=1$;\\
\eqref{en:soo1} and \eqref{en:soo2} &$C=\mf g_d$;\\
\eqref{en:soee} &$C=$ the subspace of diagonal matrices in $\mf g_d$\\&\qquad\ if the bilinear form, defining $Sp_{r_1}$, is $\left(\begin{smallmatrix}0&1\\-1&0\end{smallmatrix}\right)$;\\
\eqref{en:soeo}&$C=$ the Cartan subalgebra.
\end{tabular}

Bushes containing these $f$ of semisimple type correspond to the following partitions (with all even parts having even multiplicities):

\begin{tabular}{ll}
\eqref{en:so31} &partition itself;\\
\eqref{en:soo1} &all partitions with the same $p_1$ and $r_1=1$ satisfying $p_2<p_1-2$;\\
\eqref{en:soo2} &all partitions with the same $p_1$, $r_1=1$ and $p_2=p_1-2$;\\
\eqref{en:soee} and \eqref{en:soeo} &all partitions with the same $p_1$ with multiplicity $r_1$ or $r_1+1$.
\end{tabular}

The group $Z(\mf s)|\mf g_d$ is the same for the nilpotent elements from the bush, except for the following two cases:
\begin{itemize}
\item[\eqref{en:soo1}] partitions $(p_1^{(r_1)},...)$, where $r_1\ge1$ is odd, in which case $Z(\mf s)|\mf g_d=\ad(SO_{r_1})$;
\item[\eqref{en:soo2}] partitions $(p_1,(p_1-2)^{(r_2)},...)$, where $p_1\ge3$ is odd, in which case $Z(\mf s)|\mf g_d=\st(SO_{r_2})\oplus\id$.
\end{itemize}

The information about Lie algebra actions of centralizers of $\mf{sl}_2$-triples can be summarized in the following table.

\begin{table}[H]
\resizebox{.8\textwidth}{!}{
\begin{tabular}{c|l|c|c|c}
$\mf g$&\qquad\qquad partition&$d$&rank&$\mf z(\mf s)|\mf g_d$\\
\hline\hline
$\mf{sl}_N$\\
\hline
&$(p_1^{(r_1)},...)$&$2p_1-2$&$r_1$&$\ad(\mf{sl}_{r_1})\oplus\id$\\
\hline\hline
$\mf{sp}_N$\\
\hline
&$(p_1^{(r_1)},...)$, \hfill$p_1$ even&$2p_1-2$&$r_1$&$\sym^2\st(\mf{so}_{r_1})$\\
&$(p_1^{(r_1)},...)$, \hfill$p_1$ odd&$2p_1-2$&$\frac{r_1}2$&$\ad(\mf{sp}_{r_1})$\\
\hline\hline
$\mf{so}_N$\\
\hline
&$(p_1^{(r_1)},...)$,\hfill$p_1$ even&$2p_1-2$&$\frac{r_1}2$&$\ext^2\st(\mf{sp}_{r_1})$\bigstrut[t]\\
&$(p_1^{(r_1)},...)$, \hfill$r_1>1$, $p_1$ odd&$2p_1-2$&$\left[\frac{r_1}2\right]$&$\ad(\mf{so}_{r_1})$\\
&$(p_1,(p_1-1)^{(r_2)},...)$, \hfill$p_1$ odd&$2p_1-3$&$0$&$\st(\mf{sp}_{r_2})$\\
&$(p_1,(p_1-2)^{(r_2)},...)$, \hfill$p_1$ odd&$2p_1-4$&$2$&$\st(\mf{so}_{r_2})\oplus\id$\\
&$(p_1,(p_1-m)^{(r_2)},...)$, \hfill$m>2$, $p_1$ odd &$2p_1-4$&1&$\id$
\end{tabular}
}
\caption{Actions of centralizers of $\mf{sl}_2$-triples of nilpotent elements in simple Lie algebras $\mf g$ of classical types on $\mf g_d$}\label{tab:zsd}
\end{table}

\subsection{Integrable triples in \texorpdfstring{$\mf g=\mf{gl}_N$ and $\mf{sl}_N$}{glNandslN}}\label{sec:glsl}
%

%%%
\subsubsection{Setup and preliminary results}\label{sec:setup}

Let $\mf g=\mf{gl}_N$ or $\mf{sl}_N$.
Let $f$ be a nilpotent element of $\mf g$ in Jordan form and let $\partition p$ be the associated partition \eqref{eq:partition} of $N$,
given by the sizes of the blocks of the Jordan form.
We associate to $\partition p$ a symmetric (with respect to the $y$-axis) pyramid,
with boxes of size $2\times2$ indexed by the set $I=\{1,2,\dots,N\}$ (say starting from right to left
and bottom to top).
For example, for the partition $(9,7,4^{(2)})$ of $24$, we have the pyramid in Figure \ref{fig:pyramid}.

\begin{figure}[H]
\setlength{\unitlength}{0.14in}
% selecting unit length
\centering
% used for centering Figure
\begin{picture}(30,12)
% picture environment with the size (dimensions)
% 32 length units wide, and 15 units high.
\put(5,4){\framebox(2,2){9}}
\put(7,4){\framebox(2,2){8}}
\put(9,4){\framebox(2,2){7}}
\put(11,4){\framebox(2,2){6}}
\put(13,4){\framebox(2,2){5}}
\put(15,4){\framebox(2,2){4}}
\put(17,4){\framebox(2,2){3}}
\put(19,4){\framebox(2,2){2}}
\put(21,4){\framebox(2,2){1}}

\put(7,6){\framebox(2,2){16}}
\put(9,6){\framebox(2,2){15}}
\put(11,6){\framebox(2,2){14}}
\put(13,6){\framebox(2,2){13}}
\put(15,6){\framebox(2,2){12}}
\put(17,6){\framebox(2,2){11}}
\put(19,6){\framebox(2,2){10}}

\put(10,8){\framebox(2,2){20}}
\put(12,8){\framebox(2,2){19}}
\put(14,8){\framebox(2,2){18}}
\put(16,8){\framebox(2,2){17}}

%\put(11,10){\framebox(2,2){(43)}}
%\put(13,10){\framebox(2,2){(42)}}
%\put(15,10){\framebox(2,2){(41)}}

\put(10,10){\framebox(2,2){24}}
\put(12,10){\framebox(2,2){23}}
\put(14,10){\framebox(2,2){22}}
\put(16,10){\framebox(2,2){21}}

\put(4,2){\vector(1,0){22}}
\put(26,1){$x$}

\put(6,1.6){\line(0,1){0.8}}
\put(7,1.8){\line(0,1){0.4}}
\put(8,1.6){\line(0,1){0.8}}
\put(9,1.8){\line(0,1){0.4}}
\put(10,1.6){\line(0,1){0.8}}
\put(11,1.8){\line(0,1){0.4}}
\put(12,1.6){\line(0,1){0.8}}
\put(13,1.8){\line(0,1){0.4}}
\put(14,1.6){\line(0,1){0.8}}
\put(15,1.8){\line(0,1){0.4}}
\put(16,1.6){\line(0,1){0.8}}
\put(17,1.8){\line(0,1){0.4}}
\put(18,1.6){\line(0,1){0.8}}
\put(19,1.8){\line(0,1){0.4}}
\put(20,1.6){\line(0,1){0.8}}
\put(21,1.8){\line(0,1){0.4}}
\put(22,1.6){\line(0,1){0.8}}

\put(13.8,0.6){0}
\put(15.8,0.6){2}
\put(17.8,0.6){4}
\put(19.8,0.6){6}
\put(21.8,0.6){8}
\put(11.6,0.6){-2}
\put(9.6,0.6){-4}
\put(7.6,0.6){-6}
\put(5.6,0.6){-8}

\end{picture}
\caption{Symmetric pyramid for the partition $(9,7,4^{(2)})$ of $24$} % title of the Figure
\label{fig:pyramid}
% label to refer figure in text
\end{figure}
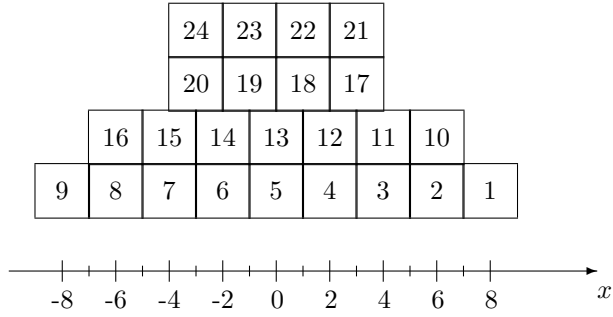
Let $V$ be the $N$-dimensional vector space over $\mb F$ with basis
$\{e_{\alpha}\}_{\alpha\in I}$.
The Lie algebra $\mf g\cong\mf{gl}(V)$ has a basis
consisting of the elementary matrices $E_{\alpha,\beta}$, $\alpha,\beta\in I$.
The elementary matrix $E_{\alpha,\beta}$ in $\mf g$ can be depicted by an arrow
going from the center of the box $\beta$ to the center of the box $\alpha$.
In particular, $f$ is the ``shift to the left'' operator. It is depicted as the sum of all the arrows pointing
from each box to the next one on the left
\begin{equation}\label{eq:f}
f=\sum_{\alpha\leftarrow\beta} E_{\alpha,\beta}
\,,
\end{equation}
where the sum is over all adjacent boxes (on the same row) $\alpha,\beta\in I$.
Let us also denote by $f\transpose$ the transpose of the matrix $f$ defined in \eqref{eq:f}. It is the ``shift to the right'' operator.

Let $h\in\mf g$ be the diagonal endomorphism of $V$
whose eigenvalue on $e_{\alpha}$ is the $x$-coordinate of the center of the box labeled by $\alpha$ (see Figure \ref{fig:pyramid})
which we denote by $x_{\alpha}$.
We then have the corresponding $h$-eigenspace decomposition of $V$
\begin{equation}\label{eq:dech}
V=\bigoplus_{k=-D}^DV[k]
\,,
\quad V[k]=\{v\in V\mid h(v)=kv\}
\,,
\end{equation}
where $D=p_1-1$ is the maximal eigenvalue of $h$.

We note that the elements $f$ and $h$ belong to a $\mf{sl}_2$-triple $\mf s=\{e,h,f\}$,
where $e=\sum_{\alpha\leftarrow\beta} c_\beta E_{\beta,\alpha}$, with $c_\beta=\sum x_\gamma$, where the sum is over all boxes $\gamma$ at the right and in the same row of the box $\beta$, including it.

The elementary matrices $E_{\alpha,\beta}$ are eigenvectors
with respect to the adjoint action of $h$:
$$
(\ad h)E_{\alpha,\beta}=(x_\alpha-x_\beta)E_{\alpha,\beta}
\,.
$$
This defines a $\mb Z$-grading of $\mf g$,
given by the $\ad h$-eigenspaces as in \eqref{eq:grading}:
\begin{equation}\label{eq:adx2}
\mf g_k=\Span_{\mb F}\{E_{\alpha,\beta}\mid x_{\alpha}-x_\beta=k\}\,,\,k\in\mb Z.
\end{equation}
The depth of this grading is $d=2D=2p_1-2$.

Next, consider the subspaces $V_-=\Ker f$ and $V_+=\Ker f\transpose$ of $V$.
We thus have the direct sums decompositions
\begin{equation}\label{eq:complementary}
V=V_-\oplus f\transpose V=V_+\oplus fV
\,.
\end{equation}
Let $D_i=p_i-1$, for $i=1,\dots,s$ (in particular $D_1=D$). Throughout the paper we will use the decompositions
$V_{\pm}=\bigoplus_{i=1}^sV_{\pm,i}$,
where
\begin{equation}\label{eq:vpmi}
V_{-,i}
=
V_-\cap f^{D_i}V_+
\,\,,\,\,\,\,
V_{+,i}
=
V_+\cap (f\transpose)^{D_i}V_-
\,\,,\,\,\,\,
i=1,\dots,s
\,,
\end{equation}
and
$V=\bigoplus_{i=1}^sV_{i}$,
where
\begin{equation}\label{eq:vi}
V_{i}
=
\bigoplus_{k=0}^{D_i}f^kV_{+,i}
\,\,,\,\,\,\,
i=1,\dots,s
\,.
\end{equation}
Representing the basis elements of $V$ as boxes of the pyramid as in Figure \ref{fig:pyramid}, $V_i$ corresponds to the $i$-th rectangle counting from the bottom, and $V_{\pm,i}$ correspond to the
right/left most boxes of the $i$-th rectangle.
With a picture:
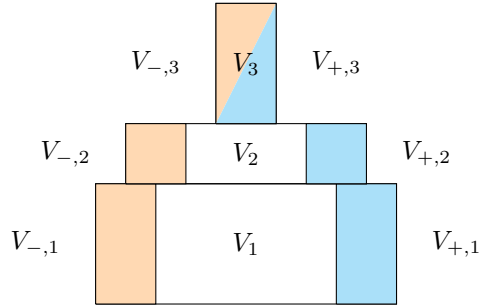
\begin{figure}[H]
\begin{tikzpicture}[scale=0.4]
\draw (0,0)--(10,0)--(10,4)--(0,4)--(0,0);
\draw [fill=orange!30] (0,0)--(2,0)--(2,4)--(0,4)--(0,0);
\draw [fill=cyan!30] (8,0)--(10,0)--(10,4)--(8,4)--(8,0);
\node at (-2,2) {$V_{-,1}$};
\node at (5,2) {$V_1$};
\node at (12,2) {$V_{+,1}$};
\draw (1,4)--(9,4)--(9,6)--(1,6)--(1,4);
\draw [fill=orange!30] (1,4)--(3,4)--(3,6)--(1,6)--(1,4);
\draw [fill=cyan!30] (7,4)--(9,4)--(9,6)--(7,6)--(7,4);
\node at (-1,5) {$V_{-,2}$};
\node at (5,5) {$V_2$};
\node at (11,5) {$V_{+,2}$};
\draw (4,6)--(6,6)--(6,10)--(4,10)--(4,6);
\draw [fill=orange!30] (4,6)--(4,10)--(6,10);
\draw [fill=cyan!30] (4,6)--(6,6)--(6,10);
\node at (2,8) {$V_{-,3}$};
\node at (5,8) {$V_3$};
\node at (8,8) {$V_{+,3}$};
\end{tikzpicture}
\caption{The spaces $V_i$ and $V_{\pm,i}$}
\label{fig:Vspaces}
\end{figure}
For the pyramid in Figure \ref{fig:Vspaces} the subspaces $V_-$ and $V_+$ correspond, respectively, to the boxes colored in orange and blue (note that they may have nontrivial intersection); the subspaces $V_i$, $i=1,2,3$, correspond to the rectangles of the pyramid, and $V_{\pm,i}$ is the intersection of the rectangle $V_i$ with $V_{\pm}$.

Throughout the paper, given a subspace $U\subset V$,
together with a ``natural'' splitting $V=U\oplus W$
(usually associated with the grading of $V$),
we shall denote, with a slight abuse of notation,
by $\id_U$
both the identity map $U\stackrel{\sim}{\longrightarrow}U$,
the inclusion map $U\hookrightarrow V$,
and the projection map (with kernel $W$) $V\twoheadrightarrow U$;
the correct meaning of $\id_U$ should then be clear from the context.

Using the above notation, we clearly have (recall the splitting \eqref{eq:complementary})
\begin{equation}\label{20180219:eq3}
ff\transpose=\id_V-\id_{V_+}=\id_{fV}
\,\,\text{ and }\,\,
f\transpose f=\id_V-\id_{V_-}=\id_{f\transpose V}
\,.
\end{equation}
The following result will be used in the sequel.
\begin{lemma}\label{prop:linalg2}
Let $x\in\mf g_{-2}^f$, the centralizer of $f$ in $\mf g_{-2}$. For every $i=1,\dots,s$ there exists $t_i\in\End(V_{+,1})$ such that
$$
\id_{V_i}x\id_{V_i}=\sum_{k=0}^{D_i-1}f^{k+1}t_i(f\transpose)^k
\,.
$$
\end{lemma}
\begin{proof}
By the direct sum \eqref{eq:vi} we can write $\id_{V_i}x\id_{V_i}=\sum_{k=0}^{D_i-1}x_k$, where $x_k\in\Hom(f^k V_{+,i},f^{k+1}V_{+,i})$.
The condition $[x,f]=0$ gives
\begin{equation}\label{eq:A}
fx_k=x_{k+1}f\in\Hom(f^kV_{+,i},f^{k+2}V_{+,i})\,,
\end{equation}
for every $k=0,\dots,D_i-2$.
Multiplying both sides of \eqref{eq:A} on the right by $f\transpose$ and using \eqref{20180219:eq3}
we get
$$fx_kf\transpose=x_{k+1}\id_{fV}=x_{k+1}\,,$$
for $k=0,\dots,D_i-2$. A recursive solution to these equations is
$$
x_k=f^kx_0(f\transpose)^k\,,
\qquad
k=1,\dots,D_i-1\,.
$$
Letting $t_i=f\transpose x_0\in\End(V_{+,i})$ we get the claim.
\end{proof}

In order to apply, in the following sections, Proposition \ref{prop:linalg} we need the following result.
\begin{lemma}\label{prop:adfinv}
Let $x\in\mf g_{-1}$ and assume that $0\neq x\in\Hom(V_j,V_i)$. Then
\begin{equation}\label{eq:adfinv}
(\ad f)^{-1}x=\left\{
\begin{array}{ll}
\displaystyle{\sum_{k\in\mb Z_{+}}(f\transpose)^{k+1}xf^k\,,
}
&
\displaystyle{i<j\,,}
\\
\displaystyle{-\sum_{k\in\mb Z_{+}}f^{k}x(f\transpose)^{k+1}\,,}
&
\displaystyle{i>j\,.}
\end{array}\right.
\end{equation}
(Note that, since $x\in\mf g_{-1}$, it must be $i\neq j$.)
\end{lemma}
\begin{proof}
Without loss of generality, let us assume that
$x\in\Hom(f^{k_j}V_{+,j},f^{k_i}V_{+,i})$.
Note that, since $x\in\mf{g}_{-1}$, we have
\begin{equation}\label{eq:B}
D_i-2k_i=D_j-2k_j-1\,.
\end{equation}
In particular, if $i<j$, then $D_i>D_j$ and therefore $k_i>k_j$. Conversely, if $i>j$, then
$D_i<D_j$ and $k_i<k_j$. Moreover, $D_j-k_j>D_i-k_i$.

First, let us assume that $i<j$. Applying $\ad f$ to the RHS of equation \eqref{eq:adfinv}
we get
\begin{align*}
&\sum_{k\in\mb Z_{+}}f(f\transpose)^{k+1}xf^k-\sum_{k\in\mb Z_{+}}(f\transpose )^{k+1}xf^{k}f
\\
&=x-\id_{V_+}\sum_{k\in\mb Z_{+}}(f\transpose)^{k}xf^k
=x-(f\transpose)^{k_i}xf^{k_i}\,.
\end{align*}
In the first equality we used the first equation in \eqref{20180219:eq3},
while, for the second equality, the operator $\id_{V_+}$ on the left forces $k=k_i$.
Since $i<j$, by the observation after equation \eqref{eq:B} we have $k_i>k_j$. Hence, $x f^{k_i}=0$ thus proving \eqref{eq:adfinv} in this case.

Next, let us assume that $i>j$ and let us apply $\ad f$ to the RHS of equation \eqref{eq:adfinv}:
\begin{align*}
&-\sum_{k\in\mb Z_{+}}ff^{k}x(f\transpose)^{k+1}
+\sum_{k\in\mb Z_{+}}f^{k}x(f\transpose)^{k+1} f
\\
&=x-\sum_{k\in\mb Z_{+}}f^{k}x(f\transpose)^k\id_{V_-}
=x-f^{D_j-k_j}x(f\transpose)^{D_j-k_j}\,.
\end{align*}
In the first equality we used the second equation in \eqref{20180219:eq3},
while, for the second equality, the operator $\id_{V_-}$ on the right forces $k=D_j-k_j$.
Since $i>j$, by the observation after equation \eqref{eq:B} we have $D_j-k_j>D_i-k_i$. Hence, $f^{D_j-k_j}x=0$. This completes the proof.
\end{proof}
The next result will be used in Sections \ref{sec:5.2.3} and \ref{sec:5.4.2}.
\begin{lemma}\label{prop:20201029a}
Let $U\in\Hom(V[D],V[-D])$ and $E=A+B$, where $A\in\Hom(V[-D+1],V[D])$ and
$B\in\Hom(V[-D],V[D-1])$. Then, $\ad E\circ(\ad f)^{-1}\circ\ad E(U)=0$ if and only if
$$
[Af^{D-1}Bf^D,(f\transpose)^{D}U]=0.
$$
\end{lemma}
\begin{proof}
We have that $[E,U]=UA-BU$. Since $UA\in\Hom(V_2,V_1)$ and $BU\in\Hom(V_1,V_2)$,
using Lemma \ref{prop:adfinv} we have that
\begin{align*}
&\left(\ad E\circ(\ad f)^{-1}\circ\ad E\right)(U)=\sum_{k\in\mb Z_+}[E,
(f\transpose)^{k+1}UAf^{k}+f^{k}BU(f\transpose)^{k+1}]
\\
&=Af^{D-1}BU(f\transpose)^{D}-(f\transpose)^{D}UAf^{D-1}B
\,.
\end{align*}
The claim follows from the fact that $(f\transpose)^{D}:V[-D]\to V[D]$ is an isomorphism with inverse $f^{D}$.
\end{proof}
The next result will be used in Sections \ref{sec:5.2.3} and \ref{subs:sodd}.
\begin{lemma}\label{prop:20201029b}
Let $U=X+Y$, where $X\in\Hom(V[D],V[-D+1])$, $Y\in\Hom(V[D-1],V[-D])$, and let $E=A+B+C$,
where $A\in\Hom(V[-D],V[D-2])$, $B\in\Hom(V[-D+2],V[D])$ and $C\in\Hom(V[-D+1],V[D-1])$. Then, $\ad E\circ(\ad f)^{-1}\circ\ad E(U)=0$ if and only if
\begin{equation}\label{20201029:eq1}
\begin{split}
&(f\transpose)^DYCf^{D-1}C=(f\transpose\id_{fV}A+B\id_{f\transpose V}f\transpose)YC+B\id_{V_-}f^{D-2}\id_{V_+}AY\,,
\\
&Cf^{D-1}CX(f\transpose )^D=CX(f\transpose\id_{fV}A+B\id_{f\transpose V}f\transpose)+(f\transpose)^{D-1}XB\id_{V_-}f^{D-2}\id_{V_+}A\,.
\end{split}
\end{equation}
\end{lemma}
\begin{proof}
Writing $A=\id_{V_+}A+\id_{fV}A$ and
$B=B_{\id_{V_-}}+B\id_{f\transpose V}$,
we have that $[E,U]=\id_{V_+}AY+BX-XB\id_{f\transpose V}+\id_{fV}AY-YB-XB\id_{V_-}$.
Since $\id_{V_+}AY+BX-XB\id_{f\transpose V}\in\oplus_{i=1,2}\Hom(V_{i},V_{i+1})$ and $\id_{fV}AY-YB-XB\id_{V_-}\in\oplus_{i=1,2}\Hom(V_{i+1},V_i)$,
using Lem\-ma \ref{prop:adfinv} we have that
\begin{equation}\label{20201102:eq1}
\begin{split}
&\left((\ad f)^{-1}\circ\ad E\right)(U)=\sum_{k\in\mb Z_+}(f\transpose)^{k+1}(\id_{V_+}AY+CX-XB\id_{f\transpose V})f^k
\\
&-\sum_{k\in\mb Z_+}f^k(\id_{fV}AY-YC-XB\id_{V_-})(f\transpose)^{k+1}
\,.
\end{split}
\end{equation}
Finally, applying $\ad E$ to both side of equation \eqref{20201102:eq1} we get
\begin{equation}\label{20201102:eq2}
\begin{split}
&\left(\ad E\circ (\ad f)^{-1}\circ\ad E\right)(U)=-B\id_{f\transpose V}f\transpose YC-B\id_{V_-}f^{D-1}\id_{V_+}AY(f\transpose)^{D-1}
\\
&-Cf^{D-1}CX(f\transpose)^D+CXB\id_{f\transpose V}f\transpose+(f\transpose)^{D-1}XB\id_{V_-}f^{D-2}\id_{V_+}A
\\
&-f\transpose\id_{fV}AYC+(f\transpose)^DYCf^{D-1}C+CXf\transpose\id_{fV}A
\,.
\end{split}
\end{equation}
Note that
$$
\begin{array}{l}
\displaystyle{
B\id_{f\transpose V}f\transpose YB+B\id_{V_-}f^{D-1}\id_{V_+}AY(f\transpose)^{D-1}
}
\\
\displaystyle{
+f\transpose\id_{fV}AYC-(f\transpose)^DYCf^{D-1}C\in\Hom(V[-D+1],V[D])
}
\end{array}
$$
and
$$
\begin{array}{l}
\displaystyle
{Cf^{D-1}CX(f\transpose)^D-CXB\id_{f\transpose V}f\transpose-(f\transpose)^{D-1}XB\id_{V_-}f^{D-2}\id_{V_+}A
}
\\
\displaystyle{
-CXf\transpose\in\Hom(V[-D],V[D-1])\,.
}
\end{array}
$$
Hence, the RHS of \eqref{20201102:eq2} vanishes if and only if equation \eqref{20201029:eq1} holds.
\end{proof}
Finally, the last result of this section will be used in Section \ref{sec:5.5.2}..
\begin{lemma}\label{prop:20201029c}
Let $U=X+Y$, where $X\in\Hom(V[D],V[-D+2])$,  $Y\in\Hom(V[D-2],V[-D])$, and let $E=A+B$,
where $A\in\Hom(V[-D],V[D-3])$, $B\in\Hom(V[-D+3],V[D])$.
Let us assume that
\begin{equation}\label{20201117:eq2}
\id_{fV}A=0\,,\qquad B\id_{f\transpose V}=0\,.
\end{equation}
Then, $\ad E\circ(\ad f)^{-1}\circ\ad E(U)=0$ if and only if
$$
(f\transpose)^{D-2}XBf^{D-3}A=Bf^{D-3}AY(f\transpose)^{D-2}=0\,.
$$
\end{lemma}
\begin{proof}
Note that $[E,U]=AY-XB$. Using the first equation in \eqref{20201117:eq2} and \eqref{eq:adfinv}
we have that
$$
(\ad f)^{-1}(AY)=-\sum_{k\in\mb Z_+}f^kAY(f\transpose)^{k+1}\,,
$$
while using the second equation in \eqref{20201117:eq2} and \eqref{eq:adfinv} we have that
$$
(\ad f)^{-1}(XB)=-\sum_{k\in\mb Z_+}(f\transpose)^{k+1}XBf^{k}\,.
$$
Hence, by a straightforward computation we get
$$
\ad E\circ(\ad f)^{-1}\ad E(U)=-Bf^{D-3}AY(f\transpose)^{D-2}+(f\transpose)^{D-2}XBf^{D-3}A
\,.
$$
Equation \eqref{20201117:eq2} follows since
$Bf^{D-3}AY(f\transpose)^{D-2}\in\Hom(V[-D+2],V[D])$ and
$(f\transpose)^{D-2}XBf^{D-3}A\in\Hom(V[-D],V[D-2])$.
\end{proof}

%%%
\subsubsection{The centralizer $\mf z(\mf g_{\ge2})$}
From \eqref{eq:adx2} we have that a homogeneous element $E\in\mf g_k$, $k\in\mb Z$, has the form
\begin{equation}\label{eq:E}
E=\sum_{x_\alpha-x_\beta=k}c_{\alpha\beta}E_{\alpha\beta}\,,\quad c_{\alpha\beta}\in\mb F
\,.
\end{equation}
The goal of this section is to describe the centralizer $\mf z(\mf g_{\ge2})$ of
$\mf g_{\ge2}$ in $\mf g$.
\begin{lemma}\label{lemma1}
If $[E,\mf g_{\ge2}]=0$, then $E\in \mf g_d\oplus\mf g_{d-1}\oplus W\oplus\mb F\id_{V}$, where
$$
W=\Hom(V[-D+1],V[D-1])\subset\mf g_{d-2}\,.
$$
\end{lemma}
\begin{proof}
Since the adjoint action of $e\in\mf g_2$ is injective on $\mf g_{<0}$, we obviously have
$\mf z(\mf g_{\ge2})\subset\mf g_{\ge0}$. Let $\tilde\alpha,\tilde\beta\in I$ be such that $x_{\tilde\alpha}=D$ and
$x_{\tilde\beta}\le D-2$ (i.~e. the box $\tilde\beta$ is completely at the left of the box $\tilde\alpha$).
Then $E_{\tilde\alpha\tilde\beta}\in\mf g_{\ge2}$. Hence, letting $E\in\mf g_k$ be as in equation \eqref{eq:E}, we have
\begin{equation}\label{eq1}
0=[E_{\tilde\alpha\tilde\beta},E]=\sum_{x_{\alpha}-x_{\beta}=k}c_{\alpha\beta}[E_{\tilde\alpha\tilde\beta},E_{\alpha\beta}]
=\sum_{x_{\tilde\beta}-x_{\beta}=k}c_{\tilde\beta\beta}E_{\tilde\alpha\beta}
-\sum_{x_{\alpha}-D=k}c_{\alpha\tilde\alpha}E_{\alpha\tilde\beta}\,.
\end{equation}
If $k\ge1$, then the condition $x_{\alpha}-D=k\ge1$ implies that $x_{\alpha}\ge1+D$ thus the second sum
in \eqref{eq1} is empty. Hence, from equation \eqref{eq1} we have that $c_{\alpha\beta}=0$ if $x_{\alpha}\le D-2$.
If $k=0$, then equation \eqref{eq1} becomes
$$
0=\sum_{x_{\tilde\beta}=x_{\beta}}c_{\tilde\beta\beta}E_{\tilde\alpha\beta}
-\sum_{x_{\alpha}=D}c_{\alpha\tilde\alpha}E_{\alpha\tilde\beta}\,,
$$
from which follows that, for $x_\alpha,x_\beta\neq D-1$, we have $c_{\alpha,\beta}=\delta_{\alpha\beta}\lambda$, for some
$\lambda\in\mb F$.
Similarly, letting $\tilde\alpha,\tilde\beta\in I$ be such that $x_{\tilde \alpha}=-D$ and $x_{\tilde\beta}\ge-D+2$
(i.~e. the box $\tilde\beta$ is completely at the right of the box $\tilde\alpha$), the
condition $[E,E_{\tilde\alpha\tilde\beta}]=0$ implies, for $k\ge1$, that $c_{\alpha\beta}=0$ if $x_{\beta}\ge -D+2$
and, for $k=0$, that $c_{\alpha,\beta}=\delta_{\alpha\beta}\lambda$ for $x_\alpha,x_\beta\neq -D+1$.
This proves that
$$
E\in\Span_{\mb F}\{\id_V,E_{\alpha\beta}\mid x_{\alpha}\ge D-1,x_{\beta}\le-D+1\}=\mf g_d\oplus\mf g_{d-1}\oplus W
\oplus\mb F \id_V
\,.
$$
\end{proof}
\begin{proposition}\label{lemma2}
The centralizer of $\mf g_{\ge2}$ in $\mf g$ is $\mf z(\mf g_{\ge2})=\mf g_d\oplus\mf g_{d-1}\oplus W\oplus\mb F\id_{V}$.
\end{proposition}
\begin{proof}
By degree consideration and the fact that the identity is a central element we clearly have
$\mf g_d\oplus\mf g_{d-1}\oplus\mb F\id_V\subset\mf z(\mf g_{\ge2})$.
Let $T\in\Hom(V[-D+1],V[D-1])$, and let $E_{\alpha\beta}\in\mf g_{\ge2}$, with $x_\alpha-x_\beta\ge2$.
In particular, $x_\alpha\ge x_\beta+2>-D+1$, and $x_\beta\le x_\alpha-2<D-1$.
Hence, $\im(E_{\alpha\beta})\cap V[-D+1]=0$, so that $TE_{\alpha\beta}=0$. Similarly,
$\im(T)\subset\ker(E_{\alpha\beta})$, so that $E_{\alpha\beta}T=0$. Hence, $T\in\mf z(\mf g_{\ge2})$.
This, combined to Lemma \ref{lemma1}, completes the proof.
\end{proof}
Recalling the definition of integrable triples given in the Introduction, if
$E$ is an integrable element for $f$, then
$E\in\left(W\oplus\mf g_{d-1}\oplus\mf g_d\right)\cap\mf g_{\ge1}$ (note that $\mf g_{d-1}\oplus\mf g_d\subset\mf g_{\ge1}$ for $p_1\ge2$, and $W\subset\mf g_{\ge1}$ for $p_1\ge3$).

%%%
\subsubsection{The coisotropy condition}\label{sec:5.2.3}

%%%
An element $E\in\mf g_{d-1}$ can be uniquely decomposed as $E=a(f\transpose)^{D-1}+b(f\transpose)^D$, where
\begin{equation}\label{eq:aandb}
a\in\Hom(V[D-1],V[D])\quad\text{ and }\quad b\in\Hom(V[D],V[D-1]).
\end{equation}
\begin{proposition}\label{prop3}
Let $E=a(f\transpose)^{D-1}+b(f\transpose)^D\in\mf g_{d-1}$. The subspace $\mf g_1^E$ is
co\-i\-so\-tro\-pic with respect to the bilinear form \eqref{eq:skewform} if and only if
\begin{equation}
ab=\lambda\id_{V[D]}
\,,
\quad\lambda\in\mb F
\,.
\end{equation}
\end{proposition}
\begin{proof}
By Lemma \ref{prop:20201029a} with $A=a(f\transpose)^{D-1}$ and $B=b(f\transpose)^D$ we have that
$\ad E\circ(\ad f)^{-1}\circ\ad E|_{\mf g_{-d}}=0$ if and only if
\begin{equation}\label{20200717:eq2}
[ab,(f\transpose)^{D}U]=0\,,
\end{equation}
for every $U\in\mf g_{-d}$. Note that $(f\transpose)^{D}\mf g_{-d}\cong\End(V[D])$. Hence, equation \eqref{20200717:eq2}
implies that $ab$ is a scalar. The claim follows from Proposition \ref{prop:linalg}.
\end{proof}
%

%%%
An element $E\in W\subset\mf g_{d-2}$ can be uniquely written as $E=c(f\transpose)^{D-1}$, where
\begin{equation}\label{eq:c}
c\in\End(V[D-1]).
\end{equation}
\begin{proposition}\label{prop4b}
Let $E=c(f\transpose)^{D-1}\in W$. The subspace $\mf g_1^E$ is
coisotropic with respect to the bilinear form \eqref{eq:skewform} if and only if
$c^2=0$.
\end{proposition}
\begin{proof}
By Lemma \ref{prop:20201029b} with $A=B=0$ and $C=c(f\transpose)^{D-1}$ we have that
$\ad E\circ(\ad f)^{-1}\circ\ad E|_{\mf g_{-d+1}}=0$ if and only if $c^2(f\transpose)^{D-1}X(f\transpose)^D=0$, for every $X\in\Hom(V[D],V[-D+1])$. Since $(f\transpose)^D:V[-D]\to V[D]$
and
$(f\transpose)^{D-1}:V[-D+1]\to V[D-1]$ are isomorphisms, this condition is the same as $c^2=0$. The claim follows from Proposition \ref{prop:linalg}.
\end{proof}
%

%%%
\subsubsection{Integrable $E\in\mf g_d$}\label{sec:gd-gln}
%%%
Let $E\in\mf g_d$.
In this section we will use the decomposition (see \eqref{eq:vi})
$$
V=V_{1}\oplus V_{\ge2}\,,
\quad\text{where }V_{\ge2}=\bigoplus_{i=2}^sV_{i}\,.
$$

Note that $(f+E)\id_{V_{\ge2}}=f\id_{V_{\ge2}}$. Hence, $(f+E)\id_{V_{\ge2}}$ is nilpotent.
Furthermore, note that $E$ can be uniquely written as $E=u(f\transpose)^D$, where $u\in\End(V[D])$.
Let
$$
U=\bigoplus_{k=0}^{\infty}f^k\im u\subset V_1\,,\quad
\tilde U=\bigoplus_{k=0}^{\infty}f^k\Ker u
\subset V_1\,.
$$
\begin{lemma}\label{prop1}
Let $E=u(f\transpose)^D\in\mf g_d$. If $u$ is semisimple, then the nilpotent part of $f+E$ is
$(f+E)_n=f\id_{\widetilde U\oplus V_{\ge2}}$.
\end{lemma}
\begin{proof}
Recall that $V_{+,1}=V[D]$ (see \eqref{eq:vpmi} and Figure \ref{fig:Vspaces}).
Since $u$ is semisimple, we have that $V[D]=\Ker u\oplus \im u$. This implies, by \eqref{eq:vi},
$V_{1}=U\oplus\tilde U$.

Clearly, $(f+E)\id_{\tilde U\oplus V_{\ge2}}=f\id_{\tilde U\oplus V_{\ge2}}$
and $(f+E)U\subset U$.
Denote $A=(f+E)\id_U$. Clearly, Since $A$ commutes with $f\id_{\tilde U\oplus V_{\ge2}}$
and $f+E=A+f\id_{\tilde U\oplus V_{\ge2}}$.
We claim that $A$ is semisimple so that the semisimple part of $f+E$ is $(f+E)_s=A$ and
$(f+E)_n=f\id_{\widetilde U\oplus V_{\ge2}}$.
Let $q(x)$ be the minimal polynomial of $u\id_{\im u}$ which has
distinct non-zero roots in $\mb F$, since $u$ is semisimple.
Let $\tilde q(x)=q(x^{p_1})$ which has also distinct roots.
Note that, if $v\in\im u$, then $A^{kp_1}f^hv=f^hu^kv$, for every $h,k\in\mb Z_+$.
Hence, $\tilde q(A)f^hv=q(A^{p_1})f^hv=f^hq(u)v=0$. This implies that the minimal polynomial of $A$ divides
$\tilde q(x)$. Since $\tilde q(x)$ has distinct roots, $A$ is semisimple, as claimed.
\end{proof}
\begin{lemma}\label{prop2}
If $E=u(f\transpose)^D\in\mf g_d$ is integrable for $f$, then $u$ is semisimple.
\end{lemma}
\begin{proof}
Note that $f\id_{V_{\ge2}}$ is nilpotent and commutes with $f\id_{V_1}+E=(f+E)\id_{V_1}$. Hence,
for the Jordan decomposition of $f+E$ we have
$$
(f+E)_s=(f\id_{V_1}+E)_s\quad
\text{and}\quad
(f+E)_n=(f\id_{V_1}+E)_n+f\id_{V_{\ge2}}\,.
$$
Since, by assumptions, $f+E$ is integrable, we then have
$$
(f\id_{V_1}+E)_s=f_1+E\quad\text{and}\quad
(f\id_{V_1}+E)_n=\tilde f_2\in\mf g_{-2}^f\cap\End(V_1)\,,
$$
with $f_1+\tilde f_2=f\id_{V_1}$.
By Lemma \ref{prop:linalg2}, we have that
\begin{equation}\label{20200810:eq1}
\tilde f_2=\sum_{k=0}^{D-1}f^{k+1}t(f\transpose)^k\,,
\end{equation}
for some $t\in\End(V[D])$.
Moreover, by Definition (I\ref{def:intrip3}) we have that $[E,\tilde f_2]=E\tilde f_2-\tilde f_2E=0$. Since $E\tilde f_2\in\Hom(V[-D+2],V[D])$
and $\tilde f_2E\in\Hom(V[-D],V[D-2])$, we have that $E\tilde f_2=\tilde f_2E=0$. Explicitly, using equation \eqref{20200810:eq1} we get
$E\tilde f_2=Ef^{D}t(f\transpose)^{D-1}=ut(f\transpose)^{D-1}=0$,
which implies
\begin{equation}\label{20200810:eq2}
ut=0
\,.
\end{equation}
Let  $A=f_1+E$ and let $v\in V[D]$.
We have, by equation \eqref{20200810:eq1}, $Av=(f+E-\tilde f_2)v=fv-\tilde f_2v=f(1-t)v$,
and repeating the same computation $k$ times,
\begin{equation}\label{20200810:eq4}
A^kv=f^k(1-t)^kv\,,\quad0\le k\le D\,.
\end{equation}
Letting $k=D=p_1-1$ in \eqref{20200810:eq4}, and applying $A$ one more time, we get
\begin{equation}\label{20200810:eq5}
A^{p_1}v=(f_1+E)f^{D}(1-t)^{D}v=u(1-t)^{D}v=uv
\,.
\end{equation}
For the last equality we used \eqref{20200810:eq2}.
Since $A$ is semisimple, $A^{p_1}$ is semisimple as well. Moreover, from equation
\eqref{20200810:eq5} we have that $A^{p_1}V[D]\subset V[D]$ and $A^{p_1}\id_{V[D]}=u$. As a consequence,
$u$ is semisimple, proving the claim.
\end{proof}
Combining Lemmas \ref{prop1} and \ref{prop2} we get the following result.
\begin{proposition}\label{thm:glintiffss}
$E=u(f\transpose)^D\in\mf g_{d}$ is integrable for $f$ if and only if $u$ is semisimple.
\end{proposition}

\begin{remark}\label{cor:clifffss}
Proposition \ref{thm:glintiffss} is in accordance with Theorem \ref{thm:intiffss} and the results of Section
\ref{sec:zsgl}: it is well-known that the closed orbits for the action of $SL_{r_1}$ on $\Mat_{r_1\times r_1}$
by conjugation are indeed the semisimple elements of $\Mat_{r_1\times r_1}$.
\end{remark}

%%%
\subsubsection{Integrable $E\in\mf g_{d-1}$}
In this subsection we will use the decomposition (see \eqref{eq:vi}) $V=V_{\le2}\oplus V_{\ge3}$,
where $V_{\le2}=V_1\oplus V_2$ and $V_{\ge3}=\bigoplus_{i=3}^s V_i$.

Let, as in \eqref{eq:aandb}, $E=a(f\transpose)^{D-1}+b(f\transpose)^D\in\mf g_{d-1}$, where
$a\in\Hom(V[D-1],V[D])$ and $b\in\Hom(V[D],V[D-1])$.
Since $V[D]\subset V_1$ and $V[D-1]\subset V_2$, we have $(f+E)\id_{V_{\ge3}}=f\id_{V_{\ge3}}$, which is
nilpotent.
Let us also denote by
$$
U=V_1\oplus\left(\bigoplus_{k=0}^{\infty}f^k\im b \right)\subset V_{\le2}
\quad\text{and}\quad
\tilde U=\bigoplus_{k=0}^{\infty}f^k\Ker a\subset V_2
\,.
$$
\begin{proposition}\label{prop4}
An element $E=a(f\transpose)^{D-1}+b(f\transpose)^D\in\mf g_{d-1}$ is
integrable for $f$ if and only if $ab=\lambda\id_{V[D]}$, $\lambda\neq 0$.
\end{proposition}
\begin{proof}
First, assume that $ab=\lambda\id_{V[D]}$, with $\lambda\neq0$.
Let $u\in V[D-1]=V_{+,2}$ (see \eqref{eq:vpmi} and Figure \ref{fig:Vspaces}). Clearly, $u=\frac{1}{\lambda}bau+(u-\frac{1}{\lambda}bau)\in\im b+\Ker a$.
Moreover, if $u\in\im b\cap\ker a$, then $au$=0 and $u=b w$, for some $w\in V[D]$, so that $0=au=abw=\lambda w$
which implies $w=0$, and hence $u=0$. This shows that $V[D-1]=\im b\oplus\Ker a$ and, by \eqref{eq:vi}, that
$V_{\le2}=U\oplus\tilde U$.
Clearly, $f+E=(f+E)\id_U+(f+E)\id_{\tilde U\oplus V_{\ge3}}$.
If $v\in V[D]\subset U$, then
\begin{align*}
&(f+E)^Dv=f^Dv\,,\quad (f+E)^{D+1}v=Ef^Dv=b(f\transpose)^Df^Dv=bv\,,
\\
&(f+E)^{2D}v=f^{D-1}bv\,,\quad (f+E)^{2D+1}v=Ef^{D-1}bv=a(f\transpose)^{D-1}f^{D-1}v=abv=\lambda v\,.
\end{align*}
Similarly, one can check that for every $v\in U$ we have $(f+E)^{2D+1}v=\lambda v$.
Since $\lambda\neq0$, $(f+E)\id_U=f\id_U+E$ is semisimple. Furthermore, we clearly have
$$(f+E)\id_{\tilde U\oplus V_{\ge3}}=f\id_{\tilde U\oplus V_{\ge3}}\,,$$
which is nilpotent and commutes with $(f+E)\id_U$.
Hence, the Jordan decomposition of $f+E$ is
$$
(f+E)_s=f\id_U+E\,,\quad
(f+E)_n=f\id_{\tilde U\oplus V_{\ge3}}
\,,
$$
which omplies that $f+E$ is integrable.

Conversely, if $f+E$ is integrable, in particular $\mf g_1^E$ must be coisotropic. Hence, by Proposition
\ref{prop3}, $ab=\lambda\id_{V[D]}$, for some $\lambda \in\mb F$. On the other hand, it must be $\lambda\neq0$
otherwise, as one can easily check, $f+E$ is nilpotent.
\end{proof}
The next result will be used in Section \ref{sec:4.5.1}. We state and prove it here since we need the notation introduced here.
\begin{lemma}\label{prop5}
Let $E=a(f\transpose)^{D-1}+b(f\transpose)^D\in\mf g_{d-1}$. Assume that $f+E$ is not nilpotent and that its nilpotent part lies in $\mf g_{-2}$.
Then $ab\in\End(V[D])$ is a non-zero semisimple element.
\end{lemma}
\begin{proof}
Note that $f\id_{V_{\ge3}}$ is nilpotent and commutes with $f\id_{V_{\le2}}+E=(f+E)\id_{V_{\le2}}$. Hence,
for the Jordan decomposition of $f+E$ we have
$$
(f+E)_s=(f\id_{V_{\le2}}+E)_s\quad
\text{and}\quad
(f+E)_n=(f\id_{V_{\le2}}+E)_n+f\id_{V_{\ge3}}\,.
$$
By assumption, $\tilde f_2=(f\id_{V_{\le2}}+E)_n\in \mf g_{-2}$. In particular,
$\tilde f_2\in\mf g_{-2}^f\cap(\End(V_1)\oplus\End(V_2))$.
Hence, by Lemma \ref{prop:linalg2}, we have that
\begin{equation}\label{20200810:eq1b}
\tilde f_2=\sum_{k=0}^{D-1}f^{k+1}(t+s)(f\transpose)^k\,,
\end{equation}
for some $t\in\End(V[D])$ and $s\in\End(V[D-1])$.
Moreover, by Definition (I\ref{def:intrip3}) we have that $[E,\tilde f_2]=\left(a(f\transpose)^{D-1}+b(f\transpose)^D\right)\tilde f_2-\tilde f_2\left(a(f\transpose)^{D-1}+b(f\transpose)^D\right))=0$. Since
$a(f\transpose)^{D-1}\tilde f_2\in\Hom(V[-D+3],V[D])$, $b(f\transpose)^D\tilde f_2\in\Hom(V[-D+2],V[D-1])$, $\tilde f_2a\in\Hom(V[D-1],V[D-2])$
and $\tilde f_2b\in\Hom(V[D],V[D-3])$, we have $a(f\transpose)^{D-1}\tilde f_2=b(f\transpose)^D\tilde f_2=\tilde f_2a=\tilde f_2b=0$. Explicitly, using equation \eqref{20200810:eq1b} we get
$b(f\transpose)^D\tilde f_2=bt(f\transpose)^{D-1}=0$,
which implies
\begin{equation}\label{20200810:eq2b}
bt=0
\,.
\end{equation}
Moreover, using again equation \eqref{20200810:eq1b} we get
$\tilde f_2b=fsb=0$,
which implies
\begin{equation}\label{20200810:eq2c}
sb=0
\,.
\end{equation}
Next, let $A=(f+E)_s$ and let $v\in V[D]$. We have,
by equation \eqref{20200810:eq1b},
$$Av=(f+E-\tilde f_2)v=fv-\tilde f_2v=f(1-t)v\,.$$
Repeating the same computation $D$ times, we get (cf. equation \eqref{20200810:eq4}) $A^Dv=f^Dv$,
and applying $A$ one more time we get
\begin{equation}\label{20200810:eq5c}
A^{D+1}v=Ef^{D}(1-t)^{D}v=b(1-t)^{D}v=bv
\,.
\end{equation}
For the last equality we used \eqref{20200810:eq2b}.
By equations \eqref{20200810:eq1b} and \eqref{20200810:eq5c} we have
$$
A^{D+2}v=(f+E-\tilde f_2)bv=fbv-fsbv=f(1-s)bv=fbv
\,,
$$
where in the last equality we used \eqref{20200810:eq2c}.
Repeating the same computation $D-1$ times we get
$$
A^{2D}v=f^{D-1}bv
\,.
$$
Applying $A$ again, we finally get
\begin{equation}\label{20200810:eq7}
A^{2D+1}v=Ef^{D-1}bv=abv
\,.
\end{equation}
Since $A$ is non-zero semisimple, $A^{2D+1}$ is non-zero semisimple as well. Moreover, from equation
\eqref{20200810:eq7} we have that $A^{2D+1}V[D]\subset V[D]$ and $A^{2D+1}\id_{V[D]}=ab$ thus showing that
$ab$ is non-zero semisimple and concluding the proof.
\end{proof}
%

%%%
\subsubsection{No integrable elements for $f$ in $W$}
In this subsection we will use the decomposition
\begin{equation}\label{20201203:eq1}
V=V_2\oplus V_{\neq2}
\,,
\quad
\text{where }
V_{\neq2}=\bigoplus_{i\neq 2}V_i
\,.
\end{equation}
\begin{proposition}\label{cor:nointinW-b}
There are no integrable elements for $f$ in $W$.
\end{proposition}
\begin{proof}
By contradiction, let $E=c(f\transpose)^{D-1}\in W$, where $c$ is as in \eqref{eq:c},
be integrable for $f$. Since $\mf g_1^E$ is coisotropic, by Proposition \ref{prop4b} we have $c^2=0$.
Clearly, $f+E=(f+E)\id_{V_2}+(f+E)\id_{V_{\neq2}}$, and $f+E$ preserves the direct sum decomposition \eqref{20201203:eq1}.
Note that $(f+E)\id_{V_{\neq2}}=f\id_{V_{\neq2}}$ which is nilpotent.
On the other hand, it is not difficult to check that if $v\in V_2$, then we have $(f+E)^{2D}v=c^2v=0$. Hence, $(f+E)\id_{V_2}$ is nilpotent as well.
This proves that $f+E$ is nilpotent, contradicting the fact that it is integrable.
\end{proof}
As a consequence of Propositions \ref{thm:glintiffss}, \ref{prop4} and \ref{cor:nointinW-b} we get the following.
\begin{corollary}\label{cor:nointinW}
If $E\in\mf g_k$, $k\ge1$, is an integrable element for $f\in\mf{gl}_N$ or $\mf{sl}_N$, then $k=d$ or $k=d-1$.
In other words, $f+E$ is integrable if and only if it is an integrable cyclic or quasicyclic element.
\end{corollary}

\begin{remark}
For $\mf g=\mf{gl}_N$, the triple $(0,f,\id_V)$ satisfies Definition (I\ref{def:intrip1}),(I\ref{def:intrip2}) and (I\ref{def:intrip3}). However,
$\id_V\in\mf g_0$.
\end{remark}

\subsubsection{} Here we reformulate the results on quasi-cyclic elements in $\mf g=\mf{gl}_N$ or $\mf{sl}_N$ in terms of polar linear groups.
First, note that $\mf g_{d-1}$ is naturally identified with the space
\begin{equation}\label{eq:hom}
\Hom(V[D-1],V[D])\oplus\Hom(V[D],V[D-1])
\end{equation}
(recall equation \eqref{eq:aandb}) with the action of the group $Z(\mf s)$ defined by the natural action of $GL_{r_2}\times GL_{r_1}$
(recall that $\dim V[D]=r_1$ and $\dim V[D-1]=r_2$). This linear group is polar since it is a theta group.

Proposition \ref{prop4} says that an element $E=\ph\oplus\psi$ from \eqref{eq:hom} is integrable if and only if
\begin{equation}\label{eq:phipsi}
\ph\circ\psi=\lambda\id_{r_1},\quad\lambda\in\mb F,\quad\lambda\ne0.
\end{equation}
It follows that for existence of an integrable $E$ in $\mf g_{d-1}$ it is necessary and sufficient that
\begin{equation}\label{eq:r1ler2}
r_1\le r_2.
\end{equation}

Suppose that condition \eqref{eq:r1ler2} holds.
Then for the polar linear group $GL_{r_2}\times GL_{r_1}$ acting on \eqref{eq:hom} one can choose a Cartan subspace $C$,
consisting of the following matrices (in some bases of $V[D-1]$ and $V[D]$);
\[
C=\left\{(A\ 0)\oplus(A\ 0)\transpose\mid A=\diag(\lambda_1,...,\lambda_{r_1})\in\Mat_{r_1\times r_1}\right\}.
\]
Then condition \eqref{eq:phipsi} on an element from $C$ means that $\lambda_1^2=...=\lambda_{r_1}^2=\lambda^2$,
hence $\lambda_j=\pm\lambda$ for all $j$. Since the Weyl group of $C$ contains all sign changes of diagonal elements,
up to the action of the Weyl group and rescaling, $C$ contains a unique integrable element $(\id_{r_1}\ 0)\oplus(\id_{r_1}\ 0)\transpose$.
Due to Theorem \ref{thm:quasintiffss} \eqref{thm:quasintiffsseven}, any integrable $E\in\mf g_{d-1}$ is $Z(\mf s)$-conjugate to $C$.
Thus we obtain the following theorem.

\begin{theorem}
Let $\mf g=\mf{gl}_N$ or $\mf{sl}_N$ and let $f\in\mf g$ be a non-zero nilpotent element of depth $d$, corresponding to the partition \eqref{eq:partition}.
\begin{enumerate}[(a)]
\item If there exists an integrable $E\in\mf g_j$, $j\ge1$, for the nilpotent element $f$, then $j=d$ or $d-1$.
\item The linear group $Z(\mf s)|\mf g_{d-1}$ is polar.
\item The element $f$ admits an integrable element $E\in\mf g_{d-1}$ if and only if $r_1\le r_2$.
Provided that this condition holds, there exists a unique, up to equivalence, integrable element $E\in\mf g_{d-1}$ for $f$.
\end{enumerate}
\end{theorem}

\begin{remark}
It follows from \cite{EKV13} and the above discussion that for a nilpotent element $f\in\mf{sl}_N$ there exists $E\in\mf g_j$,
such that the element $f+E$ is semisimple if and only if $j=d$ and $p_2=1$, or $j=d-1$ and $r_1=r_2$ and $p_3=1$.
This claim was stated in \cite{DSKV13}, where the associated integrable Hamiltonian systems were also discussed.
\end{remark}

%%%%%%%%%%%%%%%%%%%%%%%%%%%%%%%%%%%%%%%%%%%%%%%%%
\subsection{General setup for symplectic and orthogonal Lie algebras}\label{sec:setupBCD}
%

%%%
Recall from \cite{CMG93} that nilpotent orbits of $\mf{sp}_N$ (respectively $\mf{so}_N$)
are in one-to-one correspondence with partitions $\partition p$ of $N$
as in \eqref{eq:partition} with the property that if $p_a$ is odd (respectively even), then $r_a$ is even, $1\le a\le s$.

Let $V$ be the $N$-dimensional vector space over $\mb F$ with basis
$\{e_{\alpha}\}_{\alpha\in I}$, where
$I$ is an index set for the basis, which can be identified with the set of boxes in the pyramid associated to
$\partition p$ (cf. Fig. \ref{fig:pyramid}).
Given $\alpha\in I$ we let $\alpha'\in I$ correspond to the box in the same rectangle as $\alpha$ reflected
with respect to the center of the rectangle. For example, in Fig. \ref{fig:pyramid}, if $\alpha$ is the box
labelled by $17$, then $\alpha'$ is the box labelled by $24$, while if $\alpha$ is the box labelled by $23$,
then $\alpha'$ is the box labelled by $18$.

Clearly, $\alpha''=\alpha$. Let $\eta=\pm1$, and choose a map $\epsilon:I\to\{\pm1\}$
with the properties that
\begin{equation}\label{20201001:eq4}
\epsilon_\alpha\epsilon_{\alpha'}=\eta\,,
\qquad
\epsilon_{\alpha}\epsilon_{\beta}=-1\,,
\qquad
\text{for }\alpha,\beta \text{ adjacent boxes in the same row}\,.
\end{equation}
For example, let $\nu:I\to\{1,2,\dots,N\}$ be the ordering of the boxes of the pyramid going from right
to left and then from bottom to top. Then $\epsilon_\alpha=(-1)^{\nu(\alpha)}$ satisfies
the properties in \eqref{20201001:eq4} with $\eta=(-1)^{N+1}$. It is easy to see that for every choice of
$\eta$ such a map $\epsilon:I\to\{\pm1\}$ exists.

Let us define on $V$ a bilinear form $\langle \cdot\mid\cdot\rangle:V\otimes V\to\mb F$ letting
\begin{equation}\label{20200819:eq3}
\langle e_\alpha|e_\beta\rangle=\epsilon_{\alpha}\delta_{\alpha,\beta'}
\,,\qquad\alpha,\beta\in I
\,.
\end{equation}
\begin{lemma}\label{20200819:lem2}
The bilinear form \eqref{20200819:eq3} satisfies ($v,w\in V$)
$$
\langle v|w\rangle=\eta\langle w|v\rangle
\,.
$$
\end{lemma}
\begin{proof}
From equations \eqref{20201001:eq4} and \eqref{20200819:eq3} we have
$$
\langle e_\beta|e_\alpha\rangle=\epsilon_\beta\delta_{\alpha',\beta}=\epsilon_{\alpha'}\delta_{\alpha',\beta}
=\eta\epsilon_{\alpha}\delta_{\alpha,\beta'}=\eta\langle e_\alpha|e_\beta\rangle\,,
\qquad \alpha,\beta\in I
\,.
$$
\end{proof}
Given $A\in\End V$, let us denote by $A^\dagger$ its adjoint with respect to the bilinear form \eqref{20200819:eq3}. By Lemma
\ref{20200819:lem2} we have that
\begin{equation}\label{20201212:eq1}
\mf g=\{A\in\End V\mid A=-A^{\dagger}\}\simeq\mf{sp}_N\,(\text{resp. }\mf{so}_N)\,, \eta=-1\,
\,(\text{resp. }\eta=1)
\,.
\end{equation}
We denote by $E_{\alpha,\beta}\in\End V$, $\alpha,\beta\in I$, the elementary matrices:
$E_{\alpha,\beta}e_\gamma=\delta_{\beta,\gamma}e_\alpha$.
\begin{lemma}\label{20200819:lem3}
We have that $(E_{\alpha,\beta})^\dagger=\epsilon_\alpha\epsilon_\beta E_{\beta',\alpha'}$, $\alpha,\beta\in I$.
\end{lemma}
\begin{proof}
By a straightforward computation we get ($\alpha,\beta,\gamma,\eta\in I$)
\begin{align*}
\langle E_{\alpha,\beta}e_\gamma|e_\eta\rangle &=\delta_{\beta,\gamma}\langle e_{\alpha}|e_{\eta}\rangle
=\epsilon_{\alpha}\delta_{\beta,\gamma}\delta_{\alpha,\eta'}
=\epsilon_{\alpha}\epsilon_\beta\epsilon_\gamma\delta_{\beta,\gamma}\delta_{\alpha,\eta'}
\\
&=\epsilon_{\alpha}\epsilon_\beta\delta_{\alpha,\eta'}\langle e_\gamma|e_{\beta'}\rangle
=\epsilon_{\alpha}\epsilon_\beta\langle e_\gamma|E_{\beta',\alpha'}e_{\eta}\rangle
\,.
\end{align*}
The claim follows.
\end{proof}
For every $\alpha,\beta\in I$, we let
\begin{equation}\label{20200819:eq4}
F_{\alpha,\beta}=E_{\alpha,\beta}-\epsilon_{\alpha}\epsilon_{\beta}E_{\beta',\alpha'}\in\mf g\,.
\end{equation}
Note also that
\begin{equation}\label{20200819:eq4b}
F_{\beta',\alpha'}=-\epsilon_{\alpha}\epsilon_{\beta}F_{\alpha,\beta}\,,
\qquad
\alpha,\beta\in I
\,.
\end{equation}
\begin{lemma}\label{20200819:lem4}
Let $A=\sum_{\alpha,\beta\in I}a_{\alpha,\beta}E_{\alpha,\beta}\in\End V$.
\begin{enumerate}[(a)]
\item $A\in\mf g$ if and only if $a_{\beta',\alpha'}=-\epsilon_{\alpha}\epsilon_{\beta}a_{\alpha,\beta}$, for every
$\alpha,\beta\in I$.
\item If $A\in\mf g$, then
$$
A=\frac12\sum_{\alpha,\beta}a_{\alpha,\beta}F_{\alpha,\beta}
\,.
$$
\end{enumerate}
\end{lemma}
\begin{proof}
Straightforward.
\end{proof}
The following commutation relations hold ($\alpha,\beta,\gamma,\eta\in I$):
\begin{equation}\label{comm:BC}
[F_{\alpha,\beta},F_{\gamma,\eta}]=\delta_{\gamma,\beta}F_{\alpha,\eta}
-\delta_{\eta,\alpha}F_{\beta,\gamma}-\epsilon_\alpha\epsilon_\beta\delta_{\alpha',\gamma}F_{\beta',\eta}
+\epsilon_\alpha\epsilon_\beta\delta_{\eta,\beta'}F_{\gamma,\alpha'}
\,.
\end{equation}
If we depict, as in Section \ref{sec:setup}, the basis elements of $V$ as boxes of a symmetric pyramid
associated to the partition $\partition p$ (cf. Fig. \ref{fig:pyramid}), let, as usual,
$f$ be the endomorphism which corresponds to "shifting to the left". Then, $f\in\mf g$.
Indeed, using the second property in \eqref{20201001:eq4}, we have $f=\sum_{\alpha\leftarrow\beta} E_{\alpha,\beta}=\frac12\sum_{\alpha\leftarrow\beta} F_{\alpha,\beta}$.
Note that the "shift to the right" operator $f\transpose$ lies in $\mf g$ as well.

As in Section \ref{sec:setup}, let $h\in\End V$ be the diagonal endomorphism of $V$
whose eigenvalue on $e_{\alpha}$ is the $x$-coordinate of the center of the box labeled by $\alpha$ (see Figure \ref{fig:pyramid}) which we denote by $x_{\alpha}$.
We then have the $h$-eigenspace decomposition of $V$ \eqref{eq:dech} where $D=p_1-1$ is the maximal eigenvalue of $h$.

Note that $x_{\alpha'}=-x_{\alpha}$ and we have
$$h=\sum_{\alpha\in I}x_\alpha E_{\alpha,\alpha}=\frac12\sum_{\alpha\in I}x_\alpha F_{\alpha,\alpha}\in\mf g\,.$$
Clearly, the matrices $F_{\alpha,\beta}$ are eigenvectors with respect to the adjoint action of $h$:
\begin{equation}\label{eq:adx-sp}
(\ad h)F_{\alpha,\beta}=(x_\alpha-x_\beta)F_{\alpha,\beta}
\,.
\end{equation}
This defines a $\mb Z$-grading of $\mf g$,
given by the $\ad h$-eigenspaces as in \eqref{eq:grading}:
\begin{equation}\label{eq:adx2-sp}
\mf g_k=\Span_{\mb F}\{F_{\alpha,\beta}\mid x_{\alpha}-x_\beta=k\}\,.
\end{equation}

Recalling the definition \eqref{eq:vi} of the spaces $V_{\pm,i}$,
we have an isomorphism $f^{D_i}:V_{-,i}\to V_{+,i}$,
$i=1,\dots,s$.
Then the bilinear form \eqref{20200819:eq3} induces non-degenerate bilinear forms
$\beta_i(\cdot\,,\,\cdot):V_{+,i}\otimes V_{+,i}\to\mb F$, $i=1,\dots,s$, defined by
\begin{equation}\label{20200819:bilinear}
\beta_i(v,w):=\langle v|f^{D_i}w\rangle\,,\quad v,w\in V_{+,i}\,.
\end{equation}
Using Lemma \ref{20200819:lem2} and the fact that $f^\dagger=-f$, we have ($v,w\in V_{+,i}$)
\begin{equation}\label{20200819:bilinear2}
\beta_i(v,w)=(-1)^{D_i}\eta\beta_i(w,v)\,.
\end{equation}

%%%%%%%%%%%%%%%%%%%%%%%%%%%%%%%%%%%%%%%%%%%%%%%%%
\subsection{Integrable triples in \texorpdfstring{$\mf g=\mf{sp}_N$}{spN}}\label{sec:spN}
For $\eta=-1$, the Lie algebra \eqref{20201212:eq1} is $\mf g\simeq\mf{sp}_N$. In this case the depth of the grading \eqref{eq:adx2-sp} is $d=2D=2p_1-2$. To see this, we describe explicitly the space $\mf g_{2D}$ (clearly, $\mf g_k=0$, for $k>2D$).
Let $A\in\Hom(V[-D],V[D])$, and let $\bar A$ $=$ $Af^{D}$ $\in$ $\End(V[D])$.
We have
\begin{align*}
\beta_1(\bar A v,w)=\langle Af^{D}v|f^Dw\rangle=(-1)^{D}\langle v|f^{D}A^\dagger f^{D}w\rangle
=(-1)^{D}\beta_1(v|\overline{A^\dagger}w)
\,.
\end{align*}
Hence, the adjoint of $\bar A\in\End V[D]$ with respect to $\beta_1$ is
\begin{equation}\label{20201213:eq1}
\bar A^\star=(-1)^D\overline{A^\dagger}\,.
\end{equation}
As a consequence, we have a bijection
\begin{equation}\label{rem:gd}
\mf g_{2D}\simeq\{B\in\End(V[D])\mid B^\star=(-1)^{p_1}B\}\,,\quad A\mapsto\bar A
\,.
\end{equation}
Using equations \eqref{20200819:bilinear2} and \eqref{rem:gd} we have that, if $p_1$ is even, then $\mf g_d$ is identified with the space of selfadjont (with respect to the symmetric bilinear form $\beta_1$ in \eqref{20200819:bilinear}) endomorphsims of $V[D]$, while if $p_1$ is odd then $\mf g_d\simeq\mf{sp}(V[D],\beta_1)$.

%%%
\subsubsection{The centralizer $\mf z(\mf g_{\ge2})$}
From Lemma \ref{20200819:lem4} we have that $E\in\mf g_k$ can be decomposed as
\begin{equation}\label{eq:E-sp}
E=\sum_{x_\alpha-x_\beta=k}c_{\alpha\beta}F_{\alpha\beta}\,,\quad c_{\beta',\alpha'}=-\epsilon_{\alpha}\epsilon_{\beta}c_{\alpha,\beta}
\,.
\end{equation}
\begin{proposition}\label{lemma2-sp}
The centralizer of $\mf g_{\ge2}$ in $\mf g$ is $\mf z(\mf g_{\ge2})=W\oplus\mf g_{d-1}\oplus\mf g_d$, where
$$
W=\Span_{\mb F}\{F_{\alpha\beta}\mid x_{\alpha}=D-1,x_{\beta}=-D+1\}
\subset\mf g_{d-2}\,.
$$
\end{proposition}
\begin{proof}
As in Lemma \ref{lemma1} we have
$\mf z(\mf g_{\ge2})\subset\mf g_{\ge0}$.
Let $\tilde\alpha,\tilde\beta\in I$ be such that $x_{\tilde\alpha}=D$ and
$x_{\tilde\beta}\le D-2$ (the box $\tilde\beta$ is completely at the left of the box $\tilde\alpha$).
Then $F_{\tilde\alpha\tilde\beta}\in\mf g_{\ge2}$. Hence, letting $E$ as in \eqref{eq:E-sp},
and using the commutation relations \eqref{comm:BC} and equation \eqref{20200819:eq4b} we get
\begin{equation}\label{eq1-sp}
0=[F_{\tilde\alpha\tilde\beta},E]=\sum_{x_{\alpha}-x_{\beta}=k}c_{\alpha\beta}[F_{\tilde\alpha\tilde\beta},F_{\alpha\beta}]
=\sum_{x_{\tilde\beta}-x_{\beta}=k}2c_{\tilde\beta\beta}F_{\tilde\alpha\beta}
-\sum_{x_{\alpha}-D=k}2c_{\alpha\tilde\alpha}F_{\alpha\tilde\beta}\,.
\end{equation}
If $k\ge1$, then the condition $x_{\alpha}-D=k\ge1$ is empty.
Hence,
$$
\sum_{x_{\tilde\beta}-x_{\beta}=k}c_{\tilde\beta\beta}F_{\tilde\alpha\beta}=0\,,
$$
which implies that that $c_{\alpha\beta}=0$ if $x_{\alpha}\le D-2$.
Using the second equation in \eqref{eq:E-sp}, we have also that $c_{\alpha\beta}=0$ if $x_{\beta}\ge -D+2$.
If $k=0$, a similar argument to the one used in the proof of Lemma \ref{lemma1} shows that equation \eqref{eq1-sp} implies
$E=0$. As a consequence, $E\in W\oplus\mf g_{d-1}\oplus\mf g_d$. On the other hand, by Proposition
\ref{lemma2}, we have
$$
W\oplus\mf g_{d-1}\oplus\mf g_d=\mf z((\mf{gl}_N)_{\ge2})\cap\mf g\subset\mf z(\mf g_{\ge2})
\,.
$$
\end{proof}
Recalling the definition of integrable triples given in the Introduction, by Proposition \ref{lemma2-sp},
if $E$ is an integrable element for $f$, then
$E\in\left(W\oplus\mf g_{d-1}\oplus\mf g_d\right)\cap\mf g_{\ge1}$ (note that $\mf g_{d-1}\oplus\mf g_d\subset\mf g_{\ge1}$ for $p_1\ge2$, and $W\subset\mf g_{\ge1}$ for $p_1\ge3$).

%%%
\subsubsection{The coisotropy condition}\label{sec:5.4.2}

An element $E\in\mf g_{d-1}$ can be uniquely written as $E=A-A^{\dagger}$,
where $A\in\Hom(V[-D+1],V[D])$.
Let
\begin{equation}\label{eq:AfAf}
a=Af^{D-1}A^\dagger f^{D}\in\End(V[D]).
\end{equation}
\begin{proposition}\label{prop3-sp}
Let $E\in\mf g_{d-1}$. The subspace $\mf g_1^E$ is
coisotropic with respect to the bilinear form \eqref{eq:skewform} if and only if $a=0$.
\end{proposition}
\begin{proof}
By Lemma \ref{prop:20201029a} with $B=-A^\dagger$,
we have that
$\ad E\circ(\ad f)^{-1}\circ\ad E|_{\mf g_{-d}}=0$ if and only if
$$[a,(f\transpose)^{D}U]=0\,,
\quad\text{for every }U\in\mf g_{-d}\,.$$
Note that $\mf g_{-d}\cong\mf g_{d}$ and that $(f\transpose)^{D}\mf g_{-d}\subset\End(V[D])$.
By the description of $\mf g_d$ given in \eqref{rem:gd} we have that $(f\transpose)^{D}\mf g_{-d}=\{B\in\End(V[D])\mid B^\star=(-1)^{p_1}B\}$,
where the adjoint is with respect to the bilinear form $\beta_1$ defined in \eqref{20200819:bilinear}.
Recall that $a\in\End(V[D])$.
Using equation \eqref{20201213:eq1}, its adjoint with respect to the bilinear form $\beta_1$
is
$$
a^\star=(-1^D)(Af^{D-1}A^\dagger)^\dagger f^D=-a\,.
$$
Hence, by \eqref{rem:gd}, if $p_1$ is odd, then $a=0$ since it lies in the center of $\mf{sp}(V[D],\beta_1)$. On the other hand,
if $p_1$ is even, $a=0$ since it is a skewadjoint operator commuting with all selfadjoint operators.
\end{proof}
%

%%%
Now, let $E\in W=\Span_{\mb F}\{F_{\alpha\beta}\mid x_{\alpha}=D-1,x_{\beta}=-D+1\}
\subset\mf g_{d-2}$.
As in Section \ref{sec:5.2.3} we write $E=c(f\transpose)^{D-1}$, where $c\in\End(V[D-1])$.
\begin{proposition}\label{prop4-sp}
Let $E=c(f\transpose)^{D-1}\in W$. The subspace $\mf g_1^E$ is
coisotropic with respect to the bilinear form \eqref{eq:skewform} if and only if
$c^2=0$.
\end{proposition}
\begin{proof}
Same as the proof of Proposition \ref{prop4b}.
\end{proof}
%

%%%
\subsubsection{Integrable $E\in\mf g_d$}

%%%
%Let $E\in\mf g_d$.
Recall from Section \ref{sec:gd-gln} that $E=u(f\transpose)^D$, where $u\in\End(V[D])$.
The following result follows from Lemmas \ref{prop1} and \ref{prop2}.
\begin{proposition}\label{20201203:cor1}
Let $\mf g=\mf{sp}_N$ and $f$ its nonzero nilpotent element of depth $d$. Then $E=u(f\transpose)^D\in\mf g_{d}$ is integrable for $f$ if and only if $u$ is semisimple.
\end{proposition}

%%%
\subsubsection{No integrable elements for $f$ in $\mf g_{d-1}\oplus W$}
Recall the element $a$, defined by \eqref{eq:AfAf} for $E\in\mf g_{d-1}$, and the element $c$, defined by \eqref{eq:c} for $E\in W$.
\begin{proposition}\label{20201104:prop1}
There are no integrable elements for $f$ in $\mf g_{d-1}\oplus W$.
\end{proposition}
\begin{proof}
By contradiction, let $E\in\mf g_{d-1}$ be an integrable element for $f$.
Since $\mf g_1^E$ is coisotropic, by Proposition \ref{prop3-sp}, we have $a=0$.
Clearly, $f+E=(f+E)\id_{V_{\le2}}+(f+E)\id_{V_{\ge3}}$, and $f+E$ preserves the direct sum decomposition
$V=V_{\le2}\oplus V_{\ge3}$. Note that $(f+E)\id_{V_{\ge3}}=f\id_{V_{\ge3}}$ is nilpotent.
On the other hand it is not difficult to check that $(f+E)^{2D-1}v=av=0$ for every $v\in V_{\le2}$.
Hence,
$(f+E)\id_{V_{\le2}}$ is nilpotent as well. This proves that $f+E$ is nilpotent, contradicting the fact that it is integrable.
The proof of the claim for $E\in W$ is the same as the proof of Proposition \ref{cor:nointinW-b}.
\end{proof}

As a consequence of Propositions \ref{20201203:cor1} and \ref{20201104:prop1} we get the following.
\begin{corollary}
Let $\mf g=\mf{sp}_N$ and $f$ its non-zero nilpotent element. The element $f+E$, where $E\in\mf g_k$, $k\ge1$, is integrable if and only if $f+E$ is an integrable cyclic element.
\end{corollary}

%%%%%%%%%%%%%%%%%%%%%%%%%%%%%%%%%%%%%%%%%%%%%%%%%
\subsection{Integrable triples in \texorpdfstring{$\mf g=\mf{so}_N$}{soNe} for nilpotent elements of even depth}\label{subs:so}
Recall that for $\eta=1$ in \eqref{20201001:eq4}, the algebra \eqref{20201212:eq1} is $\mf g\simeq\mf{so}_N$.
\begin{lemma}[cf. \cite{EKV13}]\label{20201219:lem1}
The depth $d$ of the grading \eqref{eq:grading} for $\mf g=\mf{so}_N$ is
\begin{enumerate}[(i)]
\item $d=2D$ for $r_1\ge2$;
\item $d=2D-1$ for $p_1$ odd, $r_1=1$, $p_2=p_1-1$;
\item $d=2D-2$ for $p_1$ odd, $r_1=1$, $p_2\le p_1-2$.
\end{enumerate}
\end{lemma}
\begin{proof}
As in Section \ref{sec:spN}, let $A\in\Hom(V[-D],V[D])$, $D=p_1-1$, and consider $\bar A=Af^D\in\End(V[D])$.
Then, the adjoint of $\bar A$ with respect to the bilinear form $\beta_1$ defined in \eqref{20200819:bilinear} is given by equation \eqref{20201213:eq1} and the space $\mf g_{2D}$ is described in
\eqref{rem:gd}.
Using equation \eqref{20200819:bilinear2} we have that, if $p_1$ is even (hence $r_1$ is even), then $\mf g_{2D}$ is identified with the space of selfadjont (with respect to the skewsymmetric bilinear form $\beta_1$) endomorphisms of $V[D]$, while if $p_1$ is odd then $\mf g_{2D}\simeq\mf{so}(V[D],\beta_1)$.
It follows that the depth of the grading \eqref{eq:adx2-sp} is $d=2D$ when
$r_1=\dim V[D]\ge2$, proving case (i).

If $r_1=1$ (hence $p_1$ is odd), then $\mf g_{2D}=\{0\}$ and $d<2D$. Consider first the case when
$p_2=p_1-1$. In this case, $V[D]$ is one dimensional and $V[-D+1]$ is an $r_2$-dimensional vector space.
Furthermore, $\mf g_{2D-1}=\{A-A^\dagger \mid A\in\Hom(V[-D+1],V[D])\}\simeq V[-D+1]^*$. Hence,
in this case the depth of the grading \eqref{eq:adx2-sp} is $d=2D-1$, proving case (ii).

We are left to consider the case $r_1=1$ and $p_2\le p_1-2$. In this case $d=2D-2$.
Indeed, $F_{\alpha_1,\alpha_{p_1-1}}\in\mf g_{2D-2}$, where $\alpha_1$ is the rightmost box of the pyramid (which is in the bottom row) and $\alpha_{p_1-1}$ is the second leftmost box of the bottom row, is a non-zero element.
\end{proof}

\subsubsection{Even depth $d=2D$}\label{sec:4.5.1}
\begin{proposition}\label{centr-so1}
We have that $\mf z(\mf g_{\ge2})=W\oplus\mf g_{d-1}\oplus\mf g_d$,
where
$$W=\Span_{\mb F}\{F_{\alpha\beta}\mid x_{\alpha}=D-1,x_{\beta}=-D+1\}
\subset\mf g_{d-2}\,.
$$
\end{proposition}
\begin{proof}
Similar to the proof of Proposition \ref{lemma2-sp}.
\end{proof}
\begin{proposition}\label{coiso-so1}
\hspace{2em}
\begin{enumerate}[(i)]
\item Let, as in Section \ref{sec:5.4.2}, $E=A-A^\dagger\in\mf g_{d-1}$, with
$A\in\Hom(V[-D],V[D-1])$, and let $a=A^\dagger f^{D-1}Af^{D}\in\End (V[D])$.
If $p_1$ is even and $r_1=2$, then the subspace $\mf g_1^E$ is
coisotropic with respect to the bilinear form \eqref{eq:skewform}. In the other cases,
$\mf g_1^E$ is
coisotropic if and only if $a=0$.
\item Let $E\in W$ and let, as in Section \ref{sec:5.4.2}, $c=Ef^{D-1}\in\End(V[D-1])$. The subspace $\mf g_1^E$ is
coisotropic with respect to the bilinear form \eqref{eq:skewform} if and only if $c^2=0$.
\end{enumerate}
\end{proposition}
\begin{proof}
As in the proof of Proposition \ref{prop3-sp} we have that
$\ad E\circ(\ad f)^{-1}\circ\ad E|_{\mf g_{-d}}=0$ if and only if $[a,(f\transpose)^{D}U]=0$,
for every $U\in\mf g_{-d}$ and that $a\in\End(V[D])$ is skewadjoint with respect to the bilinear form
\eqref{20200819:bilinear}. Note that $\mf g_{-d}\cong\mf g_{d}$ and that, by equation \eqref{rem:gd},
we have $(f\transpose)^{D}\mf g_{-d}=\{B\in\End(V[D])\mid B^\star=(-1)^{p_1}B\}\End(V[D])\subset\End(V[D])$.
If $p_1$ is odd, then $a=0$ since it lies in the center of $\mf{so}(V[D],\beta_1)$.
On the other hand, if $p_1$ is even (this implies $r_1$ even), then
$a$ is a skewadjoint operator (with respect to a skewsymmetric bilinear form) commuting with all selfadjoint operators with respect to the same bilinear form. This gives no condition on $a$ when $r_1=2$, but implies that $a=0$ for $r_1>2$. This proves part i). The proof of part ii) is similar to the proof of
Proposition \ref{prop4b}.
\end{proof}

Let $p_1$ be even and $r_1=2=\dim(V[D])$.
Let us choose a basis $\{u,v\}$ of $V[D]$ such that
$$
\beta_{1}(u,v)=\langle u|f^{D}v\rangle=1
\,,
$$
where the bilinear form $\beta_1$ on $V[D]$ defined in \eqref{20200819:bilinear} is skew-symmetric by \eqref{20200819:bilinear2}.
Then, clearly $\{f^ku,f^hv\mid 0\le h,k\le D\}$ is a basis of $V_1$. In particular $V[-D]=\mb Ff^Du\oplus\mb Ff^Dv$.

For every $w\in V$, let us denote by $\phi(w):V\to\mb F$ the linear functional $\phi(w)(w_1)=\langle w|w_1\rangle$, $w_1\in V$. Hence, we can write
\begin{equation}\label{20201119:eq1}
A=x\phi(v)+y\phi(u)
\,,
\end{equation}
for some $x,y\in V[D-1]$. Then
\begin{equation}\label{20201119:eq2}
A^\dagger=u\phi(y)+v\phi(x)
\,.
\end{equation}
Indeed, using the fact that $\beta_1(w_1,w_1)=0$, for every $w_1\in V[D]$ we have
$$
\langle Af^{D}u|w\rangle=-\langle x|w\rangle=\langle f^Du|A^\dagger w\rangle\,,
\quad
\langle Af^{D}v|w\rangle=\langle y|w\rangle=\langle f^Dv|A^\dagger w\rangle
\,.
$$
Hence,
in this case, $E$ can be uniquely written as
\begin{equation}\label{eq:E-so-pari2}
E=x\phi(v)+y\phi(u)-u\phi(y)-v\phi(x)\,,
\quad x,y\in V[D-1]\,.
\end{equation}
\begin{lemma}\label{20201119:lem1}
With respect to the basis $\{u,v\}$ of $V[D]$ we have that
\begin{equation}\label{20201123:eq1}
a=A^{\dagger}f^{D-1}Af^{D}=\begin{pmatrix}-\beta_2(x,y)&\beta_2(y,y)\\-\beta_2(x,x)&\beta_2(x,y)
\end{pmatrix}\,,
\end{equation}
where $\beta_2$ is defined by \eqref{20200819:bilinear}.
\end{lemma}
\begin{proof}
We have, by \eqref{20201119:eq1} and \eqref{20201119:eq2},
\begin{align*}
a(u)&=A^\dagger f^{D-1}A f^{D}u=A^\dagger f^{D-1}
\left(\langle v|f^Du\rangle x+\langle u|f^Du\rangle y\right)
\\
&=-A^\dagger f^{D-1}x=-\left(\langle y|f^{D-1}x\rangle u+\langle x f^{D-1}x\rangle v\right)
\\
&=-\beta_2(y,x)u-\beta_2(x,x)v\,,
\end{align*}
which gives the first coloumn in the matrix \eqref{20201123:eq1}. Similarly for the second column.
\end{proof}
\begin{proposition}\label{20201119:prop1}
Let $E\in\mf g_{d-1}$ be as in \eqref{eq:E-so-pari2} and let $a\in\End(V[D])$ be as in \eqref{20201123:eq1}. If $a$ is a non-zero semisimple element, then $(f+E)_n\in\mf g_{-2}$.
\end{proposition}
\begin{proof}
By assumption, $a$ is a non-zero semisimple element.
In \eqref{20201123:eq1} $a$ is represented by a $2\times2$ traceless matrix. Hence, it is non-zero semisimple if and only if $\det(a)\neq0$. This implies that $x$ and $y$ are linearly independent.
Let
$$
\tilde U=\{w\in V[D-1]\mid \beta_2(x,w)=\beta_2(y,w)=0\}\subset V[D-1]\,.
$$
Since $\beta_2$ is non-degenerate and $\det(a)\neq0$ it easily follows that
$$
V[D-1]=\mb Fx\oplus\mb Fy\oplus\tilde U\,.
$$
As a consequence, applying $f$ repeatedly, we get
$$V_2=V_x\oplus V_y\oplus U\,,$$
where
$$V_x=\oplus_{k=0}^{D-1}\mb Ff^kx\,,\quad V_y=\oplus_{k=0}^{D-1}\mb Ff^ky\,,\quad
U=\oplus_{k=0}^{D-1}f^k\tilde U\,.
$$
Note that $(f+E)\id_{U\oplus V_{\ge3}}\!\!=\!\!f\id_{U\oplus V_{\ge3}}$ is nilpotent
and commutes with $(f+E)\id_{V_1\oplus V_x\oplus V_y}$.
Moreover, it is not difficult to check that
\begin{equation}\label{20201123:eq2}
(f+E)^{2D+1}w=a w\,,\qquad w\in V_1\oplus V_x\oplus V_y\,.
\end{equation}
Let $q(t)$ denote the minimal polynomial of $a$, which has non-zero distinct roots since $a$ is non-zero semisimple.
Equation \eqref{20201123:eq2} implies that
the minimal polynomial of $(f+E)\id_{V_1\oplus V_x\oplus V_y}$ divides $q(t^{2D+1})$ which obviously has also
distinct roots. Then $(f+E)\id_{V_1\oplus V_x\oplus V_y}$ is semisimple.
In conclusion, the Jordan decomposition of $f+E$ has $(f+E)_s=(f+E)\id_{V_1\oplus V_x\oplus V_y}$
and $(f+E)_n=f\id_{U\oplus V_{\ge3}}\in\mf g_{-2}$.
\end{proof}
The next result characterizes integrable elements for $\mf{so}_N$ when the depth of the grading is
$d=2D$.
\begin{theorem}\label{integrable-so1}
Let $\mf g=\mf{so}_N$ and let $f\in\mf g$ be a non-zero nilpotent element of depth $d=2D$, $D=p_1-1$.
\begin{enumerate}[(i)]
\item\label{it:intiffuss}
Let $E\in\mf g_{d}$ and let, as in Section \ref{sec:gd-gln}, $u=Ef^{D}\in\End(V[D])$. Then $E$ is integrable for $f$ if and only if $u$ is semisimple.
\item\label{it:intiffano0}
Let $E\in\mf g_{d-1}$ and let, as above $E=A-A^{\dagger}$, where $A\in\Hom(V[-D],V[D-1])$,
and $a=A^{\dagger} f^{D-1} A f^{D}\in\End(V[D])$. Then $E$ is integrable for $f$ if and only if $p_1$ is even, $r_1=2$ and $\det a\neq 0$.
\item
If $E\in W$, then $E$ cannot be integrable for $f$.
\end{enumerate}
\end{theorem}
\begin{proof}
Part i) follows by Proposition \ref{thm:glintiffss}. If $p_1$ is even and $r_1=2$ part ii) follows from
Proposition \ref{coiso-so1}i), Lemma \ref{prop5} and Proposition \ref{20201119:prop1}.
The remaining claim in part ii) and iii) can be proved in the same way as for the proof of Proposition \ref{20201104:prop1} using Proposition \ref{coiso-so1}.
\end{proof}
\begin{remark}
We have a non-zero $E\in\mf g_{d-1}$ if and only if $p_2=p_1-1$ (otherwise $\dim\mf g_{d-1}=0$).
In this case, if $r_2=1$, there are no elements $E\in\mf g_{d-1}$ satisfying the assumption of Theorem
\ref{integrable-so1}ii). Indeed, since $\dim V[D-1]=r_2=1$, $x$ and $y$ are linearly dependent thus
$a$ defined in \eqref{20201123:eq1} is nilpotent.
\end{remark}

\subsubsection{Even depth $d=2D-2$}\label{sec:5.5.2}
\begin{proposition}\label{20201221:prop1}
The centralizer of $\mf g_{\ge2}$ in $\mf g$ is $\mf z(\mf g_{\ge2})=\mf g_{d-1}\oplus\mf g_d$.
\end{proposition}
\begin{proof}
Clearly, $\mf z(\mf g_{\ge2})\subset\mf g_{\ge0}$. By degree considerations, $\mf z(\mf g_{\ge2})\supset\mf g_{d-1}\oplus\mf g_d$.
On the other hand, let $1\in I$ be the label of the rightmost box of the pyramid associated to $\partition p$ (note that $x_{1}=D$) and $p_1\in I$ be the label of the leftmost box (note that $x_{p_1}=-D$).
Let also $\tilde\beta\in I$ be such that $x_{\tilde\beta}\le D-2$ (the box $\tilde\beta$ is completely at the left of the box $1$).
Then $F_{1\tilde\beta}\in\mf g_{\ge2}$. Hence, letting $E$ as in equation \eqref{eq:E-sp},
using the commutation relations \eqref{comm:BC}, the second equation in \eqref{eq:E-sp} and \eqref{20200819:eq4b}, we have that
\begin{equation}\label{eq1-so-bis}
0=[F_{1\tilde\beta},E]=\sum_{x_{\alpha}-x_{\beta}=k}c_{\alpha\beta}[F_{1\tilde\beta},F_{\alpha\beta}]
=\sum_{x_{\tilde\beta}-x_{\beta}=k}2c_{\tilde\beta\beta}F_{1\beta}
-\sum_{x_{\alpha}-D=k}2c_{\alpha1}F_{\alpha\tilde\beta}\,.
\end{equation}
If $k\ge1$, then the condition $x_{\alpha}-D=k\ge1$ implies that $x_{\alpha}\ge1+D$, which is empty.
Hence, $c_{\alpha\beta}=0$ if $x_{\alpha}\le D-2$ and $\beta\neq p_1$ (since
$F_{1,p_1}=0$ by \eqref{20200819:eq4b} and \eqref{20201001:eq4}).
Using the second equation in \eqref{eq:E-sp}, we have also that $c_{\alpha\beta}=0$ if $\alpha\neq 1$ and $x_{\beta}\ge-D+2$.
Since $V[D-1]=V[-D+1]=0$, we have
$$
E=\sum_{x_\beta=D-k}c_{1\beta}F_{1\beta}\,.
$$
Let then assume that
$\tilde\alpha,\tilde\beta\in I$ are such that $x_{\tilde\alpha}-x_{\tilde\beta}\ge2$. Then
\begin{equation}\label{20201116:eq1-bis}
0=[F_{\tilde\alpha\tilde\beta},E]=2\delta_{x_{\tilde\alpha},D-k}c_{1,\tilde\alpha}F_{1,\tilde\beta}\,.
\end{equation}
If $x_{\tilde\alpha}=D-k$, then $-D\le x_{\tilde\beta}\le D-2-k$. Hence, equation \eqref{20201116:eq1-bis}
implies that $c_{1\beta}=0$ for $x_\beta>-D+3$, thus showing that $E\in\mf{g}_d\oplus\mf{g}_{d-1}$.

If $k=0$, a similar argument to the one used in the proof of Lemma \ref{lemma1} shows that $E=0$.
\end{proof}
In the sequel, we are going to use the following basis of $V_1$. Let $v_+\in V_{+,1}$ be such that
$\beta_1(v_+,v_+)=1$. Then we consider the basis $\{f^kv_+\}_{k=0}^{D}$ of
$V_{1}$. We also denote $v_-=f^Dv_+$.

An element $E\in\mf g_{d}$ can be uniquely written as $E=A-A^{\dagger}$, for some
$A$ $\in$ $\Hom(V[-D],V[D-2])$.
Note that $V[D]=\mb Fv_+$ and $V[D-2]=\mb Ffv_+\oplus V_{+,2}$. Hence,
\begin{equation}\label{eq:E2}
E=(a+\lambda fv_+)\phi(v_+)-v_+(\phi(a)+\lambda\phi(fv_+))
\,,
\end{equation}
for some $\lambda\in\mb F$ and $a\in V_{+,2}$.

Let now $E\in\mf g_{d-1}$ instead. It can be uniquely written as $E=B-B^{\dagger}$, where $B\in\Hom(V[-D],V_{+,3})$. Hence, $B=b\phi(v_+)$,
for some $b\in V_{+,3}$, and we can write
\begin{equation}\label{eq:E-so-pari}
E=b \phi(v_+)-v_+\phi(b)
\,.
\end{equation}
We also set
\begin{equation}\label{eq:beta}
\beta=\langle b|f^{D-3}b\rangle
\,.
\end{equation}
\begin{proposition}\label{coiso-so3}
Let $E\in\mf g_{d-1}$ be as in \eqref{eq:E-so-pari} and let $\beta$ be as in \eqref{eq:beta}. The subspace $\mf g_1^E$ is
coisotropic with respect to the bilinear form \eqref{eq:skewform} if and only if $\beta=0$.
\end{proposition}
\begin{proof}
Let $E=B-B^{\dagger}\in\mf g_{d-1}$ and let $U=X-X^\dagger\in\mf g_{d}$, for some $X$ $\in$ \allowbreak $\Hom(V[D],V[-D+2])$.
Recall that $B=b\phi(v_+)$. Hence, $\id_{fV}B=0$. Moreover, $B^{\dagger}=v_+\phi(b)$, Hence, $B^{\dagger}\id_{f\transpose V}=0$.
By Proposition \ref{prop:linalg} and Lemma \ref{prop:20201029c} we have that $\mf g_1^E$ is coisotropic if and only if
\begin{equation}\label{2021117:eq4}
(f\transpose)^{D-2}XB^{\dagger}f^{D-3}B=B^{\dagger}f^{D-3}AX^\dagger(f\transpose)^{D-2}=0\,,
\end{equation}
for every $X\in\Hom(V[D],V[-D+2])$. Note that the middle term in equation in \eqref{2021117:eq4} is the adjoint of the first term.
Hence, $\mf g_1^E$ is coisotropic if and only if $(f\transpose)^{D-2}XB^{\dagger}f^{D-3}B=0$ for every $X\in\Hom(V[D],V[-D+2])$.
Then, we have
$$
0=\left((f\transpose)^{D-2}XB^{\dagger}f^{D-3}B\right)(v_-)=\beta (f\transpose)^{D-2}X(v_+)
\,.
$$
Since $X$ is arbitrary we get $\beta=0$.
\end{proof}
The next result characterizes integrable elements for $\mf{so}_N$ when the depth of the grading is
$d=2D-2$.
\begin{theorem}\label{integrable-so3}
Let $\mf g=\mf{so}_N$ and let $f\in\mf g$ be a nilpotent element of even depth $d=2D-2$, $D=p_1-1$.
\begin{enumerate}[(i)]
\item\label{it:intiffla}
Let $E\in\mf g_{d}$ be as in \eqref{eq:E2}, where $\lambda\in\mb F$ and $a\in V_{+,2}$, and let also
$\alpha=\langle a|f^{D-2}a\rangle$. $E$ is integrable for $f$ if and only if
either $\alpha\neq0$, or $\alpha=a=0$ and $\lambda\neq0$.
\item\label{it:nointso}
There are no integrable elements for $f$ in $\mf g_{d-1}$.
\end{enumerate}
\end{theorem}
\begin{proof}
In order to prove part (i) one can use the same arguments that will be used for the proof Theorem \ref{integrable-so2} \eqref{it:intiffano0} (in fact, part (i) of the present Theorem corresponds to the special case $b=0$ of Theorem \ref{integrable-so2} \eqref{it:intiffano0}).
Let us then prove part (ii).
Let $E\in\mf g_{d-1}$ be as in \eqref{eq:E-so-pari} and let $\beta$ be as in \eqref{eq:beta}. If $E$ is integrable for $f$, then $\mf g_1^E$ must be coisotropic. By Proposition \ref{coiso-so3}, then $\beta=0$, and, as one can easily check, $f+E$ is nilpotent. This contradicts the fact that $E$ is integrable and proves part (ii).
\end{proof}

\begin{remark}
Theorem \ref{integrable-so3} \eqref{it:intiffla} follows from Theorem \ref{thm:intiffss} since in this case $Z(\mf s)|\mf g_d=\st(SO_{r_2})\oplus\id$ (see Table \ref{tab:zsd}).
\end{remark}

\subsubsection{} Let us reformulate the results, obtained in Subsection \ref{subs:so} about integrable quasicyclic elements in terms of polar linear groups.

First, actions of centralizers of the $\mf{sl}_2$-triples for nilpotent elements in simple Lie algebras of classical types on $\mf g_{d-1}$ are given in the following table:

\begin{table}[H]
\resizebox{.99\textwidth}{!}{
\renewcommand{\arraystretch}{1.75}
\begin{tabular}{c|l|c|c|c}
$\mf g$&\qquad\qquad partition&$d$&rank&$\mf z(\mf s)|\mf g_{d-1}$\\
\hline\hline
$\mf{sl}_N$\\
\hline
&$(p_1^{(r_1)},(p_1-1)^{(r_2)},...)$&$2p_1-2$&$\min(r_1,r_2)$&$D\left(\st(\mf{sl}_{r_1})^*\otimes\st(\mf{sl}_{r_2})\otimes\mb F\right)$\\
\hline\hline
$\mf{sp}_N$\\
\hline
&$(p_1^{(r_1)},(p_1-1)^{(r_2)},...)$, \hfill$p_1$ even&$2p_1-2$&$\min(\left[\frac{r_1}2\right],\frac{r_2}2)$&$\st(\mf{so}_{r_1})\otimes\st(\mf{sp}_{r_2})$\\
&$(p_1^{(r_1)},(p_1-1)^{(r_2)},...)$, \hfill$p_1$ odd&$2p_1-2$&$\min(\frac{r_1}2,\left[\frac{r_2}2\right])$&$\st(\mf{sp}_{r_1})
\otimes\st(\mf{so}_{r_2})$\\
\hline\hline
$\mf{so}_N$\\
\hline
&$(p_1^{(r_1)},(p_1-1)^{(r_2)},...)$, \hfill$p_1$ even&$2p_1-2$&$\min(\frac{r_1}2,\left[\frac{r_2}2\right])$&$\st(\mf{sp}_{r_1})\otimes\st(\mf{so}_{r_2})$\\
&$(p_1^{(r_1)},(p_1-1)^{(r_2)},...)$, \hfill$r_1>1$, $p_1$ odd&$2p_1-2$&$\min(\left[\frac{r_1}2\right],\frac{r_2}2)$&$\st(\mf{so}_{r_1})\otimes\st(\mf{sp}_{r_2})$\\
&$(p_1,(p_1-1)^{(r_2)},(p_1-2)^{(r_3)},...)$, \hfill$p_1$ odd&$2p_1-3$&$\frac{r_2}2+2$&$\ext^2\st(\mf{sp}_{r_2})\oplus\st(\mf{so}_{r_3})\oplus\id$\\
&$(p_1,(p_1-2)^{(r_2)},(p_1-3)^{(r_3)},...)$, \hfill$p_1$ odd&$2p_1-4$&$0$&$\st(\mf{sp}_{r_3})$
\end{tabular}
}%end of \resizebox
\caption{Actions of centralizers of the $\mf{sl}_2$-triples for nilpotent elements in simple Lie algebras of classical types on $\mf g_{d-1}$}\label{tab:zsdm1}
\end{table}

In Tables \ref{tab:zsdm1} and \ref{tab:withinv}, for a $G$-module $V$, $D(V)$ stands for the $G$-module $V\oplus V^*$.

\

Cartan subspaces for the entries in Table \ref{tab:zsdm1} are as follows.

\begin{itemize}
\item
For $D\left(\st(\mf{sl}_{r_1})^*\otimes\st(\mf{sl}_{r_2})\otimes\mb F\right)$: let us identify the representation space with  $\st(\mf{sl}_{r_1})\otimes\st(\mf{sl}_{r_2})^*+ \st(\mf{sl}_{r_1})^*\otimes\st(\mf{sl}_{r_2})$; then, a Cartan subspace is spanned by $u_i\otimes v_i^*+u_i^*\otimes v_i$, $i=1,...,\min(r_1,r_2)$,
where $u_i$, $u_i^*$, resp. $v_i$, $v_i^*$ are dual bases of $\st(\mf{sl}_{r_1})$ and $\st(\mf{sl}_{r_2})$ respectively. \item
For $\st(\mf{so}_{2m+\text{possibly $1$}})\otimes\st(\mf{sp}_{2n})$: let $u_1$, ..., $u_m$, (possibly $u_0$,) $u_{-m}$, ..., $u_{-1}$ be a basis of $\mb F^{2m+\text{possibly $1$}}$ and $v_1$, ..., $v_n$, $v_{-n}$, ..., $v_1$ be a basis of $\mb F^{2n}$; then, the subspace $C$ spanned by $\{u_i\otimes v_i+u_{-i}\otimes v_{-i}\mid i=1,...,\min(m,n)\}$ is a Cartan subspace.
\item For $\ext^2\st(\mf{sp}_{r_2})$: as in subsection \ref{subs:socases}, case \eqref{en:soee}.
\item For $\st(\mf{so}_{r_3})$: as in subsection \ref{subs:socases}, case \eqref{en:so31}.
\end{itemize}

\begin{theorem}
Let $\mf g=\mf{so}_N$, let $f\in\mf g$ be a non-zero nilpotent element of even depth $d$ and let $\partition p=(p_1^{(r_1)},p_2^{(r_2)},...)$ be the corresponding partition. Then
\begin{enumerate}[(a)]
\item\label{it:sopolar} All the linear groups $Z(\mf s)|\mf g_{d-1}$ are polar, and described in Table \ref{tab:zsdm1}.
\item\label{it:soint} There exists an integrable quasi-cyclic element for $f$ if and only if $p_1$ is even, $r_1=2$, and $p_2=p_1-1$, $r_2\ge2$. Such an element is the unique, up to equivalence, element $f+E$, where $E\in\mf g_{d-1}$ is semisimple with respect to $Z(\mf s)$.
\end{enumerate}
\end{theorem}
\begin{proof}
\eqref{it:sopolar} follows from Table \ref{tab:zsdm1}, since all these linear groups are theta groups.

In order to prove \eqref{it:soint}, note that in the case in question, $Z(\mf s)|\mf g_{d-1}=\st(Sp_2)\otimes\st(SO_{r_2})$, for which the rank equals $1$. This shows that, up to equivalence, there is at most one integrable quasi-cyclic element. Its existence follows from Theorem \ref{integrable-so1} \eqref{it:intiffano0}.
\end{proof}

\subsection{Integrable quasi-cyclic elements in exceptional Lie algebras for nilpotent elements of even depth}

\

\ 

\begin{table}[H]
\begin{center}

\ 

\resizebox{.99\textwidth}{!}{ 
\begin{tabular}{l|rr|c|l|c|c|c}
$\mf g$&nilpotent $f$&&$d$&$\mf z(\mf s)\mid{\mf g}_{d-\half}$&rank&quasi-cyclic $f+E$, $E\in\mf g_{d-1}$&quasi type\bigstrut[b]\\
\hline
\hline
$\rG_2$&$\rA_1$&\GDynkin{0,\half}&$2$&$\sym^3\st(\mf{sl}_2)$&1&$\exists$ semisimple&semisimple\bigstrut[t]\bigstrut[b]\\
\hline
$\rF_4$&$\rA_1$&\FDynkin{0,\half,0,0}&$2$&$\ext^3_0\st(\mf{sp}_6)$&1&$\exists$ semisimple&semisimple\bigstrut[t]\\
&$\widetilde\rA_1$&\FDynkin{\half,0,0,0}&$2$&$D(\st(\mf{sl}_4))$&1&nilpotent only&semisimple\bigstrut[b]\\
\hline
$\rE_6$&$\rA_1$&\EDynkin{\half,0,0,0,0,0}{1}&$2$&$\ext^3\st(\mf{sl}_6)$&1&$\exists$ semisimple&semisimple\bigstrut[t]\\
&$2\rA_1$&\EDynkin{0,\half,0,0,0,\half}{1}&$2$&$\spin_7\otimes\st(\mf{so}_2)$&2&$\exists$ semisimple&semisimple\\
&$\rA_2+\rA_1$&\EDynkin{\half,\half,0,0,0,\half}{1}&$4$&$D(\st(\mf{gl}_3))$&1&$\exists$ semisimple&semisimple\\
&$\rA_2+2\rA_1$&\EDynkin{0,0,\half,0,\half,0}{1}&$4$&$\st(\mf{so}_2)\otimes\st(\mf{sl}_2)$&1&nilpotent only&mixed\\
&$\rA_4+\rA_1$&\EDynkin{\half,\half,\half,0,\half,\half}{1}&$8$&$\st(\mf{so}_2)$&1&$\exists$ semisimple&semisimple\bigstrut[b]\\
\hline
$\rE_7$&$\rA_1$&\EDynkin{0,\half,0,0,0,0,0}{1}&$2$&$\spin_{12}$&1&$\exists$ semisimple&semisimple\bigstrut[t]\\
&$2\rA_1$&\EDynkin{0,0,0,0,0,\half,0}{1}&$2$&$\spin_9\otimes\st(\mf{sl}_2)$&2&nilpotent only&semisimple\\
&$\rA_2+\rA_1$&\EDynkin{0,\half,0,0,0,\half,0}{1}&$4$&$D(\st(\mf{gl}_2))$&1&$\exists$ semisimple&semisimple\\
&$\rA_2+2\rA_1$&\EDynkin{0,0,0,\half,0,0,0}{1}&$4$&$\st(\mf{so}_4)\otimes\st(\mf{sl}_2)$&1&nilpotent only&mixed\\
&$\rA_3+\rA_2$&\EDynkin{0,0,0,\half,0,\half,0}{1}&$6$&$D(\st(\mf{gl}_2))$&1&nilpotent only&semisimple\\
&$\rA_4+\rA_1$&\EDynkin{0,\half,0,\half,0,\half,0}{1}&$8$&$\st(\mf{so}_2)$&1&$\exists$ semisimple&semisimple\\
\hline
$\rE_8$&$\rA_1$&\EDynkin{0,0,0,0,0,0,0,\half}{1}&$2$&$\mathbf{56}$&1&$\exists$ semisimple&semisimple\bigstrut[t]\\
&$2\rA_1$&\EDynkin{0,\half,0,0,0,0,0,0}{1}&$2$&$\spin_{13}$&2&nilpotent only&semisimple\\
&$\rA_2+\rA_1$&\EDynkin{0,\half,0,0,0,0,0,\half}{1}&$4$&$D(\st(\mf{sl}_6))$&1&$\exists$ semisimple&semisimple\\
&$\rA_2+2\rA_1$&\EDynkin{0,0,0,0,0,\half,0,0}{1}&$4$&$\spin_7\otimes\st(\mf{sl}_2)$&1&nilpotent only&mixed\\
&$\rA_2+3\rA_1$&\EDynkin{0,0,\half,0,0,0,0,0}{1}&$4$&$\st(\rG_2)\otimes\st(\mf{sl}_2)$&1&nilpotent only&mixed\\
&$\rA_3+\rA_2$&\EDynkin{0,\half,0,0,0,\half,0,0}{1}&$6$&$\st(\mf{sp}_4)\otimes\st(\mf{so}_2)$&1&nilpotent only&semisimple\\
&$\rA_3+\rA_2+\rA_1$&\EDynkin{0,0,0,\half,0,0,0,0}{1}&$6$&$\st(\mf{sl}_2)\otimes\st(\mf{so}_3)$&1&nilpotent only&mixed\\
&$\rA_4+\rA_1$&\EDynkin{0,\half,0,0,0,\half,0,\half}{1}&$8$&$\st(\mf{so}_2)$&1&$\exists$ semisimple&semisimple\\
&$\rA_4+2\rA_1$&\EDynkin{0,0,0,\half,0,0,0,\half}{1}&$8$&$D(\st(\mf{gl}_2))$&1&$\exists$ integrable&mixed\\
&$\rA_4+\rA_2+\rA_1$&\EDynkin{0,0,\half,0,0,\half,0,0}{1}&$8$&$\sym^3\st(\mf{sl}_2)$&1&nilpotent only&mixed\\
&$\rD_7(a_2)$&\EDynkin{0,\half,0,\half,0,\half,0,\half}{1}&$14$&$\st(\mf{so}_2)$&1&$\exists$ semisimple&semisimple
\end{tabular}
}

\ 

\ 

\end{center}
\caption{Quasi-cyclic elements, attached to odd nilpotent elements $f$ of even depth in exceptional $\mf g$ with $Z(\mf s)|\mf g_{d-1}$ having non-trivial invariants}\label{tab:withinv}
\end{table}

By the quasi type in the last column we mean whether $f+E$ for a generic $E\in\mf g_{d-1}$ is semisimple or mixed
(meaning neither semisimple nor nilpotent).

\begin{table}[H]
\setlength\tabcolsep{2pt}
\begin{tabular}{p{.44\textwidth}p{.54\textwidth}}
\vspace{0pt}
\resizebox{.42\textwidth}{!}
{
\begin{tabular}{l|rr|c|l}
$\mf g$&nilpotent $f$&&$d$&$\mf z(\mf s)\mid{\mf g}_{d-\half}$\bigstrut[b]\\
\hline
$\rF_4$&$\rA_2+\widetilde\rA_1$&\FDynkin{0,0,\half,0}&$4$&$\st(\mf{sl}_2)$\bigstrut[t]\\
&$\rB_2$&\FDynkin{\half,2,0,0}&$6$&$0$\\
&$\rC_3(a_1)$&\FDynkin{0,\half,\half,0}&$6$&$\st(\mf{sl}_2)$\\
&$\rC_3$&\FDynkin{2,\half,\half,0}&$10$&$\st(\mf{sl}_2)$\bigstrut[b]\\
\hline
$\rE_6$&$\rA_3$&\EDynkin{2,\half,0,0,0,\half}{1}&$6$&$0$\bigstrut[t]\\
&$\rA_3+\rA_1$&\EDynkin{\half,0,\half,0,\half,0}{1}&$6$&$\st(\mf{sl}_2)$\\
&$\rA_5$&\EDynkin{\half,2,\half,0,\half,2}{1}&$10$&$\st(\mf{sl}_2)$\\
&$\rD_5(a_1)$&\EDynkin{2,\half,\half,0,\half,\half}{1}&$10$&$0$\bigstrut[b]\\
\hline
$\rE_7$&$\rA_3$&\EDynkin{0,2,0,0,0,\half,0}{1}&$6$&$0$\bigstrut[t]\\
&$\rA_3+\rA_1$&\EDynkin{0,\half,0,\half,0,0,0}{1}&$6$&$\st(\mf{sl}_2)$\\
&$\rA_3+2\rA_1$&\EDynkin{0,\half,0,0,\half,0,\half}{1}&$6$&$\st(\mf{sl}_2)\oplus\st(\mf{sl}_2)$\\
&$\rD_4(a_1)+\rA_1$&\EDynkin{\half,0,\half,0,0,0,\half}{1}&$6$&$\st(\mf{sl}_2)\oplus\st(\mf{sl}_2)$\\
&$\rD_4+\rA_1$&\EDynkin{\half,2,\half,0,0,0,\half}{1}&$10$&$0$\\
&$\rD_5(a_1)$&\EDynkin{0,2,0,\half,0,\half,0}{1}&$10$&$0$\\
&$\rA_5$&\EDynkin{0,\half,0,\half,0,2,0}{1}&$10$&$\st(\mf{sl}_2)$\\
&$\rA_5+\rA_1$&\EDynkin{0,\half,0,\half,0,\half,2}{1}&$10$&$\st(\mf{sl}_2)$\\
&$\rD_6(a_2)$&\EDynkin{\half,0,\half,0,\half,0,2}{1}&$10$&$\st(\mf{sl}_2)$\\
&$\rD_5+\rA_1$&\EDynkin{\half,2,\half,0,\half,\half,0}{1}&$14$&$0$\\
&$\rD_6(a_1)$&\EDynkin{\half,2,\half,0,\half,0,2}{1}&$14$&$0$\\
&$\rD_6$&\EDynkin{\half,2,\half,0,\half,2,2}{1}&$14$&$0$
\end{tabular}
}&\vspace{0pt}
\resizebox{.5\textwidth}{!}
{
\begin{tabular}{l|rr|c|l}
$\mf g$&nilpotent $f$&&$d$&$\mf z(\mf s)\mid{\mf g}_{d-\half}$\bigstrut[b]\\
\hline
$\rE_8$&$\rA_3$&\EDynkin{0,\half,0,0,0,0,0,2}{1}&$6$&$0$\bigstrut[t]\\
&$\rA_3+\rA_1$&\EDynkin{0,0,0,0,0,\half,0,\half}{1}&$6$&$\st(\mf{sl}_2)$\\
&$\rA_3+2\rA_1$&\EDynkin{0,0,\half,0,0,0,0,\half}{1}&$6$&$\st(\mf{sl}_2)\oplus\st(\mf{sp}_4)$\\
&$\rD_4(a_1)+\rA_1$&\EDynkin{\half,0,0,0,0,0,\half,0}{1}&$6$&$\st(\mf{sl}_2)\oplus\st(\mf{sl}_2)\oplus\st(\mf{sl}_2)$\\
&$\rD_4+\rA_1$&\EDynkin{\half,0,0,0,0,0,\half,2}{1}&$10$&$0$\\
&$\rD_5(a_1)$&\EDynkin{0,\half,0,0,0,0,\half,2}{1}&$10$&$0$\\
&$\rA_5$&\EDynkin{0,2,0,0,0,\half,0,\half}{1}&$10$&$\st(\mf{sl}_2)$\\
&$\rD_5(a_1)+\rA_1$&\EDynkin{0,0,0,\half,0,0,0,2}{1}&$10$&$0$\\
&$\rA_5+\rA_1$&\EDynkin{0,\half,0,\half,0,0,0,\half}{1}&$10$&$\st(\mf{sl}_2)\oplus\st(\mf{sl}_2)$\\
&$\rD_5(a_1)+\rA_2$&\EDynkin{0,0,\half,0,0,\half,0,\half}{1}&$10$&$\st(\mf{sl}_2)$\\
&$\rD_6(a_2)$&\EDynkin{\half,0,\half,0,0,0,\half,0}{1}&$10$&$\st(\mf{sl}_2)\oplus\st(\mf{sl}_2)$\\
&$\rE_6(a_3)+\rA_1$&\EDynkin{0,\half,0,0,\half,0,\half,0}{1}&$10$&$\st(\mf{sl}_2)$\\
&$\rE_7(a_5)$&\EDynkin{0,0,0,\half,0,\half,0,0}{1}&$10$&$\st(\mf{sl}_2)$\\
&$\rD_5+\rA_1$&\EDynkin{0,\half,0,0,\half,0,\half,2}{1}&$14$&$0$\\
&$\rD_6(a_1)$&\EDynkin{\half,0,\half,0,0,0,\half,2}{1}&$14$&$0$\\
&$\rA_6+\rA_1$&\EDynkin{0,\half,0,\half,0,\half,0,0}{1}&$14$&$\st(\mf{sl}_2)$\\
&$\rE_7(a_4)$&\EDynkin{0,0,0,\half,0,\half,0,2}{1}&$14$&$0$\\
&$\rD_6$&\EDynkin{\half,2,\half,0,0,0,\half,2}{1}&$18$&$0$\\
&$\rE_6(a_1)+\rA_1$&\EDynkin{0,\half,0,\half,0,\half,0,2}{1}&$18$&$0$\\
&$\rE_7(a_3)$&\EDynkin{0,2,0,\half,0,\half,0,2}{1}&$18$&$0$\\
&$\rE_6+\rA_1$&\EDynkin{0,\half,0,\half,0,\half,2,2}{1}&$22$&$0$\\
&$\rE_7(a_2)$&\EDynkin{\half,0,\half,0,\half,0,2,2}{1}&$22$&$0$\\
&$\rD_7$&\EDynkin{\half,2,\half,0,\half,\half,0,\half}{1}&$22$&$\st(\mf{sl}_2)$\\
&$\rE_7(a_1)$&\EDynkin{\half,2,\half,0,\half,\half,0,\half}{1}&$26$&$0$\\
&$\rE_7$&\EDynkin{\half,2,\half,0,\half,2,2,2}{1}&$34$&$0$
\end{tabular}
}
\end{tabular}
\caption{Odd nilpotent elements $f$ of even depth in exceptional $\mf g$, for which $Z(\mf s)|\mf g_{d-1}$ has only trivial invariants}\label{tab:noinv}
\end{table}

\subsubsection{Nilpotent elements with label $\rA_1$}\label{subsubs:A1}
A representative for the orbit with label $\rA_1$ is given by a negative root vector $e_{-\alpha}$ --- arbitrary for type $\rE$ and a long one for $\rF_4$ and $\rG_2$.
Depth is $2$, and a $1$-dimensional Cartan subspace for $\mf z(\mf s)|\mf g_1$ is spanned by the vector $E=v^*+v_*$,
where $v^*$, resp. $v_*$ is the highest, resp. lowest weight vector.
Moreover $E$ satisfies the coisotropy condition, and the quasi-cyclic element $f+E$ is semisimple.

\subsubsection{Nilpotent element with label $\widetilde\rA_1$ in $\rF_4$}
A representative $f$ for $\widetilde\rA_1$ is given by a short root vector.
Depth is $2$, with $\mf g_1$ of dimension $8$. The representation $\mf z(\mf s)|\mf g_1$ is the direct sum $V=V_1\oplus V_2$ of two standard representations of $\mf{sl}_4$,
and a Cartan subspace is spanned by $E=v^1+v_2$, where $v^1$ is a highest weight vector for $V_1$ and $v_2$ a lowest weight vector for $V_2$.
This $E$ does not satisfy the coisotropy condition, which implies that all quasi-cyclic elements are nilpotent.

\subsubsection{Nilpotent element with label $2\rA_1$ in $\rE_6$}\label{subsubs:2A1inE6}
A representative $f$ for $2\rA_1$ is given by the sum of any two commuting root vectors.
Depth is $2$, and $\mf z(\mf s)|\mf g_1$ is as in Example \ref{ex:spin7stso2},
so we can choose a basis $\{E_1,E_2\}$ of a Cartan subspace as there, with $E_1=v^++v_-$ and $E_2=v^-+v_+$.
Coisotropy condition for $E=x_1E_1+x_2E_2$ turns out to be $x_1^2=x_2^2$.
We then check that both $f+x(E_1+E_2)$ and $f+x(E_1-E_2)$ are semisimple for $x\ne0$.
These are all integrable quasi-cyclic elements $f+E$ for $E$ from the Cartan subspace.

In fact these two solutions $E_1+E_2$ and $E_1-E_2$ are equivalent under the action of $Z(\mf s)$. This can be seen as follows. Take
\[
f=f_{\foresix{.1}{.5}111111}+f_{\foresix{.1}{.5}112321},
\]
then
\[
E_1=e_{\foresix{.1}{.5}100000}+e_{\foresix{.1}{.5}011111},\qquad E_2=e_{\foresix{.1}{.5}000111}+e_{\foresix{.1}{.5}112210}.
\]

Consider the element $H$ of the Cartan subalgebra determined by the following values of simple roots on it:
\[
\begin{array}{ccccc}
0&2&0&\!\!\!-2&0\\
&&0
\end{array}
\]
Its eigenvalues on $\mf g$ are $-2,0,2$, so it defines a $\mb Z$-grading $\mf g=\mf g_{-2}\oplus\mf g_0\oplus\mf g_2$ of $\mf g$. Moreover $f,E_1\in\mf g_0$, $e_{\foresix{.1}{.5}000111}\in\mf g_2$ and $e_{\foresix{.1}{.5}112210}\in\mf g_{-2}$, i.~e.
\[
[H,f]=[H,E_1]=0,\quad[H,e_{\foresix{.1}{.5}000111}]=2e_{\foresix{.1}{.5}000111},\quad[H,e_{\foresix{.1}{.5}112210}]=-2e_{\foresix{.1}{.5}112210}.
\]
Hence in the corresponding $\mb Z/2\mb Z$-grading $\mf g=\mf g^0\oplus\mf g^1$, where $\mf g^0=\mf g_0$ and $\mf g^1=\mf g_2\oplus\mf g_{-2}$, one has $f,E_1\in\mf g^0$ and $E_2\in\mf g^1$. Let $\alpha_H$ be the inner automorphism corresponding to this $\mb Z/2\mb Z$-grading, i.~e. $\alpha_H(x)=x$ for $x\in\mf g^0$ and $\alpha_H(x)=-x$ for $x\in\mf g^1$. Then $\alpha_H(f)=f$, $\alpha_H(E_1)=E_1$ and $\alpha_H(E_2)=-E_2$. Hence $\alpha_H(E_1+E_2)=E_1-E_2$ and $\alpha_H(E_1-E_2)=E_1+E_2$, so $\alpha_H$ interchanges the above solutions.

\subsubsection{Nilpotent element with label $2\rA_1$ in $\rE_7$}
Also here, a representative $f$ for $2\rA_1$ is given by the sum of any two commuting root vectors.
Depth is $2$, and $\mf z(\mf s)|\mf g_1$ is as in Example \ref{ex:spin9stsl2},
so we can choose a basis $\{v_1,v_2\}$ of a Cartan subspace as there, with $v_1=v_\Lambda+v_{-\Lambda}$ and $v_2=e_{-\eps_1}e_{-\eps_2-\eps_3}v_\Lambda+e_{\eps_1}e_{\eps_2+\eps_3}v_{-\Lambda}$.
Coisotropy condition for $E=x_1v_1+x_2v_2$ turns out to be $x_1^2+x_2^2=0$ and $x_1^2-x_2^2=0$,
which implies that all quasi-cyclic elements are nilpotent.

\subsubsection{Nilpotent element with label $2\rA_1$ in $\rE_8$}\label{subs:spin13}
As in two previous cases, a representative $f$ for $2\rA_1$ is given by the sum of any two commuting root vectors.
Depth is $2$, and $\mf z(\mf s)|\mf g_1$ is $\spin_{13}$, so, as in \cite[Proposition 10]{GV78},
we can choose a basis $\{v_1,v_2\}$ of a Cartan subspace as there, with $v_1=v_\Lambda+v_{-\Lambda}$ and $v_2=e_{-\eps_1}e_{-\eps_2}e_{-\eps_3}v_\Lambda+e_{\eps_1}e_{\eps_2}e_{\eps_3}v_{-\Lambda}$, where $\Lambda$ is the highest weight.
Coisotropy condition for $E=x_1v_1+x_2v_2$ turns out to be, as in the previous case, $x_1^2+x_2^2=0$ and $x_1^2-x_2^2=0$,
which implies that all quasi-cyclic elements are nilpotent.

\subsubsection{Nilpotent elements with label $\rA_2+\rA_1$ in $\rE_6$, $\rE_7$ and $\rE_8$}
Here depth is $4$, the algebra $\mf z(\mf s)$ is $\mf{gl}_3$ for $\rE_6$, $\mf{gl}_2$ for $\rE_7$ and $\mf{sl}_6$ for $\rE_8$.
The representation $\mf z(\mf s)|\mf g_3$ is a direct sum $V_1\oplus V_2$ of two copies of a standard representation
of $\mf{sl}_3$ for $\rE_6$, of $\mf{sl}_2$ for $\rE_7$ and of $\mf{sl}_6$ for $\rE_8$.
A Cartan subspace of $\mf g_3$ is spanned by $E=v^1+v_2$, where $v^1$ is a highest weight vector for $V_1$ and $v_2$ is a lowest weight vector for $V_2$.
Each of these $E$ satisfies the coisotropy condition, and the quasi-cyclic element $f+E$ is semisimple.

\subsubsection{Nilpotent element with label $\rA_2+2\rA_1$ in $\rE_6$, $\rE_7$ and $\rE_8$}
Depth is $4$. Here $\mf z(\mf s)|\mf g_3$ is $V\otimes\st(\mf{sl}_2)$,
where $V$ is $\st(\mf{so}_2)$ for $\rE_6$, $\st(\mf{so}_4)$ for $\rE_7$ and $\spin_7$ for $\rE_8$.
In all three cases a Cartan subspace is spanned by $E=v^*\otimes v_*+v_*\otimes v^*$, where $v^*$ denote highest weight vectors and $v_*$ the lowest weight vectors,
both for $V$ and for $\st(\mf{sl}_2)$.
This $E$ does not satisfy the coisotropy condition, so that all quasi-cyclic elements are nilpotent.

\subsubsection{Nilpotent elements with label $\rA_4+\rA_1$ in $\rE_6$, $\rE_7$ and $\rE_8$}\label{subsubs:A4+A1}
Depth is $4$. The representation $\mf z(\mf s)|\mf g_3$ is $2$-dimensional, it is $\st(\mf{so}_2)$
with a $1$-dimensional Cartan subspace spanned by $E=v^*+v_*$, where $v^*$, resp. $v_*$ is the highest, resp. lowest weight vector.
This $E$ satisfies the coisotropy condition, and $f+E$ is semisimple.

\subsubsection{Nilpotent elements with label $\rA_3+\rA_2$ in $\rE_7$ and $\rE_8$}
Depth is $6$. The algebra $\mf z(\mf s)$ is $\mf{sl}_2$ plus a $1$-torus for $\rE_7$ and $\mf{sp}_4$ plus a $1$-torus for $\rE_8$.
The representation $\mf z(\mf s)|\mf g_5$ is direct sum $V_+\oplus V_-$ of two copies of a standard representation,
with the torus acting as $\pm1$ on $V_\pm$. It has a $1$-dimensional Cartan subspace spanned by $E=v^++v_-$,
the sum of the highest weight vector of $V_+$ and the lowest weight vector of $V_-$.
This $E$ does not satisfy the coisotropy condition, so that all quasi-cyclic elements are nilpotent.

\subsubsection{Nilpotent element with label $\rA_2+3\rA_1$ in $\rE_8$}
Depth is $4$, and $\mf z(\mf s)|\mf g_3$ is $\st(\rG_2)\otimes\st(\mf{sl}_2)$,
with $1$-dimensional Cartan subspace spanned by $E=v^*+v_*$, the sum of the highest and the lowest weight vectors.
This $E$ does not satisfy the coisotropy condition, which means that all quasi-cyclic elements are nilpotent.

\subsubsection{Nilpotent element with label $\rA_3+\rA_2+\rA_1$ in $\rE_8$}
Depth is $6$, and $\mf z(\mf s)|\mf g_5$ is $\st(\mf{sl}_2)\otimes\st(\mf{so}_3)$,
with $1$-dimensional Cartan subspace spanned by $E=v^*+v_*$, the sum of the highest and the lowest weight vectors.
This $E$ does not satisfy the coisotropy condition, which means that all quasi-cyclic elements are nilpotent.

\subsubsection{Nilpotent element with label $\rA_4+2\rA_1$ in $\rE_8$}\label{subsubs:A4+2A1}
Depth is $8$, the algebra $\mf z(\mf s)$ is $\mf{sl}_2$ plus a $1$-torus,
and the representation $\mf z(\mf s)|\mf g_7$ is the direct sum $V_+\oplus V_-$ of two copies of $\st(\mf{sl}_2)$, with the torus acting by $\pm1$ on $V_\pm$.
It has a $1$-dimensional Cartan subspace spanned by $E=v^++v_-$, the sum of the highest and the lowest weight vectors of $V_+$, resp. $V_-$.
This $E$ satisfies the coisotropy condition, and the quasi-cyclic element $f+E$ has Jordan decomposition $(\fs+E)+\fn$ where $\fs,\fn\in\mf g_{-2}$ are nilpotent elements with labels $\rA_4+\rA_1$ and $\rA_1$ respectively. This gives an integrable triple for this case.

\subsubsection{Nilpotent element with label $\rA_4+\rA_2+\rA_1$ in $\rE_8$}
Depth is $8$, the algebra $\mf z(\mf s)$ is $\mf{sl}_2$,
and $\mf z(\mf s)|\mf g_7$ is its $4$-dimensional irreducible representation.
It has a $1$-dimensional Cartan subspace spanned by $E=v^*+v_*$, the sum of the highest and the lowest weight vectors.
This $E$ does not satisfy the coisotropy condition, so that all quasi-cyclic elements are nilpotent.

\subsubsection{Nilpotent element with label $\rD_7(a_2)$ in $\rE_8$}\label{subsubs:D7a2}
Depth is $14$, and $\mf z(\mf s)|\mf g_{13}$ is the standard representation of $\mf{so}_2$.
It has a $1$-dimensional Cartan subspace spanned by $E=v^*+v_*$, the sum of the highest and the lowest weight vectors.
This $E$ satisfies the coisotropy condition, and the quasi-cyclic element $f+E$ is semisimple.

\begin{remark}
It was proved in \cite{DSKV13} that for a long root vector $f\in\mf g\ne\mf{sp}_N$ there exists a unique, up to equivalence, integrable quasi-cyclic element. This covers $f$ of type $\rA_1$ in all exceptional $\mf g$ (see Table \ref{tab:withinv}).
\end{remark}

\addtocontents{toc}{\SkipTocEntry}

\subsection*{Conclusion} Due to Theorem \ref{thm:quasintiffss} \eqref{thm:quasintiffsseven},
Subsections \ref{subsubs:A1}, \ref{subsubs:2A1inE6}, \ref{subsubs:A4+A1}, \ref{subsubs:A4+2A1}, and \ref{subsubs:D7a2}
describe all integrable quasi-cyclic elements $f+E$ for nilpotent elements $f$ of even depth for all examples from Table \ref{tab:withinv},
up to conjugation by $Z(\mf s)$.
Obviously even nilpotent elements and the nilpotent elements from Table \ref{tab:noinv} have no integrable quasi-cyclic elements.
As a result, we see that for each nilpotent element $f$ of even depth in an exceptional simple Lie algebra
either there are no integrable quasi-cyclic elements $f+E$, or, up to equivalence, there is exactly one.

%%%%%%%%%%%%%%%%%%%%%%%%%%%%%%%%%%%%%%%%%%%%%%%%%%%%%%%%%%%%%%%%%%%%%%
\section{Integrable quasi-cyclic elements associated to nilpotent elements of odd depth}

Recall that if $f\in\mf g$ is a nilpotent element of odd depth $d$,
then all elements of the $Z(\mf s)$-module $\mf g_d$ are nilpotent \cite[Theorem 1]{EKV13}.
Actually the linear group $Z(\mf s)|\mf g_d$ is the full symplectic group, \cite[Remark 2]{EJK20}, hence it is polar.
Thus, by Lemma \ref{lem:int=>ss}, if $f$ has odd depth, there are no integrable cyclic elements
$f+E$, with $E\in\mf g_d$.

By \cite{EKV13}, if $\mf g$ is a classical Lie algebra, nilpotent elements $f$ of odd depth exist only in $\mf{so}_n$,
and for $\mf g$ exceptional, such $f$ are listed in \cite[Table 1]{EKV13}.
These two cases are treated in Subsections \ref{subs:sodd} and \ref{subs:excodd} respectively.

\subsection{Integrable quasi-cyclic elements in \texorpdfstring{$\mf{so}_N$}{soNo} for nilpotent elements of odd depth}\label{subs:sodd}

Let $\mf g=\mf{so}_N$ and let $f\in\mf g$ be a nilpotent element associated to the partition $\partition p$
as in \eqref{eq:partition}. Assume that it has odd depth $d=2D-1$, $D=p_1-1$.
By Lemma \ref{20201219:lem1}, this happens when $p_1$ is odd, $r_1=1$ and $p_2=p_1-1$.

We realize the Lie algebra $\mf g$ ad in \eqref{20201212:eq1}, with $\eta=1$, and we let
$\{F_{\alpha,\beta},\alpha,\beta\in I\}$ be the set of generators of $\mf g$ defined in \eqref{20200819:eq4}. The first result describes $\mf z(\mf g_{\ge2})$, the centralizer of $\mf g_{\ge2}$ in $\mf g$.
\begin{proposition}\label{centr-so2}
We have that $\mf z(\mf g_{\ge2})=\mf g_{d-1}\oplus\mf g_d$.
\end{proposition}
\begin{proof}
The proof is similar to the proof of Proposition \ref{20201221:prop1}. For completeness, we replicate the argument.
Clearly, $\mf z(\mf g_{\ge2})\subset\mf g_{\ge0}$. By degree considerations, $\mf z(\mf g_{\ge2})\supset\mf g_{d-1}\oplus\mf g_d$.
On the other hand, let $1\in I$ be the label of the rightmost box of the pyramid associated to $\partition p$ (note that $x_{1}=D$) and $p_1\in I$ be the label of the leftmost box (note that $x_{p_1}=-D$).
Let also $\tilde\beta\in I$ be such that $x_{\tilde\beta}\le D-2$ (the box $\tilde\beta$ is completely at the left of the box $1$).
Then $F_{1\tilde\beta}\in\mf g_{\ge2}$. Hence, letting $E$ as in equation \eqref{eq:E-sp},
using the commutation relations \eqref{comm:BC}, the second equation in \eqref{eq:E-sp} and \eqref{20200819:eq4b}, we have that
\begin{equation}\label{eq1-so}
0=[F_{1\tilde\beta},E]=\sum_{x_{\alpha}-x_{\beta}=k}c_{\alpha\beta}[F_{1\tilde\beta},F_{\alpha\beta}]
=\sum_{x_{\tilde\beta}-x_{\beta}=k}2c_{\tilde\beta\beta}F_{1\beta}
-\sum_{x_{\alpha}-D=k}2c_{\alpha1}F_{\alpha\tilde\beta}\,.
\end{equation}
If $k\ge1$, then the condition $x_{\alpha}-D=k\ge1$ implies that $x_{\alpha}\ge1+D$, which is empty.
Hence, $c_{\alpha\beta}=0$ if $x_{\alpha}\le D-2$ and $\beta\neq p_1$ (since
$F_{1,p_1}=0$ by \eqref{20200819:eq4b} and \eqref{20201001:eq4}).
Using the second equation in \eqref{eq:E-sp}, we have also that $c_{\alpha\beta}=0$ if $\alpha\neq 1$ and $x_{\beta}\ge-D+2$.
Hence, we can write
$$
E=\sum_{\substack{x_\alpha=D-1\\x_\beta=-D+1}}c_{\alpha,\beta}F_{\alpha\beta}+\sum_{x_\beta=D-k}c_{1\beta}F_{1\beta}\,.
$$
Note that the first sum lies in $\mf{g}_{d-1}$. Let then assume that
$E=\sum_{x_\beta=D-k}c_{1\beta}F_{1\beta}$ and that $\tilde\alpha,\tilde\beta\in I$ are such that $x_{\tilde\alpha}-x_{\tilde\beta}\ge2$. Then
\begin{equation}\label{20201116:eq1}
0=[E,F_{\tilde\alpha\tilde\beta}]=2\delta_{x_{\tilde\alpha},D-k}c_{1,\tilde\alpha}F_{1,\tilde\beta}\,.
\end{equation}
If $x_{\tilde\alpha}=D-k$, then $-D\le x_{\tilde\beta}\le D-2-k$.
Hence, equation \eqref{20201116:eq1}
implies that $c_{1\beta}=0$ for $x_\beta>-D+2$, thus showing that $E\in\mf{g}_d\oplus\mf{g}_{d-1}$.

If $k=0$, a similar argument to the one used in the proof of Lemma \ref{lemma1} shows that $E=0$.
\end{proof}

By the discussion at the beginning of this section and Proposition \ref{centr-so2}, if $(f_1,f_2,E)$ is an integrable triple, then $E\in\mf g_{d-1}$.

Recall that $\dim V[D]=r_1=1$, and let $v_+\in V[D]$ be such that $\langle f^{D}v_+|v_+\rangle=1$.
We consider the basis $\{f^kv_+\}_{k=0}^{D}$ of $V_1$ (cf. Section \ref{sec:5.5.2}).
Let also $v_-=f^Dv_+$.

\begin{lemma}\label{20201116:lem2}
There is a bijective correspondence between the triples
$(\lambda,a,b)$, where
$\lambda\in\mb F$, $a\in V_{+,3}$ and $b\in\End(V_{+,2})$ selfadjoint with respect to the bilinear form
$\beta_2$ on $V_{+,2}$ defined by \eqref{20200819:bilinear}, and the elements $E\in\mf g_{d-1}$, given by
\begin{equation}\label{eq:E-so-dispari}
(\lambda,a,b)\mapsto
E=(a+\lambda fv_+)\phi(v_+)-v_+(\phi(a)+\lambda\phi(fv_+))+b(f\transpose)^{D-1}
\,,
\end{equation}
where $\phi(v)\in V^*$ is given by $\phi(v)(u)=\langle v|u\rangle$ (cf. Section \ref{sec:4.5.1}).
\end{lemma}
\begin{proof}
An element $E\in\mf g_{d-1}$ can be uniquely written as
\begin{equation}\label{20201221:eq1}
E=A-A^\dagger+B\in\mf g_{d-1}\,,
\end{equation}
where $A\in\Hom(V[-D],V[D-2])$, $B=-B^\dagger\in\Hom(V[-D+1],V[D-1])$.
Recall that $V[-D]=\mb Fv_-$ and that we have the decomposition
$V[D-2]=\mb F fv_+\oplus V_{+,3}$. Hence, we can uniquely write $Av_-=\lambda fv_++a$,
for $\lambda\in\mb F$ and $a\in V_{+,3}$.
Hence,
\begin{equation}\label{20201117:eq1}
A=(a+\lambda fv_+) \phi(v_+)
\end{equation}
and
\begin{equation}\label{20201116:eq5}
A^\dagger=v_+(\phi(a)+\lambda \phi(f v_+))\,.
\end{equation}
Furthermore, let us define $b=Bf^{D-1}\in\End(V_{+,2})$. Since $D_2=p_2-1=p_1-2=D-1$ and $p_1$ is odd,
then the bilinear form \eqref{20200819:bilinear} $\beta_2$ on $V_{+,2}$ is skewsymmetric. Note that,
for $v,w\in V_{+,2}$ we have
\begin{align*}
&\beta_2(bv,w)=\langle bv|f^{D-1}w\rangle
=\langle v|(-f)^{D-1}(-B)f^{D-1}w\rangle
\\
&=\langle B f^{D-1}v|f^{D-1}w\rangle
=\langle v|f^{D-1}bw\rangle=\beta_2(v,bw)\,.
\end{align*}
Hence, $b$ is self-adjoint with respect to $\beta_2$. A similar computation shows that, if $b$
is selfadjoint with respect to $\beta_2$, then
$B=b(f\transpose)^{D-1}\in\mf g_{d-1}$.
\end{proof}
\begin{proposition}\label{coiso-so2}
Let $E\in\mf g_{d-1}$ be as in \eqref{eq:E-so-dispari}. The subspace $\mf g_1^E$ is
coisotropic with respect to the bilinear form \eqref{eq:skewform} if and only if
\begin{equation}\label{20201116:eq2}
b^2-2\lambda b+\alpha=0\,,
\end{equation}
where
\begin{equation}\label{eq:alpha}
\alpha=\langle a|f^{D-2}a\rangle\,.
\end{equation}
\end{proposition}
\begin{proof}
Let $E=A-A^\dagger+B\in\mf g_{d-1}$ as in \eqref{20201221:eq1} and let $U=X-X^\dagger\in\mf g_{d}$, where $X\in\Hom(V[D],V[-D+1])$.
Note that the two equations in \eqref{20201029:eq1} are equivalent in the present setting.
Hence, by Proposition \ref{prop:linalg} and Lemma \ref{prop:20201029b} we have that
$\mf g_1^E$ is coisotropic if and only if
\begin{equation}\label{20201116:eq3}
Bf^{D-1}BX(f\transpose)^D=BX(f\transpose\id_{fV}A-A^{\dagger}\id_{f\transpose V}f\transpose)-(f\transpose)^{D-1}XA^{\dagger}\id_{V_-}f^{D-2}\id_{V_+}A\,,
\end{equation}
for every $X$. From equations \eqref{20201117:eq1} and \eqref{20201116:eq5} we immediately have
\begin{equation}\label{20201116:eq4}
\begin{split}
&\id_{V_+}A=a\phi(v_+)\,,\quad \id_{fV}A=\lambda fv_+\phi(v_+)\,,
\\
&
A^{\dagger}\id_{V_-}=v_+\phi(a)\,,\quad A^\dagger\id_{f\transpose V}=\lambda v_+\phi(fv_+)\,.
\end{split}
\end{equation}
Using equation \eqref{20201116:eq4} we rewrite equation \eqref{20201116:eq3} as
\begin{equation}\label{20201116:eq7}
b^2(f\transpose)^{D-1}X
=(2\lambda b-\alpha)(f\transpose)^{D-1}X(v_+)\phi(v_+)\,.
\end{equation}
Applying both sides of \eqref{20201116:eq7}
to $v_+$ we get
\begin{equation}\label{20201116:eq8}
b^2(f\transpose)^{D-1}X(v_+)=(2\lambda b-\alpha)(f\transpose)^{D-1}X(v_+)
\,.
\end{equation}
Equation \eqref{20201116:eq2} follows from \eqref{20201116:eq8} since $X$ is arbitrary.
\end{proof}
In order to prove the main result of this section we need the following.
\begin{lemma}\label{20201116:lem3}
If $a\neq 0$ and $\alpha=0$, then $(f+E)_n\not\in\mf g_{-2}$.
\end{lemma}
\begin{proof}
Note that, since $a\neq0$, $V_{+,3}=V[D-2]\cap V_+\neq 0$.
From equation \eqref{eq:E-so-dispari}, we have
$$
(f+E)^{D-1}a=Ef^{D-2}a=-\langle a|f^{D-2}a\rangle v_+=-\alpha v_+\,.
$$
Hence, if $\alpha=0$, we have $(f+E)^{D-1}(a)=0$.
Recall that the bilinear form \eqref{20200819:eq3} is non-degenerate. Hence, there exists $w\in V_{-,3}$ such that
$\langle a|w\rangle=1$. Let $u=2\lambda w+f^{D-1}v_+$. By equation \eqref{eq:E-so-dispari} we have
$$
(f+E)(u)=v_--\lambda v_+\,,\qquad (f+E)^2(u)=a\,.
$$
Thus we get that $u$ lies in the generalized eigenspace of $f+E$ of eigenvalue $0$ and $(f+E)_s(u)=0$. This implies that
$(f+E)_n(u)=(f+E)(u)=v_--\lambda v_+$.
On the other hand, $u\in V[-D+2]$, while $v_-\in V[-D]$ and $v_+\in V[D]$.
As a consequence, $(f+E)_n\not\in\mf g_{-2}$.
\end{proof}
\begin{lemma}\label{20201116:lem4}
Let $X=(f+E)\id_{V_1\oplus V_3}$ and let us assume that $\alpha\neq0$. Then $X_n\in\mf g_{-2}$.
\end{lemma}
\begin{proof}
Since $\alpha\neq 0$, we have $a\neq 0$ and the direct sum decomposition
$$V_{+,3}=V[D-2]\cap V_+=\mb Fa\oplus\ker\phi(f^{D-2}a)\id_{V_{+,3}}\,.
$$
Let $V_a=\oplus_{k=0}^{D-2}\mb Ff^ka$ and $U=\oplus_{k=0}^{D-2}\mb Ff^k\ker\phi(a)\id_{V_{+,3}}$. Then, we have the
direct sum decomposition $V_1\oplus V_3=V_1\oplus V_a\oplus U$. Note that $X(V_1\oplus V_a)\subset V_1\oplus V_a$,
$X(U)\subset U$ and $X\id_U$ is nilpotent. Fixing, for example, the basis $\{f^kv_+,f^ha|0\le k\le D,0\le h\le D-2\}$
of $V_1\oplus V_a$ it is straightforward to
check that the characteristic polynomial of $X\id_{V_1\oplus V_a}$ is $x^{2D}-2\lambda x^D+\alpha$. Hence, if $\alpha\neq \lambda^2$, then $X\id_{V_1\oplus V_a}$ is semisimple.
In this case the Jordan decomposition of $X=X_s+X_n$ has
$X_s=X\id_{V_1\oplus V_a}$ and $X_n=X\id_U\in\mf g_{-2}$, as claimed. If instead $\alpha=\lambda^2$, then the nilpotent part of $X\id_{V_1\oplus V_a}$
is the following element of $\mf g_{-2}$
$$
\frac{1}{D}(f\id_{V_1}-f\id_{V_a})
+\frac{1}{D\lambda}\sum_{k=0}^{D-2}\left((-f)^ka\phi((f\transpose)^kv_-)
-\frac{\lambda^2}{\alpha} f^{k+2}v_+\phi(f^{D-2-k}a)\right)\,.
$$
(The proof of this fact is straightforward and is omitted).
\end{proof}
The following main result of this subsection characterizes integrable quasi-cyclic elements for $\mf{so}_N$ associated to nilpotent elements of odd depth.
\begin{theorem}\label{integrable-so2}
Let $\mf g=\mf{so}_N$ and let $f\in\mf g$ be a nilpotent element of odd depth $d=2D-1$, where $D=p_1-1$. Let $E\in\mf g_{d-1}$ be decomposed as in \eqref{eq:E-so-dispari} and let $\alpha$ be as in \eqref{eq:alpha}. Then $E$ is integrable for $f$ if and only if the following two conditions hold:
(i) $b$ is semisimple with minimal polynomial dividing $x^2-2\lambda x+\alpha$,
(ii) if $a=0$ then $\lambda\ne0$, while if $a\neq 0$ then $\alpha\ne0$.
\end{theorem}
\begin{proof}
First note that $(f+E)\id_{V_2}=f\id_{V_2}+b(f\transpose)^{D-1}$. By Propositions \ref{coiso-so2} and \ref{thm:glintiffss}, $(f+E)\id_{V_2}$ is integrable if and only if
$b$ is semisimple and its minimal polynomial divides $x^2-2\lambda x+\alpha$, i.e. condition (i) holds.
Furthermore, $(f+E)\id_{V_{\ge4}}=f\id_{V_{\ge4}}$ and $(f+E)(V_1\oplus V_3)\subset V_1\oplus V_3$. Since
$f\id_{V_{\ge4}}$ is nilpotent and commutes with both $(f+E)\id_{V_2}$ and $(f+E)\id_{V_1\oplus V_3}$, we are left to understand when
$X=(f+E)\id_{V_1\oplus V_3}$ is integrable. By Lemma, \ref{20201116:lem3} $X$ is not integrable if $a\neq0$ and $\alpha=0$.
If $a=\alpha=0$, then $\lambda\neq0$, otherwise $E=0$ (indeed if $\lambda=\alpha=0$, then $b=0$). When $\lambda\neq 0$, then $X$ is semisimple since its minimal polynomial is $x(x^D-2\lambda)$ which has
distinct roots (see Example 2.12 in \cite{DSJKV20}). If $\alpha\neq0$, then the result follows from Proposition \ref{coiso-so2} and Lemma \ref{20201116:lem4}.
\end{proof}

\begin{remark}
The integrable element $E\in\mf g_{d-1}$ constructed in Example 2.12 in \cite{DSJKV20} corresponds to the choice
$\lambda=1$ and $a=b=0$ in Theorem \ref{integrable-so2}.
\end{remark}

Let us reformulate the results obtained in this subsection in terms of polar linear groups.
\begin{theorem}\label{thm:polarsodd}
Let $\mf g=\mf{so}_N$ and let $f\in\mf g$ be a nilpotent element of odd depth $d$, so that the corresponding partition has the form $\partition p=(p_1,(p_1-1)^{(r_2)},(p_1-2)^{(r_3)},...)$, and $d=2p_1-3$. Then
\begin{enumerate}[(a)]
\item\label{item:l2stsp+stso+1}
The linear group $Z(\mf s)|\mf g_{d-1}$ (rather its unity component) is isomorphic to the direct sum of polar linear groups
\begin{equation}\label{eq:l2stsp+stso+1}
    \ext^2\st(Sp_{r_2})\oplus\st(SO_{r_3})\oplus\id,
\end{equation}
hence is polar. This linear group leaves invariant the non-degenerate symmetric bilinear form, defined by
\begin{equation}
    (a,b)=\left((\ad f)^{\frac{d-1}2}a\mid b\right),
\end{equation}
where $(\cdot\mid\cdot)$ is the trace form on $\mf g$. Consequently, we may identify $Z(\mf s)|\mf g_{d-1}$ with the space of triples $(b,a,\lambda)$, where $b$ is a selfadjoint operator on $\mb F^{r_2}$ with a skewsymmetric non-degenerate bilinear form, $a\in\mb F^{r_3}$, and $\lambda\in\mb F$.
\item\label{item:ssminpoletc} A triple $E=(b,a,\lambda)\in\mf g_{d-1}$ is an integrable element for $f$ if and only if the following holds:
\begin{enumerate}[(i)]
\item $b$ is semisimple and its minimal polynomial divides $x^2-2\lambda x+(a,a)$,
\item $\lambda\ne0$ if $a\ne0$; $(a,a)\ne0$ if $a\ne0$.
\end{enumerate}
\end{enumerate}
\end{theorem}
\begin{proof}
In order to compute $Z(\mf s)|\mf g_{d-1}$, consider the subalgebra $\mf g_{\mb Z}=\bigoplus_{j\in\mb Z}\mf g_{2j}$ of $\mf g$, which contains $\mf s$ and $\mf g_{d-1}$. It is a direct sum of orthogonal subalgebras $\mf g'$ and $\mf g''$, so that $f=f'+f''$, $f'\in\mf g'$, $f''\in\mf g''$, where $f'$ (resp. $f''$) corresponds to the odd part  $(p_1,(p_1-2)^{(r_3)},...)$ (resp. even part $((p_1-1)^{(r_2)},...)$) of the partition $\partition p$, which give contributions $\st(SO_{r_3})\oplus\id$ and $\ext^2\st(Sp_{r_2})$, respectively, to \eqref{eq:l2stsp+stso+1} (cf. \cite{EKV13}). This proves \eqref{item:l2stsp+stso+1}. Claim \eqref{item:ssminpoletc} follows from Theorem \ref{thm:polarsodd} and claim \eqref{item:l2stsp+stso+1}.
\end{proof}
\begin{remark}
Another simple example of an integrable $E$ is $(\lambda I_{r_2},a,\lambda^2)$, where $\lambda\ne0$ and $(a,a)=\lambda^2$.
\end{remark}

\subsection{Integrable quasi-cyclic elements in exceptional Lie algebras for nilpotent elements of odd depth}\label{subs:excodd}

We begin by explaining details of calculations that were used to produce Table 1 in \cite{DSJKV20} (see Table \ref{tab:odd} below).

\begin{table}[H]
\resizebox{.96\textwidth}{!}{
\begin{tabular}{c|c|c|c|c|c|c}
$\mf g$&nilpotent $f$&$d$&$\mf z(\mf s)|{\mf g}_{d-1}$&rank&quasi-cyclic $f+E$, $E\in\mf g_{d-1}$&quasi type\bigstrut[b]\\
\hline
$E_6$&$3A_1$\hfill\EDynkin{0,0,0,1,0,0}{1}&$3$&$\ad(\mf{sl}_3)\oplus\id$&$3$&$\exists$ semisimple&semisimple\bigstrut[t]\\
&$2A_2+A_1$\hfill\EDynkin{0,1,0,1,0,1}{1}&$5$&$\st(\mf{so}_3)\oplus{\id}$&$2$&$\exists$ integrable&mixed\\
\hline
$E_7$&$3A_1'$\hfill\EDynkin{0,0,1,0,0,0,0}{1}&$3$&$\ext^2\st(\mf{sp}_6)$&$3$&$\exists$ semisimple&semisimple\bigstrut[t]\\
&$4A_1$\hfill\EDynkin{1,0,0,0,0,0,1}{1}&$3$&$\ext^2\st(\mf{sp}_6)\oplus\id$&$4$&$\exists$ semisimple&semisimple\\
&$2A_2+A_1$\hfill\EDynkin{0,0,1,0,0,1,0}{1}&$5$&$\st(\mf{so}_3)\oplus\st(\mf{so}_3)$&$2$&$\exists$ integrable&mixed\\
\hline
$E_8$&$3A_1$\hfill\EDynkin{0,0,0,0,0,0,1,0}{1}&$3$&$\st(F_4)\oplus\id$&$3$&$\exists$ semisimple&semisimple\bigstrut[t]\\
&$4A_1$\hfill\EDynkin{1,0,0,0,0,0,0,0}{1}&$3$&$\ext^2\st(\mf{sp}_8)$&$4$&nilpotent only&semisimple\\
&$2A_2+A_1$\hfill\EDynkin{0,1,0,0,0,0,1,0}{1}&$5$&$\st(G_2)\oplus\st(\mf{so}_3)$&$2$&$\exists$ integrable&mixed\\
&$2A_2+2A_1$\hfill\EDynkin{0,0,0,0,1,0,0,0}{1}&$5$&$\ad(\mf{so}_5)$&$2$&nilpotent only&mixed\\
&$2A_3$\hfill\EDynkin{0,1,0,0,1,0,0,0}{1}&$7$&$\st(\mf{so}_5)\oplus\id$&$2$&nilpotent only&semisimple\\
&$A_4+A_3$\hfill\EDynkin{0,0,0,1,0,0,1,0}{1}&$9$&$\st(\mf{so}_3)$&$1$&nilpotent only&mixed\\
&$A_7$\hfill\EDynkin{0,1,0,1,0,1,1,0}{1}&$15$&$\id$&$1$&none&semisimple\\
\hline
$F_4$&$A_1+\widetilde A_1$\hfill\FDynkin{0,0,0,1}&$3$&$\sym^2\st(\mf{so}_3)$&$3$&$\exists$ semisimple&semisimple\bigstrut[t]\\[1em]
&$\widetilde A_2+ A_1$\hfill\FDynkin{1,0,0,1}&$5$&$\st(\mf{so}_3)$&$1$&nilpotent only&mixed\bigstrut[b]\\
\hline
$G_2$&$\widetilde A_1$\hfill\GDynkin{1,0}&$3$&$\id$&$1$&none&semisimple\bigstrut[t]
\end{tabular}
}
\caption{Quasi-cyclic elements attached to nilpotent elements of odd depth in exceptional simple Lie algebras}\label{tab:odd}
\end{table}

As in \cite{DSJKV20}, $\st(\mf a)$ denotes the standard representation of the Lie algebra $\mf a$
(which is the 26-dimensional for $\mf a=\rF_4$). In this Table $\operatorname{rank}=\dim\mf g_{d-1}\git Z(\mf s)$.
We call $f$ to be of semisimple (resp. mixed) quasi type
if there exist $E\in\mf g_{d-1}$, such that $f+E$ is semisimple (resp. not nilpotent).

Notation for the nilpotent elements describes them as principal nilpotent elements in the corresponding Levi subalgebras.
For example, in $\rE_8$ there is a unique, up to conjugacy, Levi subalgebra of type $2\rA_2+\rA_1$;
then $f$ is the sum of the corresponding negative simple root vectors.
In $\rE_7$ there are two, up to conjugacy, Levi subalgebras of type $3\rA_1$;
$3\rA_1'$ stands for the one whose principal nilpotent has odd depth in $\rE_7$ (the principal nilpotent for the other one has even depth).
Finally, in $\rF_4$ and $\rG_2$ tilde means that we take the negative short simple root vector.

Except for the nilpotent with label $A_1+\widetilde A_1$ in $\rF_4$, Cartan subspaces in $\mf g_1$ with respect to the $\mf z(\mf s)$-module structure
are given by the zero weight spaces of these modules.

\subsubsection{Nilpotent elements with label $3\rA_1$ in $\rE_6$, $\rE_8$ and $3\rA_1'$ in $\rE_7$}\label{subsubs:3A1}

All these conjugacy classes have representatives of the form $f=e_{-\alpha_1}+e_{-\alpha_2}+e_{-\alpha_3}$,
sums of three pairwise commuting negative simple root vectors.
For $\rE_6$ and $\rE_8$ these can be arbitrary three commuting root vectors, while in $\rE_7$ arbitrary under the restriction that $f$ has odd depth.
One checks that in this case the subspace $C\subseteq\mf g_2$ spanned by $e_{\alpha_1},e_{\alpha_2},e_{\alpha_3}$ is a Cartan subspace with respect to the action of $Z(\mf s)$.
The coisotropy condition on $E=x_1e_{\alpha_1}+x_2e_{\alpha_2}+x_3e_{\alpha_3}$ is
\[
x_1^2+x_2^2+x_3^2-2x_1x_2-2x_1x_3-2x_2x_3=0.
\]
The subalgebra generated by $e_{\pm\alpha_i}$, $i=1,2,3$, is the direct sum of three copies of $\mf{sl}_2$, and it is straightforward to check that $f+E$ is semisimple when $x_1,x_2,x_3$ are arbitrary nonzero numbers satisfying the coisotropy condition. When one of them is zero, say, $x_1=0$, then the coisotropy condition forces $x_2=x_3=x\ne0$, in which case the Jordan decomposition of $f+x(e_{\alpha_2}+e_{\alpha_3})$ is $(e_{-\alpha_2}+e_{-\alpha_3}+x(e_{\alpha_2}+e_{\alpha_3}))+e_{-\alpha_1}$ and we get an integrable triple $(f_1,f_2,E)$, where $f_1$ has label $2\rA_1$ and $f_2$ has label $\rA_1$.

\subsubsection{Nilpotent elements with label $2\rA_2+\rA_1$ in $\rE_6$, $\rE_7$ and $\rE_8$}\label{subsubs:2A2+A1}

This case is described in \cite[Example 2.14]{DSJKV20}. Here $d=5$. We can take
\[
f=e_{-\alpha_1}+e_{-\alpha_2}+e_{-\beta_1}+e_{-\beta_2}+e_{-\gamma},
\]
where $\alpha_1+\alpha_2$ and $\beta_1+\beta_2$ are roots, while no other pairwise sum of the $\alpha_i$, $\beta_j$ and $\gamma$ is a root. One then checks that the subspace of $\mf g_{d-1}$, spanned by $e_{\alpha_1+\alpha_2}$ and $e_{\beta_1+\beta_2}$, is a Cartan subspace. The coisotropy condition for $E=xe_{\alpha_1+\alpha_2}+ye_{\beta_1+\beta_2}$ is then $x=y$, and for $E=x(e_{\alpha_1+\alpha_2}+e_{\beta_1+\beta_2})$ the Jordan decomposition of the quasi-cyclic element $f+E$ is $(e_{-\alpha_1}+e_{-\alpha_2}+e_{-\beta_1}+e_{-\beta_2}+x(e_{\alpha_1+\alpha_2}+e_{\beta_1+\beta_2}))+e_{-\gamma}$. This is straightforward to check after restricting considerations to the subalgebra of type $\rA_2+\rA_2+\rA_1$ containing both $f$ and $E$. We thus obtain integrable triple $(f_1,f_2,E)$, where $f_1$ has label $2\rA_2$ and $f_2$ has label $\rA_1$.

Clearly in such way we obtain all possible integrable triples that might occur in this case: if all three of the $x_1$, $x_2$, $x_3$ are nonzero, we obtain the integrable triple $(f,0,E)$, while if some of them are zero, we necessarily get the integrable triples $(f_1,f_2,E)$ as above.

\subsubsection{Nilpotent elements with label $4\rA_1$ in $\rE_7$ and $\rE_8$}\label{subsubs:4A1}
These conjugacy classes have representatives of the form $f=e_{-\alpha_1}+e_{-\alpha_2}+e_{-\alpha_3}+e_{-\alpha_4}$, sums of four pairwise commuting root vectors.
One checks that in this case the subspace $C\subseteq\mf g_2$ spanned by $e_{\alpha_1},e_{\alpha_2},e_{\alpha_3},e_{\alpha_4}$ is a Cartan subspace with respect to the action of $Z(\mf s)$.
The coisotropy condition on $E=x_1e_{\alpha_1}+x_2e_{\alpha_2}+x_3e_{\alpha_3}+x_4e_{\alpha_4}$ is
\[
x_i^2+x_j^2+x_k^2-2x_ix_j-2x_ix_k-2x_jx_k=0,
\]
where $\{i,j,k\}$ is any three-element subset of $\{1,2,3,4\}$ for $\rE_8$, while for $\rE_7$ it can be any three-element subset except one of them.

For $\rE_8$ the resulting system of quadratic equations has only zero solution, which means, by Lemma \ref{lem:nononnil}, that in this case there are no non-nilpotent quasi-cyclic elements.

For $\rE_7$, choose some numbering, say, such, that the equations correspond to $\{i,j,k\}$ equal to $\{1,2,3\}$, $\{1,2,4\}$ or $\{1,3,4\}$; then, the system has five solutions, corresponding to
\begin{align*}
E&=x(e_{\alpha_2}+e_{\alpha_3}+e_{\alpha_4}),\\
E&=x(4e_{\alpha_1}+e_{\alpha_2}+e_{\alpha_3}+e_{\alpha_4}),\\
E&=x(4e_{\alpha_1}+9e_{\alpha_2}+e_{\alpha_3}+e_{\alpha_4}),\\
E&=x(4e_{\alpha_1}+e_{\alpha_2}+9e_{\alpha_3}+e_{\alpha_4}),\\
E&=x(4e_{\alpha_1}+e_{\alpha_2}+e_{\alpha_3}+9e_{\alpha_4}).
\end{align*}
The subalgebra generated by $e_{\pm\alpha_i}$, $i=1,2,3,4$, is a direct sum of 4 copies of $\mf{sl}_2$, which easily implies that the last four solutions give semisimple quasi-cyclic elements, while the first solution gives a quasi-cyclic element with the Jordan decomposition $(e_{-\alpha_2}+e_{-\alpha_3}+e_{-\alpha_4}+x(e_{\alpha_2}+e_{\alpha_3}+e_{\alpha_4}))+e_{-\alpha_1}$, which gives an integrable triple $(f_1,f_2,E)$ where $f_1$ has label $3\rA_1$ and $f_2$ has label $\rA_1$.

This exhausts all possible integrable triples in this case, since the coisotropy equations do not have any other solutions.

\subsubsection{Nilpotent element with label $2\rA_2+2\rA_1$ in $\rE_8$}

We can take
\[
f=e_{-\alpha_1}+e_{-\alpha_2}+e_{-\beta_1}+e_{-\beta_2}+e_{-\gamma}+e_{-\delta},
\]
where $\alpha_1+\alpha_2$ and $\beta_1+\beta_2$ are roots, while no other pairwise sum of the $\alpha_i$, $\beta_j$, $\gamma$ and $\delta$ is a root. A Cartan subspace in $\mf g_{d-1}$ is spanned by $e_{\alpha_1+\alpha_2}$ and $e_{\beta_1+\beta_2}$. We then find that the coisotropy condition on $E=xe_{\alpha_1+\alpha_2}+ye_{\beta_1+\beta_2}$ are $x-y=0$ and $x+y=0$, hence, by Lemma \ref{lem:nononnil}, all quasi-cyclic elements are nilpotent.

\subsubsection{Nilpotent element with label $2\rA_3$ in $\rE_8$}

We can take
\[
f=e_{-\alpha_1}+e_{-\alpha_2}+e_{-\alpha_3}+e_{-\beta_1}+e_{-\beta_2}+e_{-\beta_3},
\]
where $\alpha_1+\alpha_2$, $\alpha_2+\alpha_3$, $\alpha_1+\alpha_2+\alpha_3$, $\beta_1+\beta_2$, $\beta_2+\beta_3$ and $\beta_1+\beta_2+\beta_3$ are roots, while no other sum of the $\alpha_i$ and $\beta_j$ is a root. A Cartan subspace in $\mf g_{d-1}$ is spanned by $e_{\alpha_1+\alpha_2+\alpha_3}$ and $e_{\beta_1+\beta_2+\beta_3}$. We then find that the coisotropy condition on $E=xe_{\alpha_1+\alpha_2+\alpha_3}+ye_{\beta_1+\beta_2+\beta_3}$ are $x(x+4y)=0$ and $y(4x+y)=0$, hence, by Lemma \ref{lem:nononnil}, all quasi-cyclic elements are nilpotent.

\subsubsection{Nilpotent element with label $\rA_4+\rA_3$ in $\rE_8$}

We can take
\[
f=e_{-\alpha_1}+e_{-\alpha_2}+e_{-\alpha_3}+e_{-\alpha_4}+e_{-\beta_1}+e_{-\beta_2}+e_{-\beta_3},
\]
where $\alpha_1+\alpha_2$, $\alpha_2+\alpha_3$, $\alpha_3+\alpha_4$, $\alpha_1+\alpha_2+\alpha_3$, $\alpha_2+\alpha_3+\alpha_4$, $\alpha_1+\alpha_2+\alpha_3+\alpha_4$, $\beta_1+\beta_2$, $\beta_2+\beta_3$ and $\beta_1+\beta_2+\beta_3$ are roots, while no other sum of the $\alpha_i$ and $\beta_j$ is a root. A Cartan subspace in $\mf g_{d-1}$ is spanned by $e_{\alpha_1+\alpha_2+\alpha_3+\alpha_4}$. We then find that the coisotropy condition on $E=xe_{\alpha_1+\alpha_2+\alpha_3+\alpha_4}$ is $x=0$, hence, by Lemma \ref{lem:nononnil}, all quasi-cyclic elements are nilpotent.

\subsubsection{Nilpotent element with label $\rA_7$ in $\rE_8$}

We can take
\[
f=e_{-\alpha_1}+e_{-\alpha_2}+e_{-\alpha_3}+e_{-\alpha_4}+e_{-\alpha_5}+e_{-\alpha_6}+e_{-\alpha_7},
\]
where $\alpha_1$, ..., $\alpha_7$ form simple roots for a root subsystem of type $\rA_7$. Here $\mf g_{d-1}$ is one-dimensional, spanned by $e_{\alpha_1+...+\alpha_7}$,
and for $E=xe_{\alpha_1+...+\alpha_7}$ the coisotropy condition fails unless $x=0$, so that there are no quasi-cyclic elements whatsoever.

\subsubsection{Nilpotent element with label $\rA_1+\widetilde\rA_1$ in $\rF_4$}\label{subsubs:A1tildeA1}
Take $f=f_{1220}+f_{1232}$.

Here $\mf z(\mf s)|\mf g_{d-1}$ is the direct sum of a $5$-dimensional irreducible and $1$-dimensional trivial representation of $\mf{sl}_2$.
The subspace of $\mf g_{d-1}$ spanned by $E_0=e_{1220}$, $E_1=e_{1222}+e_{1232}+e_{1242}$ and $E_2=e_{1222}-e_{1232}+e_{1242}$ is a Cartan subspace.
Coisotropy condition on $E=x_0E_0+x_1E_1+x_2E_2$ is
\[
x_0^2+4x_1^2+4x_2^2-4x_0x_1+4x_0x_2+8x_1x_2=0.
\]
Subalgebra generated by $f$ and the Cartan subspace is a direct sum of three copies of $\mf{sl}_2$, and the matrix of $f+E$ in the standard representation of this subalgebra is
\[
\begin{pmatrix}
0&x_0&0&0&0&0\\
1&0&0&0&0&0\\
0&0&0&x_1&0&0\\
0&0&2&0&0&0\\
0&0&0&0&0&x_2\\
0&0&0&0&-2&0
\end{pmatrix}
\]
It follows that the quasi-cyclic element $f+E$ with $E$ as above satisfying the coisotropy condition, is semisimple except for the cases
\begin{align*}
&x_0=0,x_1=-x_2;\\
&x_1=0,x_0=-2x_2;\\
&x_2=0,x_0=2x_1.
\end{align*}
The Jordan decomposition of $f+E$ in these cases is, respectively,
\begin{align*}
(f_{1232}+E)&+f_{1220};\\
(f_{1220}+\frac12(-f_{1222}+f_{1232}-f_{1242})+E)&+\frac12(f_{1222}+f_{1232}+f_{1242});\\
(f_{1220}+\frac12(f_{1222}+f_{1232}+f_{1242})+E)&+\frac12(-f_{1222}+f_{1232}-f_{1242}).
\end{align*}
In all three cases we thus get an integrable triple $(f_1,f_2,E)$, where $f_1$ has label $\widetilde\rA_1$ and $f_2$ has label $\rA_1$.

These three cases, together with the case when the quasi-cyclic element is semisimple, give all possible integrable triples for this case.

\subsubsection{Nilpotent element with label $\widetilde\rA_2+\rA_1$ in $\rF_4$}

A Cartan subspace is given by the zero weight space of the adjoint representation of $\mf{sl}_2$,
and none of its nonzero vectors satisfies the coisotropy condition.
It follows from Lemma \ref{lem:nononnil} that all quasi-cyclic elements are nilpotent.

\subsubsection{Nilpotent element with label $\widetilde\rA_1$ in $\rG_2$}\label{subsubs:tildeA1}

This nilpotent does not produce any quasi-cyclic elements, as Example 2.8 in \cite{DSJKV20} shows.
Namely, the depth is $3$, and $\mf g_2$ is $1$-dimensional, spanned by $e_{12}$;
its centralizer has zero intersection with $\mf g_1$, and the zero subspace is not coisotropic.

\addtocontents{toc}{\SkipTocEntry}
\subsection*{Conclusion} Due to Theorem \ref{thm:quasintiffss} \eqref{thm:quasintiffssodd},
Subsections \ref{subsubs:3A1}, \ref{subsubs:2A2+A1}, \ref{subsubs:4A1}, and \ref{subsubs:A1tildeA1}
describe all integrable quasi-cyclic elements $f+E$ for nilpotent elements $f$ of odd depth,
for all exceptional simple Lie algebras, up to conjugation by $Z(\mf s)$.
In particular, an integrable quasi-cyclic element exists for such $f$, except for seven cases, described in the Introduction.

\end{document}